\documentclass[11pt,a4paper]{article}

\pdfoutput=1

\def\withanimation{1}

\usepackage[margin=1in,includefoot]{geometry}
\usepackage{tikz}
\usepackage[
  bookmarks=true,
  bookmarksnumbered=true,
  bookmarksopen=true,
  pdfborder={0 0 0},
  breaklinks=true,
  colorlinks=true,
  linkcolor=black,
  citecolor=black,
  filecolor=black,
  urlcolor=black,
]{hyperref}
\usepackage{dsfont}
\usepackage{ifthen}
\usepackage{xifthen}
\usepackage[round]{natbib}
\usepackage{amssymb,amsmath,amsthm}
\usepackage{enumitem}
\usepackage{algorithm}
\usepackage{algpseudocode}
\usepackage{mathtools}
\usepackage{nicefrac}
\usepackage{subfigure}
\usepackage[font=small]{caption}
\usepackage{isomath}
\usepackage{animate}

\usepackage[capitalise]{cleveref}
\newcommand{\crefrangeconjunction}{ through\nobreakspace}
\Crefname{subsubsection}{Section}{Sections}

\algrenewcommand{\algorithmiccomment}[1]{\hfill {\bf {$\mathit{//}$} #1}}
\algnewcommand{\LineComment}[1]{\State {\bf $\mathit{//}$ #1}}

\usepackage{fnbreak}

\theoremstyle{plain}
\newtheorem{theorem}{Theorem}[section]
\newtheorem{corollary}{Corollary}[section]
\newtheorem{lemma}{Lemma}
\newtheorem{proposition}{Proposition}[section]

\theoremstyle{definition}
\newtheorem{definition}{Definition}
\newtheorem{example}{Example}[section]
\newtheorem{remark}{Remark}

\newlist{parts}{enumerate}{1}
\Crefname{partsi}{Part}{Parts}
\setlist[parts,1]{label=\alph*.,ref=\alph*}

\newlist{observations}{enumerate}{1}
\Crefname{observationsi}{Observation}{Observations}
\setlist[observations,1]{label=\Roman*.,ref=\Roman*}

\newcommand{\defnitemtitle}[1]{(#1)\textbf{.}}
\newcommand{\thmitemtitle}[1]{\emph{(#1)}\textbf{.}}
\newcommand{\crefpart}[2]{\cref{#1}(\labelcref{#1-#2})}

\newcommand{\refintitle}[1]{\texorpdfstring{\ref{#1}}{\ref*{#1}}}

\newcommand{\eqdef}{\triangleq}
\newcommand{\Rge}{\mathbb{R}_+}
\newcommand{\Rundefined}{\mathbb{R}\cup\{\undefined\}}
\newcommand{\nonemptysubs}[1]{2^{#1}_{\ne\emptyset}}
\newcommand{\types}{\nonemptysubs{[n]}}
\newcommand{\type}{R}
\newcommand{\game}{\bigl((f_j)_{j=1}^n;(\mu^{\type})_{\smash{\type\in\types}}\bigr)}
\newcommand{\idgame}{\bigl((f_j)_{j=1}^n;(\type^i)_{i=1}^k;(f^i_j)_{j\in\type\smash{^i}}^{i\in[k]}\bigr)}
\newcommand{\constraint}{\Bigl((\mu^{\type})_{\type\in\types},\bigl([t_j,T_j]\bigr)_{j=1}^n\Bigr)}
\newcommand{\genconstraint}{\Bigl(\bigl([\mu^i,\mathcal{M}^i]\bigr)_{i=1}^k,\bigl([t_j,T_j]\bigr)_{j=1}^n,\bigl([m^i_j,M^i_j]\bigr)^{i\in[k]}_{j\in[n]}\Bigr)}
\newcommand{\id}{\mathrm{id}}
\newcommand{\may}[1]{Q^{#1}}
\newcommand{\equalize}[1]{\Equalize_{#1}}
\newcommand{\undefined}{\mathrm{undefined}}
\newcommand{\load}[1]{\mu^s_{#1}}
\newcommand{\height}[1]{h^s_{#1}}
\newcommand{\loadt}[1]{\mu^{s'}_{#1}}
\newcommand{\heightt}[1]{h^{s'}_{#1}}
\newcommand{\loadtt}[1]{\mu^{s''}_{#1}}
\newcommand{\heighttt}[1]{h^{s''}_{#1}}
\newcommand{\heightttt}[1]{h^{s'''}_{#1}}

\DeclareMathOperator{\supp}{supp}
\DeclareMathOperator*{\Max}{Max}
\DeclareMathOperator*{\Equalize}{Equalize}

\xdefinecolor{green}{rgb}{0,0.7,0}
\xdefinecolor{purple}{rgb}{0.55,0,0.8}
\xdefinecolor{pink}{rgb}{0.97,0.52,0.89}
\xdefinecolor{brown}{rgb}{0.54,0.27,0.05}

\newcounter{m}
\newcounter{res}
\newcommand{\balloons}[1]{%
\begin{tikzpicture}[#1]
\pgfmathsetmacro{\containerheight}{4.5}
\pgfmathsetmacro{\pistonstart}{5.5}
\pgfmathsetmacro{\yellowamount}{1}
\pgfmathsetmacro{\cyanamount}{1/4}
\pgfmathsetmacro{\purpleamount}{3}
\pgfmathsetmacro{\blueamount}{18/16}
\pgfmathsetmacro{\redamount}{1.25}
\pgfmathsetmacro{\greenamount}{3/4}
\pgfmathsetmacro{\orangeamount}{3/4-1/12+1/32}
\pgfmathsetmacro{\brownamount}{1-1/96}
\pgfmathsetmacro{\pinkamount}{1.5}
\pgfmathsetmacro{\pistonheight}{4-4*\them/\theres}
\ifthenelse{\cnttest{\them*2}<{\theres}}{
\ifthenelse{\cnttest{\them*384}<{\theres*181}}{
\ifthenelse{\cnttest{\them*16}<{\theres*7}}{
\ifthenelse{\cnttest{\them*768}<{\theres*317}}{
\ifthenelse{\cnttest{\them*192}<{\theres*77}}{
\ifthenelse{\cnttest{\them*64}<{\theres*25}}{
\ifthenelse{\cnttest{\them*192}<{\theres*64}}{
\ifthenelse{\cnttest{\them*64}<{\theres*20}}{
\ifthenelse{\cnttest{\them*64}<{\theres*17}}{
\pgfmathsetmacro{\blueinfirst}{\blueamount/3}
\pgfmathsetmacro{\cyaninfirst}{\cyanamount/4}
\pgfmathsetmacro{\purpleinfirst}{\purpleamount/2}
}{
\pgfmathsetmacro{\blueinfirst}{\blueamount/3-(\them/\theres-17/64)*4/3}
\pgfmathsetmacro{\cyaninfirst}{\cyanamount/4-(\them/\theres-17/64)*4/3}
\pgfmathsetmacro{\purpleinfirst}{\purpleamount/2-(\them/\theres-17/64)*4/3}
}
\pgfmathsetmacro{\blueinthird}{(\blueamount-\blueinfirst)/2}
\pgfmathsetmacro{\cyaninfifth}{(\cyanamount-\cyaninfirst)/3}
\pgfmathsetmacro{\cyaninthird}{(\cyanamount-\cyaninfirst)/3}
}{
\pgfmathsetmacro{\cyaninfirst}{0}
\pgfmathsetmacro{\blueinfirst}{\blueamount/3-1/16-(\them/\theres-20/64)*4}
\pgfmathsetmacro{\purpleinfirst}{\purpleamount/2-1/16}
\pgfmathsetmacro{\blueinthird}{\blueamount/3+1/32}
\pgfmathsetmacro{\cyaninthird}{\cyanamount/3-(\them/\theres-20/64)*4}
\pgfmathsetmacro{\cyaninfifth}{(\cyanamount-\cyaninthird)/2}
}
}{
\pgfmathsetmacro{\blueinfirst}{\blueamount/3-1/16-(\them/\theres-20/64)*4}
\pgfmathsetmacro{\purpleinfirst}{\purpleamount/2-1/16}
\pgfmathsetmacro{\blueinthird}{\blueamount/3+1/32-(\them/\theres-64/192)*4}
\pgfmathsetmacro{\cyaninfirst}{0}
\pgfmathsetmacro{\cyaninthird}{0}
\pgfmathsetmacro{\cyaninfifth}{\cyanamount/2}
}
}{
\pgfmathsetmacro{\cyaninfirst}{0}
\pgfmathsetmacro{\cyaninthird}{0}
\pgfmathsetmacro{\blueinfirst}{0}
\pgfmathsetmacro{\purpleinfirst}{\purpleamount/2-1/16-(\them/\theres-25/64)*4}
\pgfmathsetmacro{\blueinthird}{\blueamount/3+1/32-11/48-(\them/\theres-25/64)*8}
\pgfmathsetmacro{\cyaninfifth}{\cyanamount/2}
}
\pgfmathsetmacro{\greeninfourth}{\greenamount/2}
}{
\pgfmathsetmacro{\cyaninfirst}{0}
\pgfmathsetmacro{\cyaninthird}{0}
\pgfmathsetmacro{\blueinfirst}{0}
\pgfmathsetmacro{\blueinthird}{\blueamount/3+1/32-11/48-(\them/\theres-25/64)*8}
\pgfmathsetmacro{\purpleinfirst}{\purpleamount/2-1/16-(\them/\theres-25/64)*4}
\pgfmathsetmacro{\cyaninfifth}{\cyanamount/2}
\pgfmathsetmacro{\greeninfourth}{\greenamount/2-(\them/\theres-77/192)*12}
}
}{
\pgfmathsetmacro{\cyaninfirst}{0}
\pgfmathsetmacro{\cyaninthird}{0}
\pgfmathsetmacro{\blueinfirst}{0}
\pgfmathsetmacro{\blueinthird}{0}
\pgfmathsetmacro{\purpleinfirst}{(\purpleamount+\orangeamount-\yellowamount)/2}
\pgfmathsetmacro{\cyaninfifth}{\cyanamount/2}
\pgfmathsetmacro{\greeninfourth}{\greenamount/2-108/768-(\them/\theres-317/768)*4}
}
\pgfmathsetmacro{\redinfifth}{\redamount/2}
}{
\pgfmathsetmacro{\cyaninfirst}{0}
\pgfmathsetmacro{\cyaninthird}{0}
\pgfmathsetmacro{\blueinfirst}{0}
\pgfmathsetmacro{\blueinthird}{0}
\pgfmathsetmacro{\purpleinfirst}{(\purpleamount+\orangeamount-\yellowamount)/2}
\pgfmathsetmacro{\greeninfourth}{\greenamount/2-108/768-(\them/\theres-317/768)*4}
\pgfmathsetmacro{\cyaninfifth}{\cyanamount/2-(\them/\theres-7/16)*2}
\pgfmathsetmacro{\redinfifth}{\redamount/2-(\them/\theres-7/16)*2}
}
}{
\pgfmathsetmacro{\cyaninfirst}{0}
\pgfmathsetmacro{\cyaninthird}{0}
\pgfmathsetmacro{\blueinfirst}{0}
\pgfmathsetmacro{\blueinthird}{0}
\pgfmathsetmacro{\purpleinfirst}{(\purpleamount+\orangeamount-\yellowamount)/2}
\pgfmathsetmacro{\greeninfourth}{0}
\pgfmathsetmacro{\cyaninfifth}{\cyanamount/2-(\them/\theres-7/16)*2}
\pgfmathsetmacro{\redinfifth}{\redamount/2-(\them/\theres-7/16)*2}
}
}{
\pgfmathsetmacro{\cyaninfirst}{0}
\pgfmathsetmacro{\cyaninthird}{0}
\pgfmathsetmacro{\blueinfirst}{0}
\pgfmathsetmacro{\blueinthird}{0}
\pgfmathsetmacro{\purpleinfirst}{(\purpleamount+\orangeamount-\yellowamount)/2}
\pgfmathsetmacro{\greeninfourth}{0}
\pgfmathsetmacro{\cyaninfifth}{0}
\pgfmathsetmacro{\redinfifth}{\redamount/2-1/8-(\them/\theres-1/2)*4}
}
\pgfmathsetmacro{\yellowinfirst}{\yellowamount}
\pgfmathsetmacro{\purpleinthird}{\purpleamount-\purpleinfirst}
\pgfmathsetmacro{\blueinfourth}{\blueamount-\blueinfirst-\blueinthird}
\pgfmathsetmacro{\cyaninsecond}{(\cyanamount-\cyaninfirst-\cyaninthird-\cyaninfifth}
\pgfmathsetmacro{\orangeinthird}{\orangeamount}
\pgfmathsetmacro{\browninfourth}{\brownamount}
\pgfmathsetmacro{\redinsecond}{\redamount-\redinfifth}
\pgfmathsetmacro{\greeninsecond}{\greenamount-\greeninfourth}
\pgfmathsetmacro{\pinkinfifth}{\pinkamount}
\draw[cyan,very thick] (0.5,\yellowinfirst+\cyaninfirst/2,0)
parabola bend (1.5,\yellowinfirst/2+\cyaninfirst/4+\redinsecond/2+\cyaninsecond/4-.25,0)
(2.5,\redinsecond+\cyaninsecond/2,0)
parabola bend (3.5,\redinsecond/2+\cyaninsecond/4+\orangeinthird/2+\cyaninthird/4-.25,0)
(4.5,\orangeinthird+\cyaninthird/2,0)
parabola bend (6.5,\orangeinthird/2+\cyaninthird/4+\pinkinfifth/2+\cyaninfifth/4-.25,0)
(8.5,\pinkinfifth+\cyaninfifth/2,0);
\draw[purple,very thick] (0.5,\yellowinfirst+\cyaninfirst+\purpleinfirst/2,0)
parabola bend (2.5,\yellowinfirst/2+\cyaninfirst/2+\purpleinfirst/4+\orangeinthird/2+\cyaninthird/2+\purpleinthird/4-.25,0)
(4.5,\orangeinthird+\cyaninthird+\purpleinthird/2,0);
\draw[blue,very thick] (0.5,\yellowinfirst+\cyaninfirst+\purpleinfirst+\blueinfirst/2,0)
parabola bend (2.5,\yellowinfirst/2+\cyaninfirst/2+\purpleinfirst/2+\blueinfirst/4+\orangeinthird/2+\cyaninthird/2+\purpleinthird/2+\blueinthird/4-.25,0)
(4.5,\orangeinthird+\cyaninthird+\purpleinthird+\blueinthird/2,0)
parabola bend (5.5,\orangeinthird/2+\cyaninthird/2+\purpleinthird/2+\blueinthird/4+\browninfourth/2+\greeninfourth/2+\blueinfourth/4-.25,0)
(6.5,\browninfourth+\greeninfourth+\blueinfourth/2,0);
\draw[red,very thick] (2.5,\redinsecond/2,0)
parabola bend (5.5,\redinsecond/4+\pinkinfifth/2+\cyaninfifth/2+\redinfifth/4-.25,0)
(8.5,\pinkinfifth+\cyaninfifth+\redinfifth/2,0);
\draw[green,very thick] (2.5,\redinsecond+\cyaninsecond+\greeninsecond/2,0)
parabola bend (4.5,\redinsecond/2+\cyaninsecond/2+\greeninsecond/4+\browninfourth/2+\greeninfourth/4-.25,0)
(6.5,\browninfourth+\greeninfourth/2,0);
\foreach \x in {0, 1, 2, 3, 4}
{
\draw[shift={(\x*2,0,0)}] (0,0,0) -- (0,\containerheight,0) -- (1,\containerheight,0);
\draw[shift={(\x*2,0,0)}] (0,\containerheight,0) -- (0,\containerheight,1);
}
\draw[fill=yellow,draw=none] (0,0,1) -- (1,0,1) -- (1,0,0) -- (1,\yellowinfirst,0) -- (0,\yellowinfirst,0) -- (0,\yellowinfirst,1) -- cycle;
\draw[fill=cyan,draw=none,shift={(0,\yellowinfirst,0)}] (0,0,1) -- (1,0,1) -- (1,0,0) -- (1,\cyaninfirst,0) -- (0,\cyaninfirst,0) -- (0,\cyaninfirst,1) -- cycle;
\draw[fill=purple,draw=none,shift={(0,\yellowinfirst+\cyaninfirst,0)}] (0,0,1) -- (1,0,1) -- (1,0,0) -- (1,\purpleinfirst,0) -- (0,\purpleinfirst,0) -- (0,\purpleinfirst,1) -- cycle;
\draw[fill=blue,draw=none,shift={(0,\yellowinfirst+\cyaninfirst+\purpleinfirst,0)}] (0,0,1) -- (1,0,1) -- (1,0,0) -- (1,\blueinfirst,0) -- (0,\blueinfirst,0) -- (0,\blueinfirst,1) -- cycle;
\draw[shift={(0,\yellowinfirst+\cyaninfirst+\purpleinfirst+\blueinfirst,0)}] (0,0,0) -- (0,0,1) -- (1,0,1) -- (1,0,0) -- cycle;
\draw[shift={(0,{max(\pistonheight,\yellowinfirst+\cyaninfirst+\purpleinfirst+\blueinfirst)},0)},fill=gray] (0,0,0) -- (0,0,1) -- (1,0,1) -- (1,0,0) -- cycle;
\draw[shift={(.5,0,.5)},very thick] (0,\pistonstart,0) -- (0,{max(\pistonheight,\yellowinfirst+\cyaninfirst+\purpleinfirst+\blueinfirst)},0);
\draw[fill=red,draw=none,shift={(2,0,0)}] (0,0,1) -- (1,0,1) -- (1,0,0) -- (1,\redinsecond,0) -- (0,\redinsecond,0) -- (0,\redinsecond,1) -- cycle;
\draw[fill=cyan,draw=none,shift={(2,\redinsecond,0)}] (0,0,1) -- (1,0,1) -- (1,0,0) -- (1,\cyaninsecond,0) -- (0,\cyaninsecond,0) -- (0,\cyaninsecond,1) -- cycle;
\draw[fill=green,draw=none,shift={(2,\redinsecond+\cyaninsecond,0)}] (0,0,1) -- (1,0,1) -- (1,0,0) -- (1,\greeninsecond,0) -- (0,\greeninsecond,0) -- (0,\greeninsecond,1) -- cycle;
\draw[shift={(2,\redinsecond+\cyaninsecond+\greeninsecond,0)}] (0,0,0) -- (0,0,1) -- (1,0,1) -- (1,0,0) -- cycle;
\draw[shift={(2,{max(\pistonheight,\redinsecond+\cyaninsecond+\greeninsecond)},0)},fill=gray] (0,0,0) -- (0,0,1) -- (1,0,1) -- (1,0,0) -- cycle;
\draw[shift={(2.5,0,.5)},very thick] (0,\pistonstart,0) -- (0,{max(\pistonheight,\redinsecond+\cyaninsecond+\greeninsecond)},0);
\draw[fill=orange,draw=none,shift={(4,0,0)}] (0,0,1) -- (1,0,1) -- (1,0,0) -- (1,\orangeinthird,0) -- (0,\orangeinthird,0) -- (0,\orangeinthird,1) -- cycle;
\draw[fill=cyan,draw=none,shift={(4,\orangeinthird,0)}] (0,0,1) -- (1,0,1) -- (1,0,0) -- (1,\cyaninthird,0) -- (0,\cyaninthird,0) -- (0,\cyaninthird,1) -- cycle;
\draw[fill=purple,draw=none,shift={(4,\orangeinthird+\cyaninthird,0)}] (0,0,1) -- (1,0,1) -- (1,0,0) -- (1,\purpleinthird,0) -- (0,\purpleinthird,0) -- (0,\purpleinthird,1) -- cycle;
\draw[fill=blue,draw=none,shift={(4,\orangeinthird+\cyaninthird+\purpleinthird,0)}] (0,0,1) -- (1,0,1) -- (1,0,0) -- (1,\blueinthird,0) -- (0,\blueinthird,0) -- (0,\blueinthird,1) -- cycle;
\draw[shift={(4,\orangeinthird+\cyaninthird+\purpleinthird+\blueinthird,0)}] (0,0,0) -- (0,0,1) -- (1,0,1) -- (1,0,0) -- cycle;
\draw[shift={(4,{max(\pistonheight,\orangeinthird+\cyaninthird+\purpleinthird+\blueinthird)},0)},fill=gray] (0,0,0) -- (0,0,1) -- (1,0,1) -- (1,0,0) -- cycle;
\draw[shift={(4.5,0,.5)},very thick] (0,\pistonstart,0) -- (0,{max(\pistonheight,\orangeinthird+\cyaninthird+\purpleinthird+\blueinthird)},0);
\draw[fill=brown,draw=none,shift={(6,0,0)}] (0,0,1) -- (1,0,1) -- (1,0,0) -- (1,\browninfourth,0) -- (0,\browninfourth,0) -- (0,\browninfourth,1) -- cycle;
\draw[fill=green,draw=none,shift={(6,\browninfourth,0)}] (0,0,1) -- (1,0,1) -- (1,0,0) -- (1,\greeninfourth,0) -- (0,\greeninfourth,0) -- (0,\greeninfourth,1) -- cycle;
\draw[fill=blue,draw=none,shift={(6,\browninfourth+\greeninfourth,0)}] (0,0,1) -- (1,0,1) -- (1,0,0) -- (1,\blueinfourth,0) -- (0,\blueinfourth,0) -- (0,\blueinfourth,1) -- cycle;
\draw[shift={(6,\browninfourth+\greeninfourth+\blueinfourth,0)}] (0,0,0) -- (0,0,1) -- (1,0,1) -- (1,0,0) -- cycle;
\draw[shift={(6,{max(\pistonheight,\browninfourth+\greeninfourth+\blueinfourth)},0)},fill=gray] (0,0,0) -- (0,0,1) -- (1,0,1) -- (1,0,0) -- cycle;
\draw[shift={(6.5,0,.5)},very thick] (0,\pistonstart,0) -- (0,{max(\pistonheight,\browninfourth+\greeninfourth+\blueinfourth)},0);
\draw[fill=pink,draw=none,shift={(8,0,0)}] (0,0,1) -- (1,0,1) -- (1,0,0) -- (1,\pinkinfifth,0) -- (0,\pinkinfifth,0) -- (0,\pinkinfifth,1) -- cycle;
\draw[fill=cyan,draw=none,shift={(8,\pinkinfifth,0)}] (0,0,1) -- (1,0,1) -- (1,0,0) -- (1,\cyaninfifth,0) -- (0,\cyaninfifth,0) -- (0,\cyaninfifth,1) -- cycle;
\draw[fill=red,draw=none,shift={(8,\pinkinfifth+\cyaninfifth,0)}] (0,0,1) -- (1,0,1) -- (1,0,0) -- (1,\redinfifth,0) -- (0,\redinfifth,0) -- (0,\redinfifth,1) -- cycle;
\draw[shift={(8,\pinkinfifth+\cyaninfifth+\redinfifth,0)}] (0,0,0) -- (0,0,1) -- (1,0,1) -- (1,0,0) -- cycle;
\draw[shift={(8,{max(\pistonheight,\pinkinfifth+\cyaninfifth+\redinfifth)},0)},fill=gray] (0,0,0) -- (0,0,1) -- (1,0,1) -- (1,0,0) -- cycle;
\draw[shift={(8.5,0,.5)},very thick] (0,\pistonstart,0) -- (0,{max(\pistonheight,\pinkinfifth+\cyaninfifth+\redinfifth)},0);
\foreach \x in {0, 1, 2, 3, 4}
{
\draw[shift={(\x*2,0,0)}] (0,\containerheight,1) -- (0,0,1) -- (1,0,1) -- (1,\containerheight,1) -- cycle;
\draw[shift={(\x*2,0,0)}] (1,\containerheight,1) -- (1,\containerheight,0) -- (1,0,0) -- (1,0,1);
}
\end{tikzpicture}%
}

\newcommand{\balloonsstatic}[1]{%
\setcounter{m}{#1}%
\balloons{yscale=0.613,xscale=0.785}%
}

\if\withanimation1
\newcommand{\balloonsanim}{
\setcounter{m}{0}%
\begin{animateinline}[poster = first, controls,buttonsize=4em]{64}%
\whiledo{\not{\cnttest{32*\them}>{17*\theres}}}{%
\balloons{xscale=1.505,yscale=1.497}%
\ifthenelse{\cnttest{32*\them}<{17*\theres}}{%
\newframe%
}{%
\end{animateinline}%
}%
\stepcounter{m}%
}%
}
\fi

\newcommand{\balloonswithanimation}{balloons\if\withanimation1,balloons-anim\fi}

\newcommand{\minvessel}[2]{
\draw[x={#2 cm},shift={(#1,0)}] (0,0.2) -- (0,5);
\draw[x={#2 cm},shift={(#1,0)}] (1,0.2) -- (1,2.25);
\draw[x={#2 cm},shift={(#1,0)}] (1,2.25) -- (1.5,1.5);
\draw[x={#2 cm},shift={(#1,0)}] (1.5,1.5) -- (1.25,1.5);
\draw[x={#2 cm},shift={(#1,0)}] (1.25,1.5) -- (1.25,0.5);
\draw[x={#2 cm},shift={(#1,0)}] (1.25,0.5) -- (3.9,0.5);
\draw[x={#2 cm},shift={(#1,0)}] (3.9,0.5) -- (3.9,1.5);
\draw[x={#2 cm},shift={(#1,0)}] (3.9,1.5) -- ({1+2.5/3},1.5);
\draw[x={#2 cm},shift={(#1,0)}] ({1+2.5/3},1.5) -- (1,2.75);
\draw[x={#2 cm},shift={(#1,0)}] (1,2.75) -- (1,5);
}

\newcommand{\oddshapeoutline}{
\draw[shift={(0*1.2,0)}] (0,0.1) -- (0,5);
\draw[shift={(0*1.2,0)}] (1,0.2) -- (1,5);
\draw[shift={(0*1.2,0)}] (1.2,0.2) -- (1,0.2);
\draw[shift={(1*1.2,0)}] (0,0.2) -- (0,5);
\draw[shift={(1*1.2,0)}] (.5,0.2) -- (.5,5);
\draw[shift={(1*1.2,0)}] (1.2,0.2) -- (.5,0.2);
\draw[shift={(2*1.2,0)}] (0,0.2) -- (-.5,5);
\draw[shift={(2*1.2,0)}] (1,0.2) -- (1,5);
\draw[shift={(2*1.2,0)}] (1.2,0.2) -- (1,0.2);
\draw[shift={(3*1.2,0)}] (0,0.2) -- (0,5);
\draw[shift={(3*1.2,0)}] plot[variable=\t,samples=100,domain=1:5.8] ({.5+.5/\t,0.2+\t-1});
\draw[shift={(3*1.2,0)}] (1.075,0.2) -- (1,0.2);
\draw[shift={(4*1.2,0)}] (-.125,0.2) -- (-.125,0.8);
\draw[shift={(4*1.2,0)}] (-.125,0.8) -- (0,0.8);
\draw[shift={(4*1.2,0)}] (-.075,0.2) -- (-.075,0.7);
\draw[shift={(4*1.2,0)}] (0,0.8) -- (0,5);
\draw[shift={(4*1.2,0)}] (-0.075,0.7) -- (1.075,0.7);
\draw[shift={(4*1.2,0)}] (1,0.8) -- (1,5);
\draw[shift={(4*1.2,0)}] (-0.075,0.2) -- (1.075,.2);
\draw[shift={(4*1.2,0)}] (1.125,0.2) -- (1.125,0.8);
\draw[shift={(4*1.2,0)}] (1.125,0.8) -- (1,0.8);
\draw[shift={(4*1.2,0)}] (1.075,0.2) -- (1.075,0.7);
\draw[shift={(4*1.2,0)}] (1.125,0.2) -- (1.2,0.2);
\minvessel{5*1.2}{1}
\draw (0,0.1) -- (5*1.2+1,0.1) -- (5*1.2+1,0.2);
}

\newcommand{\compressor}[1]{
\pgfmathsetmacro{\containerheight}{2.5}
\pgfmathsetmacro{\blueinfirst}{#1}
\pgfmathsetmacro{\redinsecond}{#1}
\draw[blue,very thick] (0.5,\blueinfirst/2) --
(1.7,{\redinsecond+2*(1-\blueinfirst)/2});
\draw[red,very thick] (0.5,{\blueinfirst+2*(1-\redinsecond)/2}) --
(1.7,\redinsecond/2);
\draw[draw=none,fill=blue,shift={(0*1.2,0)}] (0,0) -- (1,0) -- (1,\blueinfirst) -- (0,\blueinfirst) -- cycle;
\draw[draw=none,fill=red,shift={(0*1.2,\blueinfirst)}] (0,0) -- (1,0) -- (1,{2*(1-\redinsecond)}) -- (0,{2*(1-\redinsecond)}) -- cycle;
\draw[draw=none,fill=red,shift={(1*1.2,0)}] (0,0) -- (1,0) -- (1,\redinsecond) -- (0,\redinsecond) -- cycle;
\draw[draw=none,fill=blue,shift={(1*1.2,\redinsecond)}] (0,0) -- (1,0) -- (1,{2*(1-\blueinfirst)}) -- (0,{2*(1-\blueinfirst)}) -- cycle;
\draw[shift={(0*1.2,0)}] (0,0) -- (0,\containerheight);
\draw[shift={(0*1.2,0)}] (1,0) -- (1,\containerheight);
\draw[shift={(0*1.2,0)}] (0,0) -- (1,0);
\draw[shift={(1*1.2,0)}] (0,0) -- (0,\containerheight);
\draw[shift={(1*1.2,0)}] (1,0) -- (1,\containerheight);
\draw[shift={(1*1.2,0)}] (0,0) -- (1,0);
}

\hypersetup{
  pdfauthor      = {Yannai A. Gonczarowski <yannai@gonch.name> and Moshe Tennenholtz <moshet@ie.technion.ac.il>},
  pdftitle       = {A Hydraulic Approach to Equilibria of Resource Selection Games},
}

\title{A Hydraulic Approach to Equilibria of Resource Selection Games}
\author{Yannai A. Gonczarowski\thanks{Einstein Institute of Mathematics, Rachel \& Selim Benin School of Computer Science \& Engineering and Federmann Center for the Study of Rationality,
The Hebrew University of Jerusalem, Israel; and Microsoft Research, \emph{E-mail}: \href{mailto:yannai@gonch.name}{yannai@gonch.name}.} \and
Moshe Tennenholtz\thanks{William Davidson Faculty of Industrial Engineering and Management, Technion --- Israel Institute of Technology (work carried out in part while at Microsoft Research), \emph{E-mail}: \href{mailto:moshet@ie.technion.ac.il}{moshet@ie.technion.ac.il}.}}
\date{\emph{Dedicated to Dr.\ Bella Kessler, my high-school math teacher, \\ for continuously encouraging my nonstandard proofs.} {\sc-y.a.g.}\\\vspace{1em} June 5, 2016}

\begin{document}
\maketitle

\vspace{-1.4em}
{\small
\begin{quotation}\noindent
{\large \emph{``All is water''}\quad \normalsize ---Thales, c.\ 585 BC}
\end{quotation}}

\begin{abstract}
Drawing intuition from a (physical) hydraulic system, we present a novel framework, constructively showing the existence of a strong Nash equilibrium in resource
selection games (i.e., asymmetric singleton congestion games) with nonatomic players, the coincidence of strong equilibria and Nash equilibria in such games,
and the uniqueness of the cost of each given resource across all Nash equilibria.
Our proofs allow for explicit calculation of Nash equilibrium and for explicit and direct calculation of the resulting (unique) costs of resources,
and do not hinge on any fixed-point theorem, on the Minimax theorem or any equivalent result, on linear programming, or on the existence of a potential (though our analysis does provide powerful insights into the potential, via a natural concrete physical interpretation).
A generalization of resource selection games, called resource selection games with I.D.-dependent weighting, is defined, and the results are extended to this family,
showing the existence of strong equilibria, and
showing that while resource costs are no longer unique across Nash equilibria in games of this family, they are nonetheless unique across all strong Nash equilibria, drawing
a novel fundamental connection between group deviation and I.D.-congestion.
A natural application of the resulting machinery to a large class of constraint-satisfaction problems is also described.
\end{abstract}


\section{Introduction}

Taking which highway would allow me to arrive at my workplace as fast as possible this morning? Using which computer server would my jobs be completed the soonest? Which router would deliver my packets with least latency? And even... shopping at which fashion store would make my clothes as unique as possible? All these, and more, are dilemmas faced by players in \emph{resource selection games} --- games in which each player's payoff depends solely on the quantity of players choosing the same strategy (resource) as that player --- and, more generally, in \emph{congestion games} --- games in which each player chooses a feasible strategy set (e.g., road segments), and, roughly, aims for its intersections with other chosen strategy sets to be small.

Congestion games \citep{Rosenthal73,MondererShapley96} have been central to the interplay between computer science and game theory \citep{Nisanbook}.
These games arise naturally in many contexts and possess various desirable properties; in particular, both \emph{atomic} congestion games (where each of the finitely many players has positive contribution to the congestion) and \emph{nonatomic} congestion games (where the singular contribution of each of the continuum of players to the congestion is negligible) possess pure-strategy Nash equilibria. It is therefore only natural that topics of major interest in the field of Algorithmic Game Theory, such as the \emph{price of anarchy} \citep{KoutsoupiasPapadimitriouPoA,PapadimitriouPoA,RoughgardenTardos}, which quantifies the social loss in Nash equilibria, have been introduced in the context of such games.

From a game-theoretic perspective, much additional effort has been spent in introducing important extensions of congestion games
in the context of {\bf atomic} congestion games, and in particular in the context of atomic resource selection games.
Such efforts include the study of strong equilibria (stability against group deviations)
in various such games \citep{HolzmanCong1,HolzmanRoute,FeldmanMansour}; player-specific congestion games \citep{Milchtaich96}, where cost functions may be player-specific; and I.D.-congestion games \citep{MondererSolCon}, where the cost of a resource may depend on the identity (rather than merely on the quantity) of the players using it.
Interestingly, the challenges of dealing with such major extensions in the {\bf nonatomic} case have only been partially tackled, largely using only tools generalizing those developed for atomic games. Indeed, while \citeauthor{Milchtaich05} studies models with player-specific payoffs \citep{Milchtaich05} and strong equilibria with player-specific payoffs (\citealp{Milchtaich06}; \citealp[see][for a more restricted model]{Milchtaich04}), extensions dealing with I.D.-congestion (where cost depends on the identity of other players using the same resource) have not been tackled at all, nor have any tools been offered in order to deal with such extensions. This lacuna is especially puzzling given the centrality of nonatomic congestion games in presenting flow and communication networks, central to computer science, as well as in presenting large markets and economies, central to macroeconomics.

In this paper, we deal with such extensions of the study of nonatomic congestion games, and address the related challenges by introducing a novel analysis approach, which we call {\em hydraulic computing}. Using this approach, which draws intuition from a (physical) hydraulic system, in \cref{id-independent} we show the existence of strong equilibria in nonatomic resource selection games, that strong equilibria and Nash equilibria coincide in such games, and that the cost of each given resource is unique across all Nash equilibria. Generalizing to I.D.-congestion games, in \cref{compress-expand}
we show the existence of strong equilibria in resource selection games where the cost of a resource depends on the identity of the players using it, and that the cost of each given resource in such games, while interestingly no longer unique across Nash equilibria, is nonetheless unique across all strong equilibria, drawing a novel fundamental connection between group deviation\footnote{While coalitional deviations in large-scale economies, such as nonatomic games, require massive coordination to involve coalitions of nonnegligible measure, we note that such deviations are by no means purely theoretic; indeed, modern cloud-based social application such as Waze (e.g., for congestion games on graphs) allow for centralized coordination of deviations of immense scales.} and I.D.-congestion. Our theoretical treatment does not hinge on any fixed-point theorem, on the Minimax theorem or any equivalent result, on linear programming, or on the existence of a potential (though it does provide powerful insights into the potential when a potential exists, via a natural concrete physical interpretation --- see \cref{potential}), and is the first to provide explicit formulation of the resource cost obtained in equilibria of congestion games.
Looking beyond the realm of games, in \cref{hall-and-beyond} we show that our framework can serve as a constructive substitute to linear-programming approaches in other contexts as well,
such as that of Hall's marriage theorem and many constraint-satisfaction problems generalizing~it.

\paragraph{Atomic Games}
Atomic congestion games with finitely many players have been introduced by \cite{Rosenthal73}, who has shown the existence of a pure-strategy Nash equilibrium in such games. \cite{MondererShapley96} later introduced potential games, and showed that they coincide with congestion games (with finitely many players). \cite{HolzmanCong1,HolzmanRoute} studied strong equilibria in congestion games and characterized settings in which a strong equilibrium exists; in particular, they showed that the set of strong equilibria and the set of Nash equilibria coincide in resource selection games (with finitely many players). \cite{Milchtaich96} extended congestion games to player-specific congestion games, in which players' costs are player-specific, and showed the existence of Nash equilibrium in player-specific resource selection games (with finitely many players).
\cite{MondererSolCon} introduced a general class of I.D.-congestion games, which are congestion games in which the cost of a resource depends also on the identity of the players using it. On the verge between a finite and a countable cardinality of players, \cite{Milchtaich00} showed the existence and uniqueness of Nash equilibria (uniqueness of strategies, not only of costs of resources) in large replications of generic finite resource selection games, as well as in the limit countable-player game.

\paragraph{Nonatomic Games}
Nonatomic congestion games, such as those that we study, have been very popular in the computer science context; see, e.g., \cite{RoughgardenTardos} and \cite{Nisanbook}.
In such games, \cite{Beckmann56} have shown the existence of Nash equilibrium and the uniqueness of costs of resources in Nash equilibria under certain differentiability assumptions;
a general theorem by \cite{Schmeidler73} implies the existence of Nash equilibrium under the assumption of continuity (rather than differentiability) of the cost functions. \cite{Milchtaich05} characterizes congestion games with player-specific costs in which equilibrium resource costs are unique; in particular, he shows that this is the case in resource selection games with strictly increasing continuous cost functions. \cite{Milchtaich04,Milchtaich06} studies strong equilibria with player-specific costs, and in particular shows the coincidence of Nash and strong equilibria in resource selection games with strictly increasing continuous cost functions.

\vspace{1.1\baselineskip}
\noindent
Analogies to hydraulic systems have sporadically appeared in the economics literature in the past, but seem to be anecdotal in nature.
At the end of the nineteenth century, \cite{Fisher} built a complex hydraulic apparatus for calculating Walrasian equilibrium prices in competitive markets with up to three goods.
\cite{Kaminsky} \cite[see also][]{AumannThreeWives} uses an analogy to a simple hydraulic system to find the nucleolus of a small special set of cooperative games. While of their own interest, we note that it does not seem that any ``deep'' connection exists between the \emph{ad hoc} hydraulic analogies in these papers and our hydraulic framework.

\paragraph{Contributions}
The main contributions of this paper are:
\begin{enumerate}
\item
Introducing the hydraulic computing analysis framework.
\item
Providing an explicit formula (rather than an iterative procedure of computation) for calculating the cost of resources in equilibria of nonatomic resource selection games.
\item
Proving the uniqueness of equilibrium resource costs without any assumption of differentiability or even continuity; and relaxing the requirement of strict monotonicity of cost functions to weak monotonicity across several key results such as existence of strong equilibria, \linebreak equivalence of strong and Nash equilibria, and uniqueness of equilibrium resource costs.
\item
Proving the existence of strong equilibria in resource selection games with I.D.-dependent weighting with continuous cost and weighting functions, and the uniqueness of resource costs across strong equilibria in such games regardless of continuity (showing by example that these costs are not unique across all Nash equilibria),
drawing a novel fundamental connection between group deviation and I.D.-congestion.
\item
Applying hydraulic computing in lieu of linear-programming methods
in a large class of constraint-satisfaction problems, such as generalizations of finding a perfect marriage and proving Hall's theorem.
\end{enumerate}

\section{Notation}

\begin{definition}[Notation]
\leavevmode
\begin{itemize}
\item
\defnitemtitle{Naturals}
We denote the strictly positive natural numbers by $\mathbb{N}\eqdef\{1,2,3,\ldots\}$.
\item
\defnitemtitle{[n]}
For every $n \in \mathbb{N}$, we define $[n]\eqdef\{1,2,\ldots,n\}$.
\item
\defnitemtitle{Reals}
We denote the real numbers by $\mathbb{R}$.
\item
\defnitemtitle{Nonnegative Reals}
We denote the nonnegative reals by $\Rge\eqdef\{r \in \mathbb{R} \mid r \ge 0\}$.
\item
\defnitemtitle{Maximizing Arguments}
Given a set $S$ and a function $f:S\rightarrow\mathbb{R}$ that  attains a maximum value on $S$, we denote the \emph{set} of arguments in $S$ maximizing $f$ by $\arg\Max_{s\in S}f(s)\eqdef\linebreak\{s \in S \mid f(s)=m\}$, where $m\eqdef\Max_{s\in S}f(s)$.
\item
\defnitemtitle{Simplex}
For a set $\type\subseteq S$, we define:

\hfil $\Delta^{\type}=\bigl\{s \in [0,1]^S \:\big|\: \sum_{j \in \type}s_j=1 \And \forall j \in S\setminus \type:s_j=0\bigr\}$.

(The set~$S$ will be clear from context.)
\item
\defnitemtitle{Nonempty Subsets}
For a set $S$, we define $\nonemptysubs{S}\eqdef2^S\setminus\emptyset$ --- the nonempty subsets of~$S$.
\end{itemize}
\end{definition}

\begin{definition}[Plateau Height]
Let $f:\mathbb{R}\rightarrow\mathbb{R}$ be a nondecreasing function. We say that $h\in\mathbb{R}$ is a \emph{plateau height} of $f$ if there exist $x\ne y\in\mathbb{R}$ s.t.\ $f(x)=f(y)=h$.
\end{definition}

\begin{remark}
A strictly increasing function has no plateau heights.
\end{remark}

\section{\texorpdfstring{``}{"}Standard\texorpdfstring{''}{"} Resource Selection Games}\label{id-independent}

\subsection{Setting}

\begin{definition}[Resource Selection]
Let $n\in\mathbb{N}$. An \emph{$n$-resource selection game} is defined by a pair $\game$,
where $f_j:\Rge\rightarrow\mathbb{R}$ is a nondecreasing function for every $j\in[n]$, and $\mu^{\type}\in\Rge$ for every~$\type\in\types$.
\end{definition}

In a resource selection game, each $\type\in\types$  indicates a player type. Each player of type $\type$ may consume only from resources in $\type$; the total mass of the continuum of players of type $\type$ is $\mu^{\type}$. For each resource~$j\in[n]$, $f_j$ is a function from the consumption amount of this resource to the cost of consuming from the resource. We now formally define these concepts.

\begin{definition}[Consumption Profile; $\load{j}$; $\height{j}$; Nash Equilibrium]
Let $G=\game$ be a resource selection game.
\begin{itemize} 
\item
A \emph{consumption (strategy) profile} in $G$ is a function
$s:\types\rightarrow \Rge^{[n]}$ s.t.\ $s(\type) \in \mu^{\type}\cdot\Delta^{\type}$ for every $\type\in\types$.
\item
Given a consumption profile $s$ in $G$, for every $j \in [n]$ we define $\load{j}\eqdef\sum_{\smash{\type\in\types}}s_j(\type)$ --- the load on (i.e., total consumption from) resource $j$. Furthermore, we define $\height{j}\eqdef f_j(\load{j})$ --- the cost of resource $j$.
\item
A \emph{Nash equilibrium} in $G$ is a consumption profile $s$ s.t.\ for every $\type\in\types$
and for every $k\in\supp\bigl(s(\type)\bigr)$ and $j\in\type$, it is the case that $\height{k}\le\height{j}$.
\end{itemize}
\end{definition}

\begin{example}[A Home Internet / Cellular Market]
Consider a scenario in which the resources are internet service providers (ISPs), and the players are customers on the market for home internet. (Alternatively,
one could think of resources as cellular operators, and of players as customers on the market for cellular service.)
Each customer may choose between the providers available in this customer's geographical area, and would like to get a connection with the largest bandwidth possible given this constraint.
$\mu^{\type}$ in this case is proportional to the amount of customers with possible ISPs $\type$, and for each $j\in[n]$, we choose $f_j$ s.t.\ $\height{j}=f_j(\load{j})$ is inversely proportional to the effective bandwidth of each subscriber of ISP~$j$, when there are $\load{j}$ subscribers to this ISP.

If each ISP has the same total (i.e., overall) bandwidth, then the speed of the connection of a single customer subscribed to an ISP is inversely proportional to this ISP's number of subscribers, and so obtaining the fastest connection possible is equivalent to subscribing to a least-subscribed-to ISP, and so this case is captured by setting $f_j\eqdef\id$ for every $j\in[n]$. Generalizing, we may imagine that, say, some ISPs may have different total bandwidths than others (which may be captured by setting $f_j(x)\eqdef\nicefrac{x}{b_j}$, where $b_j$ is the total bandwidth of ISP~$j$), or that some ISPs may even purchase some additional total bandwidth as their subscriber pool grows; in either scenario, in order to surf with greatest speed, each customer would prefer to subscribe not necessarily to a least-subscribed-to ISP (i.e., one with minimal $\load{j}$), but rather to an ISP from which the customer would receive the fastest connection, i.e., one with minimal~\mbox{$\height{j}=f_j(\load{j})$}.
\end{example}

The study of stability against group deviations was initiated by \cite{Aumann59}, who considered
deviations from which all deviators gain. Recently, the CS literature considers a considerably stronger solution concept, according to which a deviation is considered beneficial even if
only some of the participants in the deviating coalition gain, as long as none of the participants lose \cite[see, e.g.,][]{Rozenfeld06}. While stability against the classical all-gaining coalitional deviation is termed \emph{strong equilibrium}, this more demanding concept is referred to as \emph{super-strong equilibrium}; there are very few results showing its existence in nontrivial settings. We now formally define both concepts.

\begin{definition}[Strong / Super-Strong Nash Equilibrium]
Let $G=\game$ be a
resource selection game and let $s$ be a Nash equilibrium in $G$. For every $\type\in\types$ with $\mu^{\type}>0$, let $h^{\type}\eqdef\height{j}$ for
every~$j\in\supp\bigl(s(\type)\bigr)$. ($h^{\type}$ is well defined by definition of Nash equilibrium.)
\begin{itemize}
\item
$s$ is a \emph{strong Nash equilibrium} if
there exists no consumption profile~$s'\ne s$ s.t.\
for every $\type\in\types$ and $k\in\supp\bigl(s'(\type)\bigr)$ s.t.\ $s'_k(\type)>s_k(\type)$, it is the case that $\heightt{k}< h^{\type}$.\footnote{\label{strong-coalition}The minimal coalition that can cause a deviation from $s$ to $s'$ is the coalition containing, for every type $\type\in\types$ and resource $k\in\supp\bigl(s'(\type)\bigr)$ s.t.\ $s'_k(\type)>s_k(\type)$, a mass of $s'_k(\type)\!-\!s_k(\type)$ players of type $\type$ who consume from $k$ in $s'$ but not in $s$.}
\item
$s$ is a \emph{super-strong Nash equilibrium} if
there exists no consumption profile $s'$ s.t.\
for every $\type\in\types$ and $k\in\supp\bigl(s'(\type)\bigr)$ s.t.\ $s'_k(\type)>s_k(\type)$, it is the case that $\heightt{k}\le h^{\type}$,
with $\heightt{k}<h^{\type}$ for at least one pair of type $\type\in\types$ and resource $k\in\supp\bigl(s'(\type)\bigr)$.\footnote{We require that no member of the minimal coalition described in \cref{strong-coalition} lose, but allow the gaining member to be any player, i.e., even one whose consumption is not necessarily changed. (Indeed, we do not require that $s'_k(\type)>s_k(\type)$ for the pair $\type$ and $k$ for which $\heightt{k}<h^{\type}$.)}
\end{itemize}
\end{definition}

\begin{remark}
Every super-strong Nash equilibrium is a strong Nash equilibrium.
\end{remark}

\subsection{Formal Results}\label{results}

In \cref{definitions,derivation}, we constructively prove the following three \lcnamecrefs{nash-exists} and \lcnamecref{indifference}.

\begin{theorem}[$\exists$ Strong Nash Equilibrium]\label{nash-exists}
Let $G=\game$ be a resource selection game.
If $f_1,\ldots,f_n$ are continuous, then a strong Nash equilibrium exists in $G$.
\end{theorem}

\begin{theorem}[Uniqueness of Nash Equilibrium Resource Costs]\label{indifference-resource-costs}
Let $G$ be an $n$-resource selection game.
$\height{j}=\heightt{j}$ for every $j\in[n]$ and every two Nash equilibria $s,s'$ in $G$.
\end{theorem}

\begin{corollary}\label{indifference} Let $G=\game$ be a resource selection game.
\begin{parts}
\item\thmitemtitle{Players are Indifferent between Nash Equilibria}\label{indifference-players}
$\height{k}=\heightt{k'}$ for every $k\in\supp\bigl(s(\type)\bigr)$ and $k'\in\supp\bigl(s'(\type)\bigr)$, for every $\type\in\types$ and every two Nash equilibria $s,s'$ in $G$.
\item\thmitemtitle{Uniqueness of Nash Equilibrium Resource Loads}\label{indifference-resources}
If no two of $(f_j)_{j=1}^{n}$ share any plateau height, then $\load{j}=\loadt{j}$ for every $j\in[n]$ and every two Nash equilibria $s,s'$ in $G$.
\end{parts}
\end{corollary}

\begin{theorem}[All Nash Equilibria are Strong / Super-Strong]\label{super-strong}
Let $G=\game$ be a resource selection game.
\begin{parts}
\item\label{super-strong-strong}
All Nash equilibria in $G$ are strong.
\item\label{super-strong-super-strong}
If $\height{j}$ is not a plateau height of $f_j$ for each $j\in[n]$ in any/every Nash equilibrium $s$, then all Nash equilibria in $G$ are super-strong.
\end{parts}
\end{theorem}

We significantly generalize all of these results in \cref{compress-expand}.
\cref{nash-exists} significantly strengthens a corollary of \cite{Schmeidler73} that shows existence of a (not-necessarily-strong) Nash equilibrium for continuous cost functions; \cref{nash-exists,super-strong} strengthen corollaries of \cite{Milchtaich04,Milchtaich06} that show, for strictly increasing cost functions, existence of a strong equilibrium and equivalence of super-strong and Nash equilibria. \cref{indifference-resource-costs} strengthens a result of \cite{Beckmann56} that requires differentiability of cost functions, and a result of \cite{Milchtaich05} that requires both continuity and strict monotonicity of these functions. We emphasize that unlike \citeauthor{Schmeidler73}'s and \citeauthor{Milchtaich06}'s proofs, our proof of \cref{nash-exists} allows for explicit calculation of an equilibrium. Indeed, none of our proofs hinge on any fixed-point theorem, on the Minimax theorem or any equivalent result, on the existence of a potential, or on linear programming. (Moreover, as shown in \cref{hall-and-beyond}, our results can even replace nonconstructive techniques such as linear programming in certain problems that are traditionally viewed as unrelated to games.) Similarly, unlike \citeauthor{Beckmann56}'s and \citeauthor{Milchtaich05}'s proofs, our proof of \cref{indifference-resource-costs} gives an explicit formula for $\height{j}$ for every~$j\in[n]$ (proving the \lcnamecref{indifference-resource-costs} by noting that this formula does not depend on $s$). See \cref{discussion} for a discussion of the benefits of such explicit formulations.

\subsection{Construction and Hydraulic Intuition}\label{hydraulic-intuition}

In this \lcnamecref{hydraulic-intuition}, we intuitively survey the construction underlying our results, as a prelude to the formal analysis given in \cref{definitions,derivation}.
We start with the special case in which $f_j=\id$ for every~$j\in[n]$, i.e., $\height{j}=\load{j}$ for every $j\in[n]$ and consumption profile~$s$. Our hydraulic construction for this case, from which our analysis draws intuition, consists of a system of containers, interconnected balloons, and pistons, which is illustrated in
\cref{\balloonswithanimation}.\footnote{This construction significantly generalizes an \emph{ad hoc} construction that appears as a secondary auxiliary result in a previous discussion paper by the authors \citep{noncoop-market-alloc}. That discussion paper deals with combining a highly restricted form of resources selection games (with degenerate strategy sets, i.e., where $\mu^{\type}=0$ for all sets $\type$ but those of a very specific form) with facility location games, drawing conclusions regarding the possibility of false appearance of collusion in internet markets and various food markets.}\textsuperscript{,}\footnote{This construction may be thought of, in some sense, as a continuous counterpart to the dual greedy algorithm of \cite{HHKS13}, with hydraulic dynamics replacing their ``packing oracle''. As we show below, our hydraulic approach provides important intuition and novel insights (e.g., into the potential), as well as paves the way for proving a gamut of novel results.}%
\begin{figure}[p]
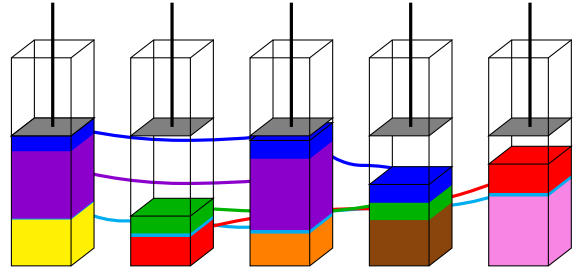
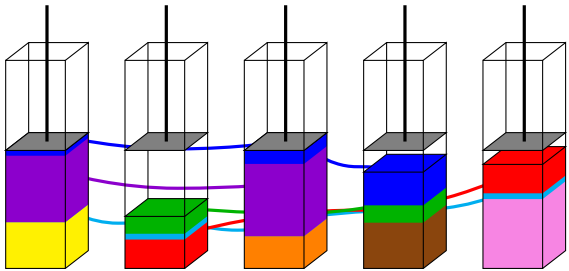
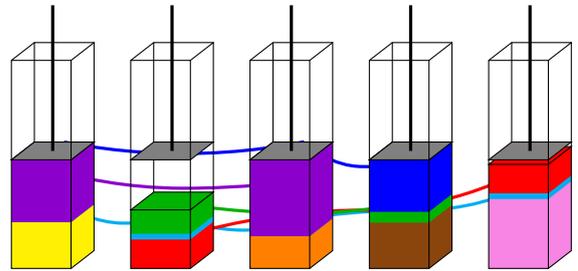
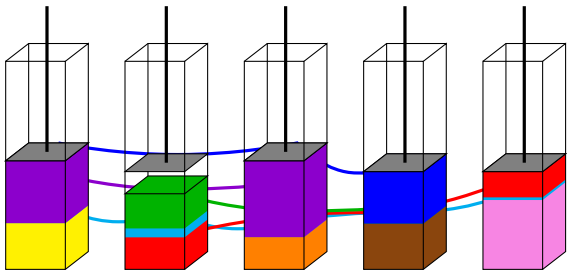
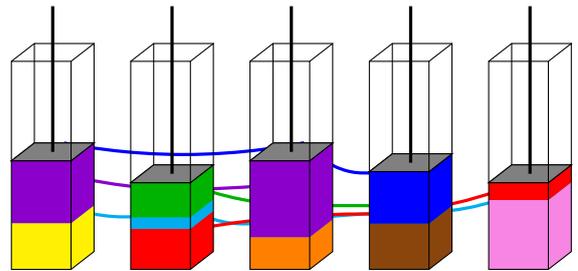
%
\centering%
\subfigure[A set of $5$ open-top hollow box containers, corresponding, from left to right, to resources $1,\ldots,5$, respectively. For each player type $\type$ with $\mu^{\type}>0$ (each such type is assigned a distinct color in the illustration), a balloon, or plastic bag, is placed in each container $j\in\type$. Balloons corresponding to the same type $\type$
are connected via a thin tube emerging from a narrow slit (not shown) running vertically along the back of each container, and are jointly filled with $\mu^{\type}$ liquid.]{%
\balloonsstatic{0}%
}\qquad\quad
\subfigure[Pistons are simultaneously lowered through the top sides of all the containers. As the piston in the first container reaches the balloons in this container, they are compressed, causing the balloons
connected to them (i.e., the purple balloon in the third container, the blue balloons in the third and fourth containers, and the light blue balloons in the second, third, and fifth containers) to inflate.]{%
\balloonsstatic{228}%
}
\subfigure[As the piston in the third container reaches the balloons in this container, they start to compress as well, causing, e.g., the interconnected blue balloon in the fourth container to inflate even faster.]{%
\balloonsstatic{278}%
}\qquad\quad
\subfigure[At a certain point in time, no balloon in the first or third containers can be compressed any further, as all the liquid in these containers that could have escaped to other containers has been depleted. The pistons in the first and third containers halt, and the remaining pistons continue their descent.]{%
\label{balloons-first-third}%
\balloonsstatic{317}%
}
\subfigure[At some later point in time, no balloon in the fourth container can be compressed any further, as all the liquid in this container that could have escaped to any container other
than the first or the third ones has been depleted.]{%
\balloonsstatic{362}%
}\qquad\quad
\subfigure[Eventually, no balloon in the second or fifth containers can be compressed any further, and the process concludes.]{%
\label{balloons-final}%
\balloonsstatic{408}%
}%
\caption{%
(See \if\withanimation1\cref{balloons-anim}\else the online paper\fi\ for an animated version\if\withanimation1, which unfortunately cannot be printed\fi.)
Illustration of the construction underlying our analysis, for $n=5$ and for $f_j=\id$ for every $j\in[5]$. E.g., as exactly~$87.5\%$ of the red liquid in \cref{balloons-final} is in the second container and the remaining $12.5\%$ is in the fifth container, the strategy for the player type corresponding to the red color (i.e., $\type=\{2,5\}$) in the (super-)strong Nash equilibrium that we construct is $0.875\cdot\mu^{\{2,5\}}$ consumption from resource~$2$ and $0.125\cdot\mu^{\{2,5\}}$ consumption from resource~$5$; similarly, as all of the blue liquid is in the fourth container, the strategy for the ``blue'' type ($\{1,3,4\}$) in this equilibrium is $\mu^{\{1,3,4\}}$ consumption, solely from resource~$4$.%
}%
\label{balloons}%
\end{figure}\if\withanimation1\begin{figure}[ht]%
\centering%
\balloonsanim%
\caption{%
Animated version of \cref{balloons}; requires Adobe Reader.
Click the $\triangleright$ button to start the animation.
}%
\label{balloons-anim}%
\end{figure}\fi{}
The intuition underlying our results draws from a number of key \lcnamecrefs{obs-min-liquid} regarding this construction (we generalize and formalize these \lcnamecrefs{obs-min-liquid} in \cref{definitions,derivation}):
\begin{observations}
\item\label{obs-min-liquid}
If the pistons in a set $S$ (e.g., $S=\{1,3\}$ or $S=\{4\}$) of containers stop simultaneously, then at the time of their stopping, no liquid under them can escape to any container in which
the piston has not yet stopped (or else it would do so and the piston above it would not stop).
\item\label{obs-nash}
By \cref{obs-min-liquid}, and as pistons that stop
later in time stop at a lower height,  in the resulting consumption profile
no player type has any incentive to deviate, and so it is indeed a Nash equilibrium.
\item\label{obs-any-nash}
If we initially distribute the liquid of each ``color'' (among the various balloons corresponding to this color) according to some Nash equilibrium (e.g., if we initially distribute the liquid as in \cref{balloons-final}), then the
liquid distribution would not change during the entire process of descent of the pistons. Therefore, each Nash equilibrium may be attained from some
initial liquid distribution.
\item\label{obs-start-over}
After pistons $1$ and $3$ (in \cref{balloons}) stop, we effectively start over, solving a $3$-resource ($2,4,5$) selection game between all player types whose original acceptable resources were not merely resource $1$ and/or $3$.
\item\label{obs-indifference}
In \cref{balloons}, pistons $1$ and $3$ are the earliest to stop.
By \cref{obs-min-liquid} above, no part of the liquid under these pistons when they stop can ever, regardless of the initial liquid distribution, end up in any container other than $1$ or $3$. Therefore, these pistons always stop having under them at least the liquid that is under them in \cref{balloons-first-third}, and accordingly at least at the height at which they stop in \cref{balloons-first-third}. By the same \lcnamecref{obs-min-liquid},
the pistons stopping earliest always stop having under them solely liquid that cannot escape to any other container, and so,
regardless of the initial liquid distribution, if this set were not pistons~$1$ and~$3$, then it would stop below
the stopping height of pistons $1$ and~$3$. Therefore, pistons $1$ and $3$ always stop earliest, and at the same height.
Using \cref{obs-start-over}, an inductive argument can show that the height at which each piston stops (and the stopping order) is independent of the initial liquid distribution, and so by \cref{obs-any-nash}, $\height{j}$ for every $j\in[n]$ is independent of the choice of Nash equilibrium.
Furthermore, by the same argument, each player
always consumes from resources with the same $\height{j}$, independently of the choice of Nash equilibrium~$s$.
\end{observations}

We note that while the final piston heights (i.e., values of $\height{j}$) are independent of the initial distribution of liquid among connected balloons (i.e., of the choice of Nash equilibrium $s$),
the final liquid distribution (i.e., players' strategies) is not; in \cref{balloons-final}, e.g., any amount of light blue liquid may be transferred from the second to the fifth container ``in exchange for'' an identical amount of red liquid.\footnote{While this nonuniqueness seemingly contradicts a uniqueness theorem of \cite{OrdaRomShimkin93}, we note that their setting in fact differs from ours; they deal with finitely many players with splittable demand, while we deal with a continuum of nonatomic players. We furthermore note that their analysis requires that the cost functions~$f_j$ be strictly increasing; in their setting, this guarantees uniqueness of both equilibrium loads $\load{j}$ and consumptions $s$, while in our setting, this assumption guarantees solely the uniqueness of equilibrium loads~$\load{j}$ (without this assumption, we have only uniqueness of equilibrium resource costs $\height{j}$); see \crefpart{indifference}{resources}.}

For the general case of arbitrary $f_j$, we intuitively think of replacing the $j$th box container, for every $j\in[n]$, with a container shaped so that whenever it is filled with any amount $\mu_j\in\Rge$ of liquid, the resulting surface level would be precisely $f_j(\mu_j)$. See \cref{vessels-odd-shapes} for an illustration.%
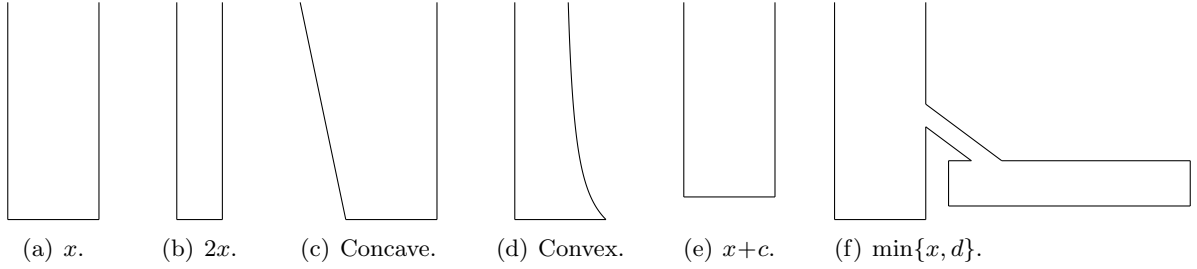
\begin{figure}[ht]%
\centering%
\subfigure[$x$.]{%
\begin{tikzpicture}[yscale=0.6,xscale=1.2]
\draw[shift={(0*1.2,0)}] (0,0.2) -- (0,5);
\draw[shift={(0*1.2,0)}] (0,0.2) -- (1,0.2);
\draw[shift={(0*1.2,0)}] (1,0.2) -- (1,5);
\end{tikzpicture}%
}\qquad
\subfigure[$2x$.]{%
\begin{tikzpicture}[yscale=0.6,xscale=1.2]
\draw[shift={(1*1.2,0)}] (0,0.2) -- (0,5);
\draw[shift={(1*1.2,0)}] (0,0.2) -- (.5,0.2);
\draw[shift={(1*1.2,0)}] (.5,0.2) -- (.5,5);
\draw[shift={(1*1.2,0)},white] (-.05,0.2) -- (-.2,0.2);
\draw[shift={(1*1.2,0)},white] (0.55,0.2) -- (.7,0.2);
\end{tikzpicture}%
}\qquad
\subfigure[Concave.]{%
\begin{tikzpicture}[yscale=0.6,xscale=1.2]
\draw[shift={(2*1.2,0)}] (0,0.2) -- (-.5,5);
\draw[shift={(2*1.2,0)}] (0,0.2) -- (1,0.2);
\draw[shift={(2*1.2,0)}] (1,0.2) -- (1,5);
\end{tikzpicture}%
}\qquad
\subfigure[Convex.]{%
\begin{tikzpicture}[yscale=0.6,xscale=1.2]
\draw[shift={(3*1.2,0)}] (0,0.2) -- (0,5);
\draw[shift={(3*1.2,0)}] (0,0.2) -- (1,0.2);
\draw[shift={(3*1.2,0)}] plot[variable=\t,samples=100,domain=1:5.8] ({.5+.5/\t,0.2+\t-1});
\draw[shift={(3*1.2,0)},white] (-.05,0.2) -- (-.2,0.2);
\draw[shift={(3*1.2,0)},white] (1.05,0.2) -- (1.2,0.2);
\end{tikzpicture}%
}\qquad
\subfigure[$x+c$.]{%
\begin{tikzpicture}[yscale=0.6,xscale=1.2]
\draw[shift={(4*1.2,0)},white] (0,0.2) -- (0,5);
\draw[shift={(4*1.2,0)}] (0,0.7) -- (0,5);
\draw[shift={(4*1.2,0)}] (0,0.7) -- (1,0.7);
\draw[shift={(4*1.2,0)}] (1,0.7) -- (1,5);
\end{tikzpicture}%
}\qquad
\subfigure[$\min\{x,d\}$.~~~~~~~~~~~~~~~~~~~~~~~~~~~]{%
\label{vessels-odd-shapes-min}%
\begin{tikzpicture}[yscale=0.6,xscale=1.2]
\minvessel{5*1.2}{1}
\draw[shift={(5*1.2,0)}] (0,0.2) -- (1,0.2);
\end{tikzpicture}%
}%
\caption{Containers corresponding to various functions $f_j$.
(Assume that the container to the right of the vessel depicted in \cref{vessels-odd-shapes-min} is large enough so as to never fill up, yet may only be occupied by balloons as long as the piston does not pass the tube connecting this container to the main vessel.)}%
\label{vessels-odd-shapes}%
\end{figure}
We emphasize that while the actual construction of such vessels may require that the cost functions $f_j$ meet certain differentiability conditions, our formal proof of \cref{nash-exists} only requires continuity of the cost functions, while our formal proofs of \cref{indifference-resource-costs,indifference,super-strong} do not require even that.
We note that this continuity assumption (in \cref{nash-exists}) is in fact not superfluous; indeed, if even one of the cost functions is discontinuous, then a
Nash equilibrium need not necessarily exist; see \cref{no-nash} for an example and an illustration.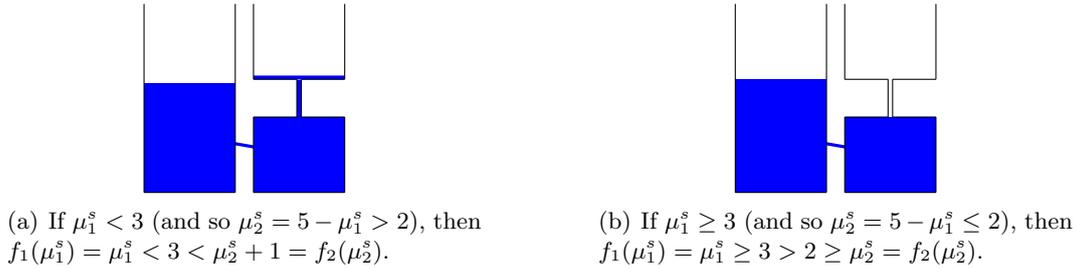
\begin{figure}[ht]%
\centering%
\subfigure[If $\load{1}<3$ (and so $\load{2}=5-\load{1}>2$), then $f_1(\load{1})=\load{1}<3<\load{2}+1=f_2(\load{2})$.]{%
\begin{tikzpicture}[yscale=.5,xscale=1.2]
\pgfmathsetmacro{\containerheight}{5}
\draw[blue,very thick] (0.5,1.5) -- (1.7,1);
\draw[draw=none] (-1.5,0) -- (-1,0);
\draw[draw=none,shift={(1*1.2+1,0)}] (1.5,0) -- (1,0);
\draw[draw=none,fill=blue,shift={(0*1.2,0)}] (0,0) -- (1,0) -- (1,2.9) -- (0,2.9) -- cycle;
\draw[draw=none,fill=blue,shift={(1*1.2,0)}] (0,0) -- (1,0) -- (1,2) -- (0.525,2) -- (0.525,3) -- (0,3) -- (0,3.1) -- (1,3.1) -- (1,3) -- (0.475,3) -- (0.475,2)-- (0,2) -- cycle;
\draw[shift={(0*1.2,0)}] (0,0) -- (0,\containerheight);
\draw[shift={(0*1.2,0)}] (1,0) -- (1,\containerheight);
\draw[shift={(0*1.2,0)}] (0,0) -- (1,0);
\draw[shift={(1*1.2,0)}] (0,0) -- (0,2);
\draw[shift={(1*1.2,0)}] (1,0) -- (1,2);
\draw[shift={(1*1.2,0)}] (0,0) -- (1,0);
\draw[shift={(1*1.2,0)}] (0,2) -- (.475,2);
\draw[shift={(1*1.2,0)}] (1,2) -- (.525,2);
\draw[shift={(1*1.2,0)}] (.475,2) -- (.475,3);
\draw[shift={(1*1.2,0)}] (.525,2) -- (.525,3);
\draw[shift={(1*1.2,0)}] (0,3) -- (0,\containerheight);
\draw[shift={(1*1.2,0)}] (1,3) -- (1,\containerheight);
\draw[shift={(1*1.2,0)}] (0,3) -- (.475,3);
\draw[shift={(1*1.2,0)}] (1,3) -- (.525,3);
\end{tikzpicture}%
}\qquad\qquad
\subfigure[If $\load{1}\ge3$ (and so $\load{2}=5-\load{1}\le2$), then $f_1(\load{1})=\load{1}\ge3>2\ge\load{2}=f_2(\load{2})$.]{%
\begin{tikzpicture}[yscale=.5,xscale=1.2]
\pgfmathsetmacro{\containerheight}{5}
\draw[blue,very thick] (0.5,1.5) -- (1.7,1);
\draw[draw=none] (-1.5,0) -- (-1,0);
\draw[draw=none,shift={(1*1.2+1,0)}] (1.5,0) -- (1,0);
\draw[draw=none,fill=blue,shift={(0*1.2,0)}] (0,0) -- (1,0) -- (1,3) -- (0,3) -- cycle;
\draw[draw=none,fill=blue,shift={(1*1.2,0)}] (0,0) -- (1,0) -- (1,2) -- (0,2) -- cycle;
\draw[shift={(0*1.2,0)}] (0,0) -- (0,\containerheight);
\draw[shift={(0*1.2,0)}] (1,0) -- (1,\containerheight);
\draw[shift={(0*1.2,0)}] (0,0) -- (1,0);
\draw[shift={(1*1.2,0)}] (0,0) -- (0,2);
\draw[shift={(1*1.2,0)}] (1,0) -- (1,2);
\draw[shift={(1*1.2,0)}] (0,0) -- (1,0);
\draw[shift={(1*1.2,0)}] (0,2) -- (.475,2);
\draw[shift={(1*1.2,0)}] (1,2) -- (.525,2);
\draw[shift={(1*1.2,0)}] (.475,2) -- (.475,3);
\draw[shift={(1*1.2,0)}] (.525,2) -- (.525,3);
\draw[shift={(1*1.2,0)}] (0,3) -- (0,\containerheight);
\draw[shift={(1*1.2,0)}] (1,3) -- (1,\containerheight);
\draw[shift={(1*1.2,0)}] (0,3) -- (.475,3);
\draw[shift={(1*1.2,0)}] (1,3) -- (.525,3);
\end{tikzpicture}%
}%
\caption{No Nash equilibrium exists when $n=2$, $f_1=\id$, $f_2(x)=(x\!>\!2\;?\;x\!+\!1 : x)$, $\mu^{\{1,2\}}=5$ and $\mu^{\{1\}}=\mu^{\{2\}}=0$.
(Assume that the tube connecting the two parts of the second vessel is of zero volume.)}%
\label{no-nash}%
\end{figure}

\subsection{Definitions for Formalizing the Observations from Section~\refintitle{hydraulic-intuition}}\label{definitions}

Building upon the intuition of \cref{hydraulic-intuition}, we formally derive the results of \cref{results} in \cref{derivation}, with proofs in \cref{proofs-derivation,proofs-results}. In this \lcnamecref{definitions}, we review the formal definitions underlying this derivation. Full proofs of all claims given below are given in \cref{proofs-derivation}.

\subsubsection{Communicating-Vessel Equalization}

Let $S$ be the set of pistons stopping earliest during the process depicted in \cref{balloons}. Assume that when these pistons stop, the total amount of liquid in the respective containers is $\mu$. At what height did the pistons stop? In this section we formalize the answer to this question.

\begin{definition}[Nondecreasing Function to $\Rundefined$]\label{nondecreasing}
Let $f:\Rge\rightarrow\Rundefined$. We say that $f$ is \emph{nondecreasing} if
$f|_{f^{-1}(\mathbb{R})}$ is nondecreasing; i.e., if for every $\mu<\mu'\in\Rge$, if both $f(\mu)\in\mathbb{R}$
and $f(\mu')\in\mathbb{R}$, then $f(\mu)\le f(\mu')$.
\end{definition}

\begin{definition}[Communicating-Vessel Equalization]\label{equalize}
Let $m\in\mathbb{N}$ and let $f_1,\ldots,f_m:\Rge\rightarrow\Rundefined$ be nondecreasing functions.\footnote{We allow the functions $f_1,\ldots,f_m$ to assume the value $\undefined$ for technical reasons that become apparent in \cref{intermediate-equalize} below. The reader may intuitively think of $f_1,\ldots,f_m$ as real functions until reaching that \lcnamecref{intermediate-equalize}.} We define the function $\equalize{f_1,\ldots,f_m}:\Rge\rightarrow\Rundefined$ by
\[
\mu\mapsto
\begin{cases}
f_1(\mu_1) & \exists\mu_1,\ldots,\mu_m\in\Rge:\sum_{j=1}^{m}\mu_j=\mu \And f_1(\mu_1)=f_2(\mu_2)=\cdots=f_m(\mu_m)\in\mathbb{R} \\
\undefined & \mbox{otherwise.}
\end{cases}
\]
\end{definition}

\begin{remark}[Equalizing Multiple Identical Functions]\label{symmetric-equalize}
If $f_1=f_2=\cdots=f_m$ and this function is defined on all $\Rge$, then $\equalize{f_1,\ldots,f_m}(\mu)=f_1\bigl(\frac{\mu}{m}\bigr)$.
\end{remark}

For $f_1,\ldots,f_m:\Rge\rightarrow\mathbb{R}$, one may intuitively think of $\equalize{f_1,\ldots,f_m}(\mu)$ as exactly the answer to the question raised above:
if $f_1,\ldots,f_m$ are the functions corresponding (see \cref{vessels-odd-shapes}) to the containers of the pistons stopping earliest during the process depicted in \cref{balloons}, and if the total amount of liquid in the respective containers when these pistons stop is $\mu$, then $\equalize{f_1,\ldots,f_m}(\mu)$ is the height at which these pistons stop; $\equalize{f_1,\ldots,f_m}(\mu)=\undefined$ if
it is impossible that all these pistons simultaneously stop when the total amount of liquid in these containers is $\mu$. Alternatively and equivalently,
if empty containers corresponding (see \cref{vessels-odd-shapes}) to $f_1,\ldots,f_m$ are connected at their base and the resulting system of communicating vessels is jointly filled with $\mu$ liquid, then $\equalize{f_1,\ldots,f_m}(\mu)$ is the resulting liquid surface level; see \cref{vessels-equalize}
for an illustration.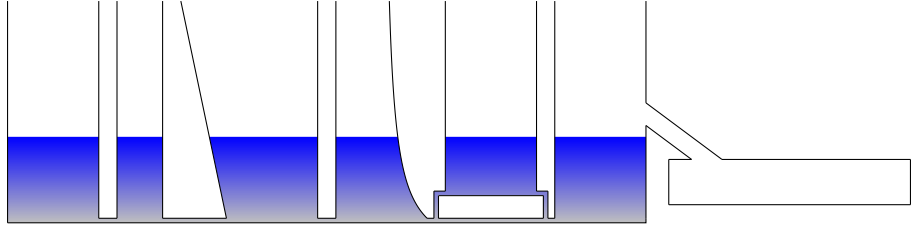
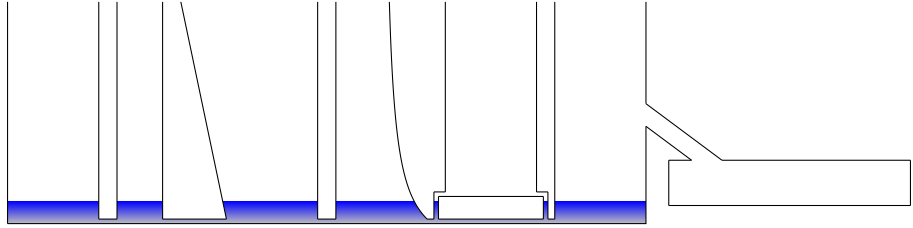
\begin{figure}[ht]%
\centering%
\subfigure[$\equalize{f_1,\ldots,f_6}(\mu)$ equals the liquid surface level when the containers are jointly filled with $\mu$ liquid.]{%
\begin{tikzpicture}[yscale=0.6,xscale=1.2]
\shade[top color=blue,bottom color=gray!50]
(0,0.1) -- (0,2) -- (1,2) -- (1,0.2) --
(1.2,0.2) -- (1.2,2) -- ({1.2+.5},2) -- ({1.2+.5},0.2) --
(2.4,0.2) -- ({2.4-.5/4.8*1.8},2) -- ({2.4+1},2) -- ({2.4+1},0.2) --
(3.6,0.2) -- (3.6,2) -- plot[variable=\t,samples=100,domain=2.8:1] ({3.6+.5+.5/\t,0.2+\t-1}) -- ({3.6+1},0.2) -- (3.6+1.075,0.2) --
(4.8-.125,0.8) -- (4.8,0.8) -- (4.8,2) -- (4.8+1,2) -- ({4.8+1},0.8) -- (4.8+1.125,0.8) -- (4.8+1.125,0.2) -- (4.8+1.2,0.2) --
(6,0.2) -- (6,2) -- (6+1,2) -- (6+1,0.1) --
cycle;
\draw[draw=none,fill=white,shift={(4*1.2,0)}] (-0.075,0.2) -- (-0.075,0.7) -- (1.075,0.7) -- (1.075,.2) -- cycle;
\oddshapeoutline
\end{tikzpicture}%
}
\subfigure[$\equalize{f_1,\ldots,f_6}(\mu)=\undefined$, as no distribution of $\mu$ liquid between the containers results in an even liquid surface level across all containers (recall that the fifth
container corresponds to the function $x+c$ for some constant $c$, and therefore if it is empty, then its liquid surface level is defined as the level $c$ of its bottom side).]{%
\label{vessels-equalize-undefined}%
\begin{tikzpicture}[yscale=0.6,xscale=1.2]
\shade[top color=blue,bottom color=gray!50]
(0,0.1) -- (0,0.6) -- (1,0.6) -- (1,0.2) --
(1.2,0.2) -- (1.2,0.6) -- ({1.2+.5},0.6) -- ({1.2+.5},0.2) --
(2.4,0.2) -- ({2.4-.5/4.8*.4},0.6) -- ({2.4+1},0.6) -- ({2.4+1},0.2) --
(3.6,0.2) -- (3.6,0.6) -- plot[variable=\t,samples=100,domain=1.4:1] ({3.6+.5+.5/\t,0.2+\t-1}) -- ({3.6+1},0.2) -- (3.6+1.075,0.2) --
(4.8-.125,0.6) -- (4.8-0.075,0.6) -- (4.8-0.075,0.2) -- (4.8+1.075,0.2) -- (4.8+1.075,0.6) -- (4.8+1.125,0.6) -- (4.8+1.125,0.2) -- (4.8+1.2,0.2) --
(6,0.2) -- (6,0.6) -- (6+1,0.6) -- (6+1,0.1) --
cycle;
\oddshapeoutline
\end{tikzpicture}%
}%
\caption{Equalizing the functions from \cref{vessels-odd-shapes}; assume that the connecting tubes are of zero volume.}%
\label{vessels-equalize}%
\end{figure}

When two of the functions $f_1,\ldots,f_m$ share a plateau height (cf.\ \crefpart{indifference}{resources}), then the liquid distribution~$\mu_1,\ldots,\mu_m$ may not be well defined; see \cref{vessels-equalize-not-increasing}
for an illustration.\begin{figure}[ht]%
\centering%
\subfigure[]{%
\begin{tikzpicture}[yscale=0.6,xscale=.969]
\shade[top color=blue,bottom color=gray!50]
(-1.1,.1) -- (-1.1,2.25) -- (-.1,2.25) -- (-.1,.2) -- (.1,.2) -- (.1,2.25) -- (1.1,2.25) -- (1.1,.1) -- cycle;
\shade[top color=blue,bottom color=gray!50]
(-4,.5) -- (-4,1.25) -- (-1.35,1.25) -- (-1.35,.5) -- cycle;
\shade[top color=blue,bottom color=gray!50]
(4,.5) -- (4,0.75) -- (1.35,0.75) -- (1.35,.5) -- cycle;
\minvessel{0.1}{-1}
\minvessel{0.1}{1}
\draw (-1.1,.2) -- (-1.1,.1) -- (1.1,.1) -- (1.1,.2);
\draw (-.1,.2) -- (.1,.2);
\end{tikzpicture}%
}\quad
\subfigure[]{%
\begin{tikzpicture}[yscale=0.6,xscale=.969]
\shade[top color=blue,bottom color=gray!50]
(-1.1,.1) -- (-1.1,2.25) -- (-.1,2.25) -- (-.1,.2) -- (.1,.2) -- (.1,2.25) -- (1.1,2.25) -- (1.1,.1) -- cycle;
\shade[top color=blue,bottom color=gray!50]
(-4,.5) -- (-4,0.65) -- (-1.35,0.65) -- (-1.35,.5) -- cycle;
\shade[top color=blue,bottom color=gray!50]
(4,.5) -- (4,1.35) -- (1.35,1.35) -- (1.35,.5) -- cycle;
\minvessel{0.1}{-1}
\minvessel{0.1}{1}
\draw (-1.1,.2) -- (-1.1,.1) -- (1.1,.1) -- (1.1,.2);
\draw (-.1,.2) -- (.1,.2);
\end{tikzpicture}%
}%
\caption{Equalization of two copies of the function from \cref{vessels-odd-shapes-min}, via two distinct liquid distributions.
Formally, when~$\mu>2d$, there exists a continuum of pairs $\mu_1,\mu_2\in\Rge$ s.t.\ $\mu_1+\mu_2=\mu$ and $\min\{\mu_1,d\}=\min\{\mu_2,d\}$.
For all such $\mu_1,\mu_2$, it is nonetheless always the case that $\min\{\mu_1,d\}=d=\min\{\mu_2,d\}$, and so $\equalize{f_1,f_2}(\mu)=d$ is well defined.}%
\label{vessels-equalize-not-increasing}%
\end{figure}
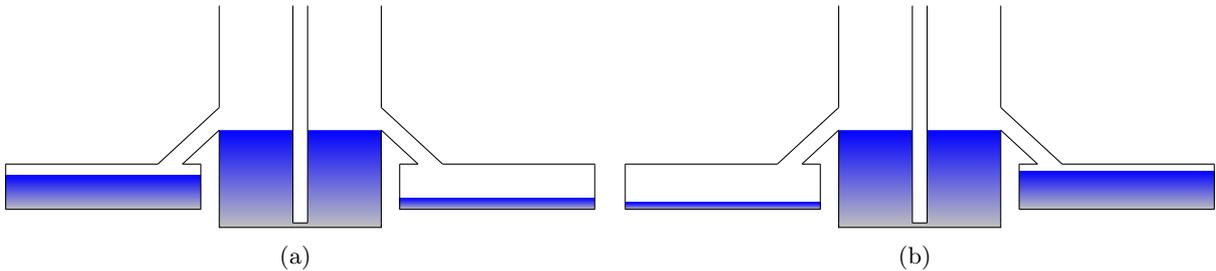
Nonetheless, the following \lcnamecref{equalize-well-defined-nondecreasing} shows that
the resulting surface level $\equalize{f_1,\ldots,f_m}$ is well defined, i.e., independent of the chosen liquid distribution $\mu_1,\ldots,\mu_m$.

\begin{lemma}[Equalization is Well Defined and Nondecreasing]\label{equalize-well-defined-nondecreasing}
Let $m\in\mathbb{N}$ and let $f_1,\ldots,f_m:\Rge\rightarrow\Rundefined$ be nondecreasing functions.
$\equalize{f_1,\ldots,f_m}$ is a well-defined
non\-de\-creasing function from $\Rge$ to $\Rundefined$.
\end{lemma}

In the remainder of this section, we derive some additional properties of communicating-vessel equalization.
The following \lcnamecref{intermediate-equalize} notes that ``connecting a single vessel with itself'' has no effect, while connecting several vessels may be done by first connecting subsets
of these vessels into ``intermediate vessels'', and only then connecting all ``intermediate vessels'' together; it is for the sake of the latter that we have allowed the functions
$f_1,\ldots,f_m$ in \cref{equalize} to assume the value $\undefined$.

\begin{lemma}[Composition of Equalizations]\label{intermediate-equalize}
Let $m\in\mathbb{N}$ and let $f_1,\ldots,f_m:\Rge\rightarrow\Rundefined$ be nondecreasing functions.
\begin{parts}
\item\label{intermediate-equalize-idempotent}
$\equalize{f_1}\equiv f_1$.
\item\label{intermediate-equalize-compose}
$\equalize{f_1,\ldots,f_m}\equiv\equalize{\equalize{f_1,\ldots,f_{j_1}},\equalize{f_{{j_1}+1},\ldots,f_{j_2}},\ldots,\equalize{f_{{j_k}+1},\ldots,f_m}}$, for every $k\in[m]$ and $1\le j_1<j_2<\cdots<j_k<m$.
\end{parts}
\end{lemma}

Recall that \cref{nash-exists} requires that $f_1,\ldots,f_n$ be continuous. (An example in which one of these functions is discontinuous and no Nash equilibrium exists was given in \cref{no-nash}.)
We therefore conclude this section with an analysis of the equalization of continuous functions.

\begin{definition}[Function to $\Rundefined$: Continuous / Defined on a Suffix
of $\Rge$]
Let $f:\Rge\rightarrow\Rundefined$.
\begin{itemize}
\item
We say that $f$ is \emph{continuous} if $f|_{f^{-1}(\mathbb{R})}:f^{-1}(\mathbb{R})\rightarrow\mathbb{R}$ is continuous.
\item
We say that $f$ is \emph{defined on a suffix of $\Rge$} if for every $\mu<\mu'\in\Rge$, if $f(\mu)\in\mathbb{R}$, then $f(\mu')\in\mathbb{R}$ as well.
\end{itemize} 
\end{definition}

\begin{lemma}[Equalization of Continuous Functions]\label{equalize-continuous}
Let $m\in\mathbb{N}$ and let $f_1,\ldots,f_m:\Rge\rightarrow\Rundefined$ be nondecreasing functions.
\begin{parts}
\item\label{equalize-continuous-continuous}
If at least one of $f_1,\ldots,f_m$ is continuous, then $\equalize{f_1,\ldots,f_m}$ is continuous.
\item\label{equalize-continuous-suffix}
If each of $f_1,\ldots,f_m$ is continuous and defined on a suffix of $\Rge$, then $\equalize{f_1,\ldots,f_m}$ is continuous and defined on a suffix of $\Rge$ as well.
\end{parts}
\end{lemma}

\begin{remark}[Equalization of Lipschitz Functions]
Let $m\in\mathbb{N}$ and let $f_1,\ldots,f_m:\Rge\rightarrow\Rundefined$ be nondecreasing functions.
A proof virtually identical to that of \crefpart{equalize-continuous}{continuous} can be used to show that if at least one of $f_1,\ldots,f_m$ is Lipschitz, then $\equalize{f_1,\ldots,f_m}$ is Lipschitz with the same Lipschitz constant.
\end{remark}

The following \lcnamecref{equalize-same-bottom} shows that for continuous real functions, the only ``reason'' for their equalization to be undefined is of the type depicted in 
\cref{vessels-equalize-undefined}, i.e., an uneven bottom of the corresponding containers.

\begin{corollary}\label{equalize-same-bottom}
Let $m\in\mathbb{N}$ and let $f_1,\ldots,f_m:\Rge\rightarrow\mathbb{R}$ be nondecreasing continuous real functions.
$\equalize{f_1,\ldots,f_m}$ is a real function iff $f_1(0)=f_2(0)=\cdots=f_m(0)$.
\end{corollary}

\subsubsection{An Explicit Formula for the Highest-Costing Resources and their Cost}

Following the discussion in the previous section, if the set of highest-costing resources (highest-stopping pistons) is $P$, then by \cref{obs-min-liquid} from \cref{hydraulic-intuition}, we expect them to cost (stop at height) $E_G(P)$, where $E_G$ is defined as follows.

\begin{definition}[$E_G$]
Let $G=\game$ be a resource selection game. We define \mbox{$\displaystyle E_G(S)\eqdef\equalize{f_k:k\in S}\Bigl(\smashoperator{\sum_{\type\in\nonemptysubs{S}}}\mu^{\type}\Bigr)$}, for every $S\in\types$.
\end{definition}

A main challenge that remains before moving on to prove the results of \cref{results}, therefore, is to find an expression for $P$, as given such an expression, we could find $E_G(P)$ and proceed inductively via a formalization of \cref{obs-start-over} from \cref{hydraulic-intuition}. A natural first candidate for the role of $P$ may be to take a set of resources with maximal~$E_G$, i.e., some element of $\arg\Max_{\smash{S\in\types}} E_G(S)$, where the value $\undefined$ is here and henceforth treated as~$-\infty$ for comparisons by the $\Max$ operator. Noticing that for many natural choices of resource selection games $G=\game$, the set of all such maximizing sets of resources $\arg\Max_{\smash{S\in\types}} E_G(S)$ is closed under set union, and therefore contains a greatest element (namely, $\bigcup\arg\Max_{\smash{S\in\types}} E_G(S)$), a natural candidate for the role of $P$ would be this greatest element. While for many natural choices of resource selection games $G=\game$, this greatest element indeed exists and coincides with the set of highest-costing resources (indeed, this is the case when all $f_j$ are strictly increasing --- see below), this need not generally be the case. To see that these need not coincide, consider the following example.

\begin{example}[$\bigcup\arg\Max_{\smash{S\in\types}}E_G(S)$ is Not the Set of Highest-Costing Resources]\label{need-mg} Consider the game $G=\game$, for $n=2$, $f_1=\id$, $f_2(x)=\min\{x,2\}$, $\mu^{\{1\}}=1$, $\mu^{\{2\}}=3$, and $\mu^{\{1,2\}}=0$. In this game, consumption of each player type $\{i\}$ solely from resource $i$ is the unique consumption profile and hence the unique Nash
equilibrium --- denote it by~$s$. Note that $\height{1}=1$ and $\height{2}=2$, and so $2$ is the unique highest-costing
resource, albeit $\bigcup\arg\Max_{\smash{S\in\types}}E_G(S)=\bigcup\bigl\{\{2\},\{1,2\}\bigr\}=\{1,2\}$. Indeed, in this case while the set of highest-costing resources $P=\{2\}$ is an element in $\arg\Max_{\smash{S\in\types}} E_G(S)=\bigl\{\{2\},\{1,2\}\bigr\}$, it is not the greatest element of this set.
\end{example}

In fact, as shown in the following \lcnamecref{no-max-without-mg}, the set $\arg\Max_{\smash{S\in\types}} E_G(S)$ need not even contain a greatest element.

\sloppy
\begin{example}[$\arg\Max_{\smash{S\in\types}}E_G(S)$ Has No Greatest Element]\label{no-max-without-mg} Consider the game $G=\game$, for $n=3$, $f_1=f_2=\id$, $f_3(x)=\min\{x,2\}$, $\mu^{\{1\}}=\mu^{\{2\}}=1$, $\mu^{\{3\}}=3$, and $\mu^{\type}=0$ for all nonsingleton $\type\in\types$. It is easy to verify that 
$\arg\Max_{\smash{S\in\types}}E_G(S)=\bigl\{\{3\},\{1,3\},\{2,3\}\bigr\}$, and so this set contains no greatest element.
\end{example}
\fussy

Removing any hope of representing $P$ using some other function of $\arg\Max_{\smash{S\in\types}} E_G(S)$,
the following \lcnamecref{not-enough-info-without-mg} shows that the set of highest-costing resources cannot be inferred from $\arg\Max_{\smash{S\in\types}}E_G(S)$ alone.

\begin{example}[$\arg\Max_{\smash{S\in\types}}E_G(S)$ Does Not Determine the Set of Highest-Costing Resources]\label{not-enough-info-without-mg} Consider the game $G=\game$, for $n=2$, $f_1=f_2=\id$, $\mu^{\{2\}}=\mu^{\{1,2\}}=1$, and $\mu^{\{1\}}=0$. In this game, the unique Nash equilibrium $s$ is for all players of type $\{2\}$ to consume from resource $2$ and for all players of type $\{1,2\}$ to consume from resource $1$. Note that $\height{1}=1$ and $\height{2}=1$, and so the set of highest-costing resources is $\{1,2\}$. We note that $\arg\Max_{\smash{S\in\types}}E_G(S)=\bigl\{\{2\},\{1,2\}\bigr\}$, just as in \cref{need-mg}, even though the set of highest-costing resources in that examples is different.
\end{example}

Examining \cref{need-mg}, we note that while, indeed, the total mass of all players who cannot consume from any resource outside $\{1,2\}$ in that \lcnamecref{need-mg}, when ``equalized'' among the resources in $\{1,2\}$, yields a ``height'' of~$2=\Max_{\smash{S\in\types}} E_G(S)=E_G(P)$, in fact no consumption profile corresponds to this equalization, as not enough of these players are allowed to consume from resource $1$.\footnote{If both $f_1$ and $f_2$ were strictly increasing, then this would imply that $E_G\bigl(\{2\}\bigr)>E_G\bigl(\{1,2\}\bigr)$; such an argument may be used to show that when all $f_j$ are strictly increasing,
$\arg\Max_{\smash{S\in\types}} E_G(S)$ is indeed closed under set union, and that its greatest element indeed
coincides with the set of highest-costing resources.} To derive a general formula for the set of highest-costing resources, we therefore need to exclude such ``problematic'' sets of resources.

\begin{definition}[$M_G$; $D_G$; $P_G$; $h_G$]
Let $G=\game$ be a resource selection game. We define:
\begin{itemize}
\item
$\displaystyle M_G(S)\eqdef\Bigl\{S'\in\nonemptysubs{S}~\Big|~\forall\mu\le\smashoperator{\sum_{\type\in\nonemptysubs{S}\setminus\smash{\nonemptysubs{S\setminus S'}}}}\mu^{\type}:\equalize{f_k:k\in S'}(\mu)\ne E_G(S)\Bigl\}\subseteq\nonemptysubs{S}\setminus\{S\}$, for every $S\in\types$. A set \mbox{$S'\in M_G(S)$} is not allowed for consumption by enough players (out of those who can consume only from $S$) to create a consumption profile s.t.\ only consumers who can consume solely from $S$ consume from $S$, and s.t.\ the cost of each resource in $S$ is $E_G(S)$. In \cref{need-mg}, $M_G\bigl(\{1,2\}\bigr)=\bigl\{\{1\}\bigr\}$ while $M_G\bigl(\{1\}\bigr)=\emptyset$.
\item
$\displaystyle D_G\eqdef\bigl\{S\in\types~\big|~E_G(S)\in\mathbb{R} \And M_G(S)=\emptyset\bigr\}$ --- these are the sets termed ``nonproblematic'' above.
\item
$\displaystyle P_G\eqdef\bigcup\arg\smashoperator{\Max_{S\in D_G}}E_G(S)$ --- we show below that this is precisely the set of highest-costing resources.
\item
$\displaystyle h_G\eqdef\smashoperator{\Max_{S\in D_G}}E_G(S)$ --- we show below that this is precisely the cost of every resource in $P_G$.
\end{itemize}
\end{definition}

\begin{remark}\label{p-remark} Let $G=\game$ be a resource selection game.
We show in \cref{proofs-derivation} that in the cases that we study (i.e., where $G$ has a Nash equilibrium or where $f_1,\ldots,f_n$ are continuous), $P_G$~is the greatest element of $\arg\Max_{S\in D_G}E_G(S)$, and so $h_G=E_G(P_G)$. Furthermore, we show that in these cases $h_G=\Max_{\smash{S\in\types}}E_G(S)$; i.e., $D_G$ may be replaced by $\types$ in the definition of $h_G$. (But not in the definition of $P_G$ if even one of $f_j$ is not strictly increasing, by \cref{need-mg}.)
\end{remark}

\begin{lemma}\label{p-nonempty}
In every resource selection game $G$,\ \ $P_G\ne\emptyset$, and $h_G\in\mathbb{R}$ is well defined.
\end{lemma}

\subsubsection{Resource Removal}

The following definition will be useful when formalizing \cref{obs-start-over} from \cref{hydraulic-intuition}.

\begin{definition}[Resource Removal]
Let $G=\game$ be a resource selection game and let $S\subseteq[n]$ be a subset of the resources in $G$.
\begin{itemize}
\item
For every $\type'\in\nonemptysubs{[n]\setminus S}$, we define $\mathcal{R}(\type',G-S)\eqdef\bigl\{\type\in\types~\big|~\type\setminus 
S=\type'\bigr\}\subseteq\types\setminus\nonemptysubs{S}$ --- the set of player types in $G$ for whom the allowed resources outside $S$ are precisely $\type'$.\footnote{We emphasize that $\mathcal{R}(\type',G-S)$ is defined as a function of three parameters: $\type'$, $G$, and $S$, rather than as a function of two parameters ($\type'$ and $G-S$, the latter of which we have not yet defined). We use the notation $\mathcal{R}(\type',G-S)$ rather than $\mathcal{R}(\type',G,S)$ solely for readability.}
\item
We define $G-S \eqdef\Bigl((f_j)_{j\in[n]\setminus S};\bigl(\sum_{\type\in \mathcal{R}(\type',G-S)}\mu^{\type}\bigr)_{\smash{\type'\in\nonemptysubs{[n]\setminus S}}}\Bigr)$ --- the $\bigl|[n]\setminus S\bigr|$-resource selection game obtained from $G$ by disallowing any consumption from resources in $S$ and removing all players who cannot consume from any resource outside $S$.
\end{itemize}
\end{definition}

\begin{lemma}[Fundamental Properties of Resource Removal]\label{associative-remove}
Let $G$ be a resource selection game.
\begin{parts}
\item\label{associative-remove-empty}
$G-\emptyset=G$.
\item\label{associative-remove-associative}
$G-S-S'=G-(S\cup S')$, for every two disjoint subsets $S,S'$ of the resources in G.
\end{parts}
\end{lemma}

\subsection{Formal Derivation of the Results of Section~\refintitle{results}}\label{derivation}

In this \lcnamecref{derivation}, we
present our analysis formalizing the \lcnamecrefs{obs-min-liquid} from \cref{hydraulic-intuition} via the definitions of \cref{definitions} and leading to the results of \cref{results}. Full proofs of all the results of this \lcnamecref{derivation} are given in \cref{proofs-derivation}; the subsequent proofs of the results of \cref{results} are given in \cref{proofs-results}.

\subsubsection{Uniqueness and Strength}

At the heart of our proof of \cref{indifference-resource-costs} lies \cref{highest-stopping},
formalizing \cref{obs-min-liquid,obs-any-nash,obs-start-over,obs-indifference} from \cref{hydraulic-intuition}. We note that unlike \cref{nash-exists}, neither \cref{highest-stopping} nor \namecref{indifference-resource-costs} \labelcref{indifference-resource-costs} or \labelcref{super-strong} require the continuity of $f_1,\ldots,f_n$.

\begin{lemma}[Uniqueness of Highest-Costing Resources and their Cost]\label{highest-stopping}
Let $s$ be a Nash equilibrium in a resource selection game
$G=\game$, and let $P^s\eqdef\arg\Max_{j\in[n]}\height{j}$.
\begin{parts}
\item\label{highest-stopping-who}
$P^s=P_G$.
\item\label{highest-stopping-where}
$\height{j}=h_G$, for every $j\in P^s$.
\item\label{highest-stopping-on}
$s_j(\type)=0$ for every $\type\in\types\setminus\nonemptysubs{P^s}$ and $j\in P^s$.
\item\label{highest-stopping-rest}
The function $s':\nonemptysubs{[n]\setminus P^s}\rightarrow\Rge^{[n]\setminus P^s}$, defined by
$s'_j(\type')\eqdef\sum_{\type\in \mathcal{R}(\type',G-P^s)}s_j(\type)$ for every $j\in[n]\setminus P^s$ and $\type'\in\nonemptysubs{[n]\setminus P^s}$,
constitutes a Nash equilibrium in the game
$G-P^s$.
Furthermore, $\height{j}=\heightt{j}$ for every~$j\in[n]\setminus P^s$.
\end{parts}
\end{lemma}

\begin{remark}
In \cref{highest-stopping}, The r.h.s.\ of \cref{highest-stopping-who,highest-stopping-where},
and therefore also the quantifications in \cref{highest-stopping-where,highest-stopping-on,highest-stopping-rest} and the game defined using resource removal in \cref{highest-stopping-rest}, are independent of the choice of $s$.
\end{remark}

The proof of \cref{indifference-resource-costs} using \cref{highest-stopping} is given in \cref{proofs-results}.
This proof effectively follows \cref{compute-height},%
\begin{algorithm}[ht]%
\caption{Direct computation of $\height{j}$ for all $j\in[n]$, regardless of the choice of Nash equilibrium~$s$.}\label{compute-height}%
\small%
\begin{algorithmic}[1]%
\LineComment{$\mathbfit{S}$ is the set of pistons that have already stopped.}
\State $S \gets \emptyset$\Comment{See \crefpart{associative-remove}{empty}.}
\While{$S\ne[n]$}
\LineComment{By \crefpart{highest-stopping}{rest}, the pistons $\mathbfit{P}$ that stop next are those that stop earliest in}
\LineComment{the $\mathbfit{\bigl|[n]\setminus S\bigr|}$-resource selection game $\mathbfit{G-S}$.}
\State $P\gets P_{G-S}$\Comment{By \crefpart{highest-stopping}{who}.}
\LineComment{$\mathbfit{h}$ is the height at which the pistons $\mathbfit{P}$ stop.}
\State $h \gets h_{G-S}$\Comment{By \crefpart{highest-stopping}{where}.}
\ForAll{$j\in P$}
\State $h_j \gets h$
\EndFor
\State $S \gets S\cup P$\Comment{See \crefpart{associative-remove}{associative}.}
\EndWhile
\State \Return{$(h_1,\ldots,h_n)$}
\end{algorithmic}%
\end{algorithm}
a succinct algorithm (based upon \cref{highest-stopping}), which, if any Nash equilibrium exists, directly and explicitly calculates $\height{j}$ for all $j$ in every Nash equilibrium $s$ (without the need to first calculate players' strategies, which are dependent on $s$).

Full proofs of \cref{indifference} and \cref{super-strong} are given in \cref{proofs-results}. The former is based on \cref{indifference-resource-costs} (as explained in \cref{obs-indifference} from \cref{hydraulic-intuition}), and the latter on the analysis of \cref{highest-stopping}, following and formalizing an extension of \cref{obs-nash} from \cref{hydraulic-intuition}.
We conclude this section by demonstrating that, as suggested by the manner in which \cref{super-strong} is stated, a Nash equilibrium is not necessarily super-strong when the condition
of \cref{super-strong-super-strong} of this \lcnamecref{super-strong} (regarding the plateau heights of the cost functions) is not met.

\begin{example}[A Not-Super-Strong
Equilibrium] Consider
the game $G\!=\!\game$, for $n=2$, $f_1=\id$, $f_2(x)=\min\{x,3\}$, $\mu^{\{1\}}=1$, $\mu^{\{2\}}=2$, and $\mu^{\{1,2\}}=3$. In this game, a (strong) Nash equilibrium~$s$ is given by
$\{1\}\mapsto(1,0)$, $\{2\}\mapsto(0,2)$, $\{1,2\}\mapsto(2,1)$. (Note that $\height{2}=3$ is a plateau height of $f_2$.) This Nash equilibrium is not super-strong, since a coalition of players of types $\{1\}$ and $\{1,2\}$ can deviate with $\{1\}\mapsto(1,0)$ (no change) and $\{1,2\}\mapsto(0,3)$ (more players consuming from resource $2$), from which no coalition member is harmed, while coalition members of type $\{1\}$ benefit.
A different deviation showing that this Nash equilibrium is not super-strong and worth mentioning is of a coalition consisting solely of players of type~$\{1,2\}$, which can deviate with $\{1\}\mapsto(1,2)$ (by, e.g., some coalition members switching to consume from resource $2$ instead of resource~$1$ while the others do not change strategies, or, e.g., by each player of type~$\{1,2\}$ switching resources), from which no coalition member is harmed, while the coalition members consuming from resource $1$ (whether they have actually changed strategies or not) benefit.
\end{example}

\subsubsection{Existence}

We proceed to the proof of existence of equilibrium. A full proof of \cref{nash-exists} is given in \cref{proofs-results}. This proof formalizes
\cref{obs-min-liquid,obs-nash} from \cref{hydraulic-intuition}, effectively following the construction of \cref{balloons} and
showing that in each step, the pistons stopping are those computed in \cref{compute-height}. This is done using the following \lcnamecref{construct}, \emph{constructively} showing, even in the absence of prior knowledge of existence of Nash equilibrium, that
the liquid that by \crefpart{highest-stopping}{on} should be under the pistons $P$ when they stop can be distributed appropriately among them, and that \cref{compute-height} indeed finds the sets $P$ in decreasing order of stopping height.

\sloppy
\begin{lemma}[$P_G$ and $h_G$ are Viable as Highest-Costing Resources and their Cost]\label{construct}
Let $G=\game$ be a resource selection game s.t.\ $f_1,\ldots,f_n$ are continuous.
\begin{parts}
\item\label{construct-distribute}\thmitemtitle{Liquid Distribution under $P_G$}
There exists a consumption  profile $s$ in the $|P_G|$-resource selection game $\bigl((f_j)_{j\in P_G};(\mu^{\type})_{\type\in\nonemptysubs{P_G}}\bigr)$,
s.t.\ $\height{j}=h_G$ for every $j\in P_G$.
\item\label{construct-stopping-order}\thmitemtitle{Pistons Stopping Order}
If $P_G\ne[n]$, then $h_G>h_{G-P_G}$.
\end{parts}
\end{lemma}
\fussy

We remark that the constructive machinery that we develop in \cref{proofs-derivation}, and in particular in \cref{constrained-distribution}, in order to prove \crefpart{construct}{distribute} can also be used to derive a variant of \cref{compute-height} that computes not only the equilibrium resource costs, but also some concrete equilibium, by using this machinery to compute, in each step of the \lcnamecref{compute-height}, a consumption profile for all players of types in $\nonemptysubs{P_{G-S}}$ such that they all consume from resources in $P$ and such that the cost of each resource in $P$ is $h_{G-S}$.

\subsection{Unified, Tangible Intuition and Insights}

Several approaches in the literature nonconstructively show the existence of Nash equilibria in resource selection games by characterizing Nash equilibria as (local) minima of certain functions. A ramification of these approaches is that \emph{every}~construction that finds a Nash equilibrium in a resource selection game must minimize these functions (at least locally); in particular, even though it does not explicitly aim to do so (indeed, these functions do not appear in any of our proofs or definitions), so does our hydraulic construction. It is interesting to note, though, that our hydraulic construction does not merely minimize these functions as a ``side effect'' of finding a Nash (in fact, strong) equilibrium, but in fact also provides very tangible intuition into these abstract nonconstructive approaches, and thus, in a sense, unifies the intuition behind these approaches, which may otherwise be thought of as somewhat detached from one another.

One approach, which is more straightforward to connect to our hydraulic construction, is that of characterizing Nash equilibria as consumption profiles $s$ for which the resource costs vector $(\height{j})_{j=1}^n$, when sorted in nondecreasing order, is (locally) lexicographically minimal \citep{FKKMS09}. Indeed, among other properties, our hydraulic system constructively lexicographically minimizes the sorted vector of resource costs by first minimizing the highest resource cost, then minimizing the second-highest resource cost, etc. A seemingly-unrelated approach, which at first glance may seem less straightforward to connect to our hydraulic construction, is that of characterizing Nash equilibria as consumption profiles minimizing an appropriate \emph{potential function} \citep{MondererShapley96,Beckmann56}, so named due to its properties, which resemble those of an abstract physical potential.
While this approach may at first glance seem less straightforward to connect to our hydraulic construction, in the following \lcnamecref{potential}
we show that our analysis in fact provides powerful insights into the abstract game-theoretic potential function, via a natural concrete physical interpretation.

\subsubsection{Abstract Game-Theoretic Potential as Physical Gravitational Potential}\label{potential}

A popular nonconstructive method for proving the existence of Nash equilibrium in resource selection games (and, more generally, in congestion games) is to define an appropriate \emph{potential function}, so named due to its properties, which resemble those of an abstract physical potential. (The term ``potential function'', in this context, is due to \cite{MondererShapley96}, who used it for atomic congestion games; see \cite{Nisanbook} for details regarding the usage for nonatomic congestion games, which was originated by \cite{Beckmann56}.) In the notation of this paper, the proof defines the following scalar function of consumption profiles:
\[P^*(s)\eqdef\sum_{j\in[n]}\int_0^{\load{j}}f_j(x)dx,\]
and shows that when a single player deviates from one strategy to another, the change in this player's cost equals, roughly speaking, the derivative of $P^*$ in the direction of the deviation. The conclusion (under certain assumptions) is that there exists a consumption profile minimizing~$P^*$, and that this profile is therefore a Nash equilibrium; in fact, it can be shown that a consumption profile is a Nash equilibrium iff it minimizes $P^*$.
We now show that as claimed above, our hydraulic construction does not merely minimize $P^*$ as a ``side effect'' of finding a Nash equilibrium, but also gives a natural physical interpretation to $P^*$, which justifies the name ``potential function'' not only abstractly, but also concretely.

The reader may recall from high-school physics class that the \emph{gravitational potential energy} of a point mass of mass $m$ near the surface of the earth is given by $mgh$, where $h$ is the height of the mass, and $g$ is the standard acceleration due to gravity.\footnote{As is customary, we ignore the negligible effect of small changes in $h$ on the value of $g$.} More generally, the gravitational potential energy of a non-point-mass system may be expressed by the \mbox{Riemann--Stieltjes} integral $\int_{-\infty}^{\infty} gh\,dm(h)$,
where $m(h)$ is the cumulative mass in the system up to height $h$. In our hydraulic system, we have
\[ \int_{-\infty}^{\infty} gh\,dm(h) = g\sum_{j\in[n]}\int_0^{\load{j}}f_j(x)dx=g\cdot P^*(s), \]
and so $P^*$, up to a multiplicative constant, is precisely the gravitational potential energy of our system (which our construction therefore turns out to minimize).
Perhaps more intuitive to nonphysicists would be to reason not about the gravitational potential energy of our hydraulic system, but rather about the height of the \emph{center of mass} of the system, given by:
\[ \frac{1}{\mu}\int_{-\infty}^{\infty} h\,dm(h)=\frac{1}{\mu}\sum_{j\in[n]}\int_0^{\load{j}}f_j(x)dx=\frac{P^*(s)}{\mu},\]
where $\mu\eqdef \lim_{h\to\infty} m(h)=\sum_{\type\in\type}\mu^{\type}$ is the total mass of the system. Once again, the height of the center of mass, which our construction turns out to minimize, equals $P^*$ up to a multiplicative constant.\footnote{The reader is referred once again to the special case of our construction that is given in \cite{noncoop-market-alloc}, which can be more easily and intuitively shown to minimize gravitation potential energy and height of center of mass.}

We conclude this discussion by noting that the generalization of resource selection games studied in the next \lcnamecref{compress-expand} also demonstrates
that hydraulic analysis is by no means confined to games that can be analyzed
via the game-theoretic potential approach.
Indeed, even though the games studied in the next \lcnamecref{compress-expand} are generally not potential games in the sense of \cite{MondererShapley96}, our hydraulic construction naturally extends to solving them. (Indeed, in that generalized setting, in which the total mass of the system is no longer constant or the density not uniform, our construction no longer necessarily minimizes the gravitational potential energy, nor the height of the center of mass, of the system.)

\section{Resource Selection Games with I.D.-Dependent Weighting}\label{compress-expand}

In this \lcnamecref{compress-expand}, we describe
an extension of the results of \cref{id-independent} to a model where the cost of a resource may depend on the identity, rather than merely the quantity, of players using it.
While such major extensions have been studied in the context of atomic games, no tools have been previously offered to tackle them in nonatomic settings.

\subsection{Setting}

For $n,k\in\mathbb{N}$, an \emph{$n$-resource/$k$-player-type resource selection game with  I.D.-dependent weighting} is defined by a triple
$\idgame$,
where $f_j:\Rge\rightarrow\mathbb{R}$ is a nondecreasing function for every resource $j\in[n]$, where $\type^i\in\types$ for every player type~${i\in[k]}$, and where ${f^i_j:[0,1]\rightarrow\Rge}$ is an increasing function for every player-type/legal-resource pair~$(i,j)\in\bigcup_{i\in[k]}\{i\}\times\type^i$. For each player type~$i\in[k]$,\ \ $\type^i$ specifies the set of resources from
which this player type may consume. As before, for each resource~$j\in[n]$,\ \ $f_j$~is a function from the consumption amount of this resource to its
cost. The newly introduced functions $f^i_j$ replace and generalize the player-type masses $\mu^{\type}$ from \cref{id-independent} and indicate the weighting of the consumption of player type~$i$ from resource $j$ (see below).

A \emph{consumption (strategy) profile} in this game is a function $s:[k]\rightarrow \Rge^{[n]}$ s.t.\ $s(i) \in \Delta^{\type^i}$ for every $i\in[k]$,
indicating the fraction of players of type $i$ that consume from each resource.
Given a consumption profile $s$ in this game, for every $j \in [n]$ we define $\load{j}\eqdef\sum_{i:j\in\type^i} f_j^i\bigl(s_j(i)\bigr)$ (note the newly introduced weighting) --- the \emph{weighted} load on (i.e., total weighted consumption from) resource $j$. As before, we define $\height{j}\eqdef f_j(\load{j})$ for every $j\in[n]$ --- the cost of resource $j$. A \emph{Nash equilibrium} in this game is a consumption profile $s$ s.t.\ for every $i\in[k]$ and for every $\ell\in\supp\bigl(s(i)\bigr)$ and $j\in\type^i$, it is the case that $\height{\ell}\le\height{j}$.

\begin{remark}
A resource selection game as in \cref{id-independent} can be represented as a resource selection game with I.D.-dependent weighting by defining, for
every player type $\type$ in the former game with $\mu^{\type}\ne0$, a player type $i$ in the latter game with $\type^i=\type$ and $f^i_j=\mu^{\type}\cdot\id$ for all $j\in\type=\type^i$. So, the setting of this \lcnamecref{compress-expand} is indeed a (strict) generalization of resource selection games as defined
in \cref{id-independent}.
\end{remark}

\begin{example}[A Cloud Computing Market]
Consider a scenario in which the resources are computer servers, and each of the many player wishes to run a relatively small computing job, where jobs corresponding to players of the same type are of a similar nature.
A player of type $i\in[k]$ may choose between the machines~$\type^i$, whose hardware is compatible with jobs of players of this type, and would like for her job to be completed as soon as possible given this constraint. $f^i_j$~in this case is a linear function s.t.\ $f^i_j(x)$~is proportional to the number of cycles of machine~$j$ required to compute the jobs of an $x$ fraction of the players of type~$i$. (The hardware of each machine may run jobs of some nature  more efficiently than jobs of another nature; e.g., machine~$2$ may run image-processing jobs faster than text-analysis ones, while machine $3$ may run the latter faster than the former.) For each $j\in[n]$, we choose~$f_j$ s.t.\ $\height{j}=f_j(\load{j})$ is proportional to the number of seconds required for $\load{j}$ cycles of machine $j$ to complete. (Assume that the resources of each machine are parallelized between its different users, so that their jobs are all completed at the same time.)
\end{example}

An additional example for the natural emergence of I.D.-dependent weighting may be given by analyzing traffic congestion games, where some vehicles, such as trucks, may cause significantly higher congestion than cars in narrow roads, while causing only moderately higher congestion than cars (or even the same congestion) in wide roads.

\subsection{Hydraulic Adaptation and Formal Results}

Intuitively, the hydraulic construction from \cref{\balloonswithanimation}
may be adapted to this generalized framework by inserting ``compressors/expanders'' into the tubes between balloons corresponding to the same player type. E.g., if $\type^1=\{1,2\}$, $f_1^1(x)=x$ and $f_2^1(x)=2x$, then the balloon system corresponding to player type $1$ consists of two balloons, one in container $1$ and the other in container $2$, connected by a compressor/expander tube s.t.\ for each drop of liquid that enters the tube from the balloon in container $1$, two drops exit into the balloon in container $2$, and for every two drops of liquid that enter the tube from the balloon in container $2$, one drop exits into the balloon in container~$1$.

The first thing that we note about this generalized game is that it no longer holds that $\height{j}$ is independent of the choice of Nash equilibrium $s$; see \cref{many-heights} for an illustration.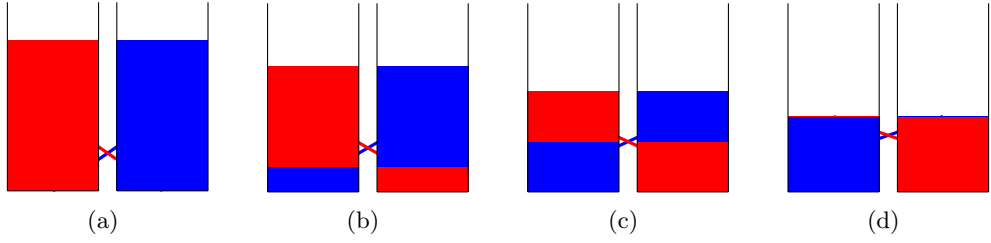
\begin{figure}[ht]%
\centering%
\subfigure[]{%
\begin{tikzpicture}[yscale=1,xscale=1.2]
\compressor{0}
\end{tikzpicture}%
}\qquad
\subfigure[]{%
\begin{tikzpicture}[yscale=1,xscale=1.2]
\compressor{1/3}
\end{tikzpicture}%
}\qquad
\subfigure[]{%
\begin{tikzpicture}[yscale=1,xscale=1.2]
\compressor{2/3}
\end{tikzpicture}%
}\qquad
\subfigure[]{%
\label{many-heights-super-strong}%
\begin{tikzpicture}[yscale=1,xscale=1.2]
\compressor{1}
\end{tikzpicture}%
}
\caption{Liquid distributions among balloons, corresponding to a plethora of Nash equilibria $s$ with distinct $\height{j}$, when $n=2$, $k=2$ (blue corresponding to $i=1$, and red --- to $i=2$), $f_1=f_2=\id$, $\type^1=\type^2=\{1,2\}$, $f^1_1(x)=f^2_2(x)=x$ and $f^1_2(x)=f^2_1(x)=2x$.
Only the Nash equilibrium depicted in \cref{many-heights-super-strong} is strong (in fact, it is super-strong); this is the unique equilibrium that our hydraulic construction finds.}%
\label{many-heights}%
\end{figure}
Nonetheless, if we accept the physical intuition that when compressed via pistons, each of the liquid distributions given in \cref{many-heights} eventually reaches the liquid distribution depicted in \cref{many-heights-super-strong}, then it is intuitively clear why our construction can be formally shown to yield a strong (and under conditions similar to those of \crefpart{super-strong}{super-strong}, super-strong) Nash equilibrium, proving the existence of such an equilibrium. Consequently, uniqueness of~$\height{j}$ can still be shown to hold among strong Nash equilibria; this result draws a novel fundamental connection between group deviation and I.D.-congestion, which to the best of our knowledge has never been drawn before in either nonatomic or atomic resource selection or congestion games. Formally, \cref{nash-exists,indifference-resource-costs,indifference,super-strong} generalize as follows.

\begin{theorem}[$\exists$ Strong Nash Equilibrium]\label{id-nash-exists}
Let $G=\idgame$ be a resource selection game with I.D.-dependent weighting.
If $(f_j)_{j=1}^{n}$ and $(f^i_j)_{j\in\type\smash{^i}}^{i\in[k]}$ are all continuous, then a strong Nash equilibrium exists in $G$.
\end{theorem}

\begin{theorem}[Uniqueness of Strong Equilibrium Resource Costs]\label{id-indifference-resource-costs}
Let $G$ be an $n$-resource selection game with I.D.-dependent weighting.
$\height{j}=\heightt{j}$ for every $j\in[n]$ and every two strong Nash equilibria $s,s'$ in $G$.
\end{theorem}

\begin{corollary}\label{id-indifference} Let $G=\idgame$ be a resource selection game with I.D.-dependent weighting.
\begin{parts}
\item\thmitemtitle{Players are Indifferent between Strong Equilibria}\label{id-indifference-players}
$\height{k}=\heightt{k'}$ for every $k\in\supp\bigl(s(\type)\bigr)$ and $k'\in\supp\bigl(s'(\type)\bigr)$, for every $\type\in\types$ and every two strong Nash equilibria $s,s'$ in $G$.
\item\thmitemtitle{Uniqueness of Strong Equilibrium Resource Loads}\label{id-indifference-resources}
If $(f_j)_{j=1}^n$ are strictly increasing, then $\load{j}=\loadt{j}$ for every $j\in[n]$ and every two strong Nash equilibria $s,s'$ in $G$.
\end{parts}
\end{corollary}

\begin{theorem}[All Strong Equilibria are Super-Strong]\label{id-super-strong}
Let $G=\idgame$ be a resource selection game with I.D.-dependent weighting.
If $\height{j}$ is not a plateau height of $f_j$ for each $j\in[n]$ in any/every strong Nash equilibrium $s$, then all strong Nash equilibria in $G$ are super-strong.
\end{theorem}

We note that the conditions of \crefpart{id-indifference}{resources} (strict monotonicity of cost functions) are considerably stricter than those of the analogous \crefpart{indifference}{resources} (no two cost functions sharing a plateau height). As we now show, the weaker conditions of the latter \lcnamecref{indifference} do not suffice for the former one.

\begin{example}[Nonuniqueness of Strong Equilibrium Resource Loads When All But One Function are Strictly Increasing]
Consider the game $G=\idgame$, for $n=2$, $k=2$, $f_1=\id$ (strictly increasing, and so having no plateau heights), $f_2(x)=\min\{x,1\}$, $\type^1=\type^2=\{1,2\}$, $f^1_1(x)=f^2_1(x)=f^1_2(x)=x$ and $f^2_2(x)=2x$. In this game, a strong Nash equilibrium~$s$ is given by
$1\mapsto(0,1)$ and $2\mapsto(1,0)$. Note that an additional strong Nash equilibrium $s'$ is given by $1\mapsto(1,0)$ and $2\mapsto(0,1)$.
As \cref{id-indifference-resource-costs} predicts, indeed $\height{1}=1=\heightt{1}$ and $\height{2}=1=\heightt{2}$. Nonetheless, $\load{2}=1\ne2=\loadt{2}$. (In fact, a continuum of strong Nash equilibria $s''$ exist in this game, with $\loadtt{2}$ attaining all values in $[1,2]$.)
\end{example}

\subsection{Formal Tools Adaptation}

The formal analysis is similar to that of \cref{derivation}, and is in fact simpler. This seemingly curious simplicity is due to the nature of the explicit formula for strong equilibrium resource costs that we obtain in each of the scenarios. In the scenario of \cref{id-independent}, much of the complexity of our proof was in order to give the (relatively simple) explicit formula $h_G=\Max_{\smash{S\in\types}} E_G(S)$ (see \cref{p-remark}), where we recall that $E_G(S)=\equalize{f_k:k\in S}(\sum_{\smash{\type\in\nonemptysubs{S}}}\mu^{\type})$ --- the communicating-vessel equalization, among the vessels corresponding to $S$, of the entire mass of players who cannot consume from resources outside $S$. (Recall that the maximum is taken only over sets $S$ for which $E_G(S)$ is defined.) We remark that a bulkier and far less elegant formula, which is considerably easier to prove, is $h_G=\Max_{\smash{S\in\types}} B_G(S)$, where $B_G(S)$ stands for the cost in a ``balanced-cost'' strategy profile (i.e., a strategy profile in which all resources have the same cost) in the game $\bigl((f_j)_{j\in S};(\mu^{\type})_{\smash{\type\in\nonemptysubs{S}}}\bigr)$. (Again, this maximum is taken only over sets $S$ for which $B_G(S)$ is defined.\footnote{It is easy to verify that $E_G(S)$ is defined whenever $B_G(S)$ is defined, but not \emph{vice versa}.}) The former (much simpler) $\Max_{\smash{S\in\types}} E_G(S)$ does not explicitly check whether any of the distributions of the total mass of players $\sum_{\smash{\type\in\nonemptysubs{S}}}\mu^{\type}$ that yield an equalized cost across $S$ is in fact a legal strategy profile, i.e., one that does not assign players to resources that are illegal for them --- the existence of such a legal player distribution is shown in our proof (see \crefpart{construct}{distribute} and the results supporting it: \cref{distribute,p-in-argmax,max-on-all} and the entire analysis of \cref{constrained-distribution}) rather than taken as an (unrequired) assumption.\footnote{Such a ``bulkier'' $B_G(S)$ may be seen as somewhat closer in a sense to the ``packing oracle'' of \cite{HHKS13}.} Unfortunately, however, in the case of I.D.-depdendent weighting, as each player's mass changes in a type-dependent manner with her chosen resource, there is no way to provide an explicit formula for resource costs without resorting to a formula of the ``bulkier'' kind\footnote{In this case, $B_G(S)$ is defined as the cost in a balanced-cost  strong Nash equilibrium (if such exists) in the game $\bigl((f_j)_{j\in S};(\type^i)_{\smash{i:\type^i\in\nonemptysubs{S}}};(f^i_j)_{j\in\type\smash{^i}}^{i:\type\smash{^i\in\nonemptysubs{S}}}\bigr)$. (For comparison with how $B_G(S)$ would have been defined in the scenario of \cref{id-independent}, note that when a balanced-cost strong Nash equilibrium exists, then it is a balanced-cost strategy profile \emph{with least cost} among all balanced-cost strategy profiles, which may indeed have various costs as in \cref{many-heights}.)} that we were able to avoid in \cref{id-independent} (it is no longer possible to simply reason about the ``sum of the mass'', or the ``total load'', of players who cannot consume from resources outside $S$), resulting in simpler proofs, which we defer, along with the full details, to the full paper.

\section{Beyond Games}
\label{hall-and-beyond}

As we have been pointing out above, the development of the machinery of this paper is free of any fixed-point theorem, of the Minimax theorem and any equivalent result, and of linear programming. As such nonconstructive techniques are traditionally the tools used when attempting to establish the existence of equilibria, one may claim that in a sense, hydraulic analysis ``replaces'' such techniques in our analysis.
It is therefore only natural to ask whether other results that are traditionally obtained via linear-programming methods can be also be derived as consequences of our machinery.\footnote{This question also naturally arises from noting that a na\"{\i}ve \emph{in silico} computation of whether a descending piston is blocked bears a striking resemblance to a search for an augmenting path.} In this \lcnamecref{hall-and-beyond}, we show that hydraulic analysis can indeed serve as a constructive substitute to linear-programming approaches also outside the realm of games, shedding new light on several flow/linear-programming problems. We start by deducing a novel, surprisingly intuitive, proof of Hall's theorem using hydraulic analysis.

\subsection{Case Study: Application to Hall's Fractional Marriage Theorem}\label{hall}

For this \lcnamecref{hall}, let $n\in\mathbb{N}$, and for every $i\in[n]$, let $\type^i\subseteq[n]$. We consider a scenario involving $n$ women and $n$ men, where for every $i\in[n]$, we interpret $\type^i$ as the set of men acceptable to woman $i$; a \emph{perfect marriage} is a one-to-one correspondence $\mathcal{M}:[n]\rightarrow[n]$, where $\mathcal{M}(i)\in\type^i$ for every $i\in[n]$. For every subset $I\subseteq [n]$, which we interpret as a set of women, we define $\type^I=\cup_{i\in I}\type^i$ --- the set of men acceptable to at least one woman in~$I$. A well-known result in graph theory is the following characterization of the conditions for the existence of a perfect marriage.

\begin{theorem}[Hall's Marriage Theorem~\citep{Hall35}]\label{hall-int}
A perfect marriage exists iff $|I|\le|\type^I|$ for every $I\subseteq[n]$.
\end{theorem}

We now use the machinery of this paper to prove a slightly weaker form of \cref{hall-int}. A \emph{perfect fractional marriage} is a function $s:[n]\rightarrow\Rge^{[n]}$ s.t.\ $s(i)\in\Delta^{\type^i}$ for every $i\in[n]$ and s.t.\ $\sum_{i=1}^n s_j(i)=1$ for every $j\in[n]$.

\begin{theorem}[Fractional Version of Hall's Marriage Theorem]\label{hall-frac}
A perfect fractional marriage exists iff $|I|\le|\type^I|$ for every $I\subseteq[n]$.
\end{theorem}

To prove \cref{hall-frac}, we analyze the underlying scenario as a resource selection game.
Let $f_j\eqdef\id$ for every $j\in[n]$, let $\mu^{\type}\eqdef\Bigl|\bigl\{i\in[n]~\big|~\type^i=\type\bigr\}\Bigr|$ for every $\type\in\types$, and define $G\eqdef\game$ --- an $n$-resource selection game. The following two \lcnamecrefs{hall-char} are obtained directly from definitions.

\begin{lemma}\label{hall-char}
A function $s:[n]\rightarrow\Rge^{[n]}$ is a perfect fractional marriage iff $s$ is a Nash equilibrium in $G$ with $\height{j}=1$ for every $j\in[n]$.
\end{lemma}

\begin{proof}
By definition, $s$ is a perfect fractional marriage iff $s$ is a consumption profile in $G$ with $\height{j}=1$ for every $j\in[n]$. A consumption profile $s$ in $G$ with $\height{j}=1$ for every $j\in[n]$ is, by definition, a Nash equilibrium.
\end{proof}

\begin{lemma}\label{hall-cond}
A perfect fractional marriage exists iff $\height{j}=1$ for every $j\in[n]$ and every Nash equilibrium $s$ in $G$.
\end{lemma}

\begin{proof}
Immediate from \cref{hall-char,nash-exists,indifference-resource-costs}.
\end{proof}

\cref{hall-char,hall-cond}, in conjunction with the analysis of \cref{id-independent}, give rise to the following hydraulic algorithm for finding a perfect fractional marriage (or disproving its existence): set up the hydraulic system corresponding to $G$ (as in \cref{hydraulic-intuition}), and start lowering the pistons until a Nash equilibrium is obtained. If the resulting stopping heights are all $1$, then this equilibrium is a perfect fractional marriage; otherwise, no perfect fractional marriage exists. We use this algorithm to outline what we consider to be a surprisingly intuitive proof for \cref{hall-frac}; we now focus on the ``harder'' direction of this \lcnamecref{hall-frac}, i.e., that lack of a perfect fractional marriage implies that $|I|>|\type^I|$ for some $I\subseteq[n]$; the other (trivial) direction is left to the reader. A succinct formalization of the following argument is given in \cref{proofs-hall}.

\begin{proof}[Proof sketch]
Assume that not all stopping heights are $1$. Hence, the earliest-stopping pistons, $P_G$,  stop at a height higher than $1$. Therefore,
there is more than $|P_G|$ liquid under the pistons~$P_G$ when they stop. Thus,
there exists a set $I$ (of all players corresponding to this mass of liquid) s.t.\ $\type^I=P_G$ even though~$|I|>|P_G|=|\type^I|$, as required.
\end{proof}

\subsection{Beyond Hall's Theorem}\label{beyond-hall}

We note that, in fact,
the argument presented in \cref{hall}
does not require hydraulic analysis, and can also be carried out using \cref{nash-exists,indifference-resource-costs} (and our explicit formula for $P_G$) as black boxes. While we find that phrasing it in terms of hydraulic analysis makes for far more tangible intuition, this is not the only reason why we have chosen this presentation. Indeed, as we show in this \lcnamecref{beyond-hall}, while several generalizations of Hall's (fractional) theorem can still be analyzed via a reduction to resource selection games, it may be intuitively simpler to analyze them directly via hydraulic analysis, and, moreover, even further generalizations can be hydraulically analyzed while it is not clear how to analyze them using resource selection games as defined in this paper.
(We stress that for ease of exposition, we intentionally do not describe the most general class of flow/linear-programming problems solvable using our machinery.)

For the remainder of this \lcnamecref{beyond-hall}, let $n,k\in\mathbb{N}$, for every $i\in[k]$ let $0\le\mu^i\le \mathcal{M}^i$, for every $j\in[n]$ let $0\le t_j\le T_j$, and for every $i\in[k]$ and $j\in[n]$,
let $0\le m^i_j\le M^i_j$.
A \emph{solution} to the triple $\genconstraint$ is a matrix $(q^i_j)_{j\in[n]}^{i\in[k]}$ satisfying all of the following.
\begin{itemize}
\item
$m^i_j\le q^i_j\le M^i_j$ for every $i\in[k]$ and $j\in[n]$.
\item
$\mu^i\le\sum_{j=1}^n q^i_j\le\mathcal{M}^i$ for every $i\in[k]$.
\item
$t_j\le\sum_{i=1}^k q^i_j\le T_j$ for every $j\in[n]$.
\end{itemize}

We note that Hall's theorem deals with the question of the existence of a solution for $\mu^i=\mathcal{M}^i=t_j=T_j=1$, $m^i_j=0$ and $M^i_j=\mathds{1}_{\type^i}(j)$.
(The attentive reader may note that the scenario described in this \lcnamecref{beyond-hall} is also a strict generalization of the problem of the satisfiability of distribution constraints from \cref{constrained-distribution}.)

While it is not hard to find conditions for the existence of such a solution, as well as methods for efficiently finding such a solution, by formulating an equivalent flow problem, we now analyze this problem using the hydraulic machinery of this paper.

We note that w.l.o.g.\ we may assume that $m^i_j=0$ for all $i$ and $j$ (otherwise, $m^i_j$ may be subtracted from $m^i_j,\mu^i,\mathcal{M}^i,t_j$, and $T_j$).
We first consider an ``intermediate'' case in which $\mu^i=\mathcal{M}^i$ for every $i\in[k]$, $t_j=T_j$ for every $j\in[n]$, and $M^i_j\in\{0,\mu^i\}$ (i.e., either forcing $q^i_j=0$ or not enforcing any limitation thereon) for every $i\in[k]$ and $j\in[n]$. (We assume in this case that $\sum_{i=1}^k\mu^i=\sum_{j=1}^n t_j$; otherwise, no solution can possibly exist.)

While this case may be easily solved using the same machinery as in the previous \lcnamecref{hall}, by setting $f_j(\mu)\eqdef\frac{\mu}{t_j}$ and setting $\mu^{\type}\eqdef
\sum_{i\in[k]:\type^i=\type}\mu^i$, where $\type^i\eqdef\bigl\{j\in[n]~\big|~ M^i_j>0\bigr\}$, it may also be directly analyzed using hydraulic analysis without the need to vary the shape of containers.

The main observation that we now make use of is that in the hydraulic algorithm presented in \cref{hall}, it is in fact not necessary to lower the pistons simultaneously. In fact, lowering them using any timing, as long as all of them eventually reach a height of $1$, results in a perfect fractional marriage, while failing to do so (using any timing) proves the absence of a perfect fractional marriage. Using this observation, we readily obtain a simpler hydraulic algorithm for solving the above ``intermediate case'': 
set up a hydraulic system with $n$ containers (all corresponding to the identity function, as in \cref{hydraulic-intuition}) and $k$ ``liquid colors'', where the $i$th liquid color has $\mu^i$ volume and has balloons in all containers $j\in[n]$ s.t.\ $M^i_j>0$. Start lowering the pistons in any order (say, sequentially) so that for every $j\in[n]$, piston $j$ eventually reaches height $t_j$ (or gets blocked from reaching this height). If all pistons successfully reach their respective desired stopping heights, then the liquid distribution is a solution as required; otherwise, no solution exists.\footnote{While colored liquids are often solutions, they are not the type of solutions we are interested in.}

Let us now consider arbitrary $0\le\mu^i\le\mathcal{M}^i$ and $0\le t_j\le T_j$ (but not yet arbitrary $M^i_j$).
While this scenario may still be analyzed as a resource selection game, the transformation into such a game becomes increasingly complex: (the verification of the following transformation into an $(n\!+\!1)$-resource selection game is left to the reader)
\[
f_j(\mu)\eqdef
\begin{cases}\frac{\mu}{t_j} & \mu<t_j \\
1 & t_j\le\mu\le T_j \\
\frac{\mu}{T_j} & T_j < \mu
\end{cases}
,\qquad
f_{n+1}\equiv 1
,\qquad
\mu^{\type}=\smashoperator{\sum_{\substack{i\in[k]: \\ \type^i=\type}}}\mu^i
,\qquad
\mu^{\type\cup\{n+1\}}=\smashoperator{\sum_{\substack{i\in[k]: \\ \type^i=\type}}}\mathcal{M}^i-\mu^i
.
\]
Nonetheless, using an argument similar to the one above, the hydraulic algorithm for solving such a case, while also more complex and with an additional element, is considerably easier to visualize than the one solving this resource selection game; indeed, it does not require oddly shaped containers corresponding to functions that are not strictly increasing (as in \cref{vessels-odd-shapes-min}), nor does it require the description of how precisely pistons interact with such containers. This algorithm may be described as follows:
set up a hydraulic system with $n\!+\!1$ containers (all corresponding to the identity function, as in \cref{hydraulic-intuition}) and $2k$ ``liquid colors'', where the $i$th liquid color has $\mu^i$ volume and has balloons in all containers $j\in[n]$ s.t.\ $M^i_j>0$ and where the $(k\!+\!i)$th liquid color has $\mathcal{M}^i\!-\!\mu^i$ volume and has balloons in the same containers as the $i$th liquid color, and, additionally, in container $n\!+\!1$. Imagine also that container $n\!+\!1$ has no associated piston, but is ``pressurized'' so that liquid flows into it only if a currently descending piston would otherwise get stuck (i.e., only if all space in all other relevant containers is already occupied). Start lowering the pistons in any order (say, sequentially) so that for every $j\in[n]$, piston $j$ eventually reaches height $T_j$ (or gets stuck in the process). If any piston gets stuck during this process, then no solution exists. We note that at the end of this process, the liquid distribution may not (yet) constitute a solution as required, since it is possible that some container $j$ contains less than $t_j$ liquid. Hence (say, one by one), we lower each piston $j$ from height $T_j$ to height $t_j$, or until it gets stuck in the process. At the end of this process (with each piston $j$ either at height $t_j$ or blocked from reaching this height), if each container $j\in[n]$ contains at least $t_j$ liquid (i.e., if each piston touches the liquid surface of its container), then the liquid distribution is a solution as required; otherwise, no solution exists.

Finally, generalizing to arbitrary $M^i_j$ (once again assuming w.l.o.g.\ that $m^i_j=0$), while the corresponding resource selection game would have to be generalized beyond the definition of a resource selection game as presented in either \cref{id-independent} or \cref{compress-expand}, the hydraulic algorithm may be modified by adding one additional very intuitive ``physical'' constraint: for every $i\in[k]$ and $j\in[n]$, all balloons of liquid color $i$ and of liquid color~$k\!+\!i$ in container $j$ are wrapped together in an outer balloon that may not inflate to a height greater than $M^i_j$.

We conjecture that hydraulic analysis may indeed yield many more intuitive and visually appealing proofs for various other linear-programming problems. Can it be applied even beyond linear-programming problems?

\section{Discussion}\label{discussion}

In addition to proving a gamut of novel theoretical results (in \cref{id-independent,compress-expand}), including some results that significantly strengthen age-old central theorems, and in addition to constructively reproving old results (in \cref{hall-and-beyond}), a significant feature of our machinery is that it provides an explicit expression (see \cref{definitions}, as well as \cref{highest-stopping} in \cref{derivation}) for the highest resource cost (and for the set of highest-costing resources) in equilibrium, which can be sequentially used (see \cref{compute-height} in \cref{derivation}) to compute all resource costs in equilibrium.

While the benefits of having an explicit expression for resource costs are by far not limited merely to actual computation, the question of the complexity of such actual computation is a valid one.
The complexity of calculating this expression in practice depends critically on our ability to compute the equalization $\equalize{f_k:f\in S}$ of any of the cost functions defining the resource selection game at hand; as with our ability to compute the cost functions themselves, our ability to compute their equalization strongly depends on the way they are specified. Indeed, in situations where the mere evaluation of some of the cost functions $f_j$ may be costly, it is hard to expect the calculation of their equalization to be any less costly; on the other hand, in many naturally occurring scenarios, the calculation of this equalization can be undertaken easily and efficiently, as is demonstrated in \cref{symmetric-equalize} (see, e.g., our proof of \cref{hall-frac}).
Assuming for a moment that the computation of this equalization can be carried out efficiently, then by \cref{p-remark}, the complexity of calculating the maximum cost (and when $f_j$ is strictly increasing --- also the highest-costing resources) is linear in the size of the input~$(\mu^{\type})_{\smash{\type\in\types}}$. It should be noted, though, that in some real-life scenarios, the input is sparse and can be efficiently encoded into a size considerably smaller than $\Theta(2^n)$; when no prior information is known regarding the structure of the input, this may render the complexity of this computation exponential in the encoded input size.

Nonetheless, as emphasized above, the benefits of having an explicit expression for resource costs are by far not limited merely to actual computation. Indeed, in \cref{hall-and-beyond} this explicit expression is used to phrase an extremely concise proof of Hall's theorem. For a more elaborate example, we turn to \cite{noncoop-market-alloc}, where we analyze the dynamics of a complex two-stage game: in the first stage, merchants choose store locations (some store locations are accessible to more customers than others, but in turn are associated with higher real-estate prices), and the second stage is a resource selection game, where each customer aims to purchase from a least crowded store; the payoffs to the merchants are determined according to the Nash equilibrium loads in the second-stage resource selection game (these are well defined by \cref{nash-exists,indifference-resource-costs}). When analyzing game dynamics between the merchants in this complex multistage game, we must somehow quantify the effect of 
dynamic changes in merchants' strategies in the first stage of the game (i.e., the effect of changes in the availability of a certain merchant to some customers) on the Nash equilibrium loads in the second-stage resource selection game; in other words, we must perform a comparative-statics analysis of the second-stage game.
Our explicit expression for resource costs offers a concise way to do precisely this.
(A proof of \cref{cost-lipschitz} that directly uses our explicit expression for resource costs and is \emph{free of any reasoning about incentives or deviations} is given in \cref{proofs-discussion}.)

\begin{proposition}[Comparative Statics: $h_j$ as a Function of $\mu^{\type}$]\label{cost-lipschitz}
Let $G=\game$ be  a resource selection game s.t.\ $f_1,\ldots,f_n$ are continuous. For every $j\in[n]$ and $\type\in\types$, both of the following hold, where $h_j$ is the cost of resource $j$ in all Nash equilibria of $G$.
\begin{parts}
\item\label{cost-lipschitz-nondecreasing}
$h_j$ is continuous and nondecreasing as a function of $\mu^{\type}$.
\item\label{cost-lipschitz-lipschitz}
If $f_j$ is Lipschitz, then $h_j$ is Lipschitz as a function of $\mu^{\type}$, with the same Lipschitz constant.\footnote{One may compare the elementary tools that we use to derive \cref{cost-lipschitz} with the considerably more complex tools used in \cite{Milchtaich00} to derive a statement similar in spirit to our \crefpart{cost-lipschitz}{lipschitz}; we note that \crefpart{cost-lipschitz}{nondecreasing}, around which the bulk of our proof of \cref{cost-lipschitz} revolves, has no counterpart following from the analysis of \cite{Milchtaich00}.}
\end{parts}
\end{proposition}

Indeed, a special case of \cref{cost-lipschitz} allows us \cite[in][]{noncoop-market-alloc} to prove powerful results regarding convergence of dynamics in this complex multistage game.

We conclude with a note about hydraulics. 
As can be seen when examining the extensions of our machinery in \cref{compress-expand} and in \cref{beyond-hall}, our hydraulic analysis framework is both flexible and robust; indeed, we conjecture that the full extent of its power is yet to be discovered, both within the realm of games and beyond.
The results of this paper,
as well as the earlier results of \cite{Fisher},
show not only that physical hydraulic systems may be a fruitful source
of intuition for proofs regarding equilibria, but furthermore that they may be used to naturally
``calculate'' a variety of flavors of equilibria. It would be interesting to rigorously define a ``hydraulic'' calculation, and to study its strength and limitations.

\section*{Acknowledgments}

Yannai Gonczarowski is supported by the Adams Fellowship Program of the Israel Academy of Sciences and Humanities; his work is supported by the European Research Council under the European Community's Seventh Framework Programme (FP7/2007-2013) / ERC grant agreement no.\ [249159], by ISF grants 230/10 and 1435/14 administered by the Israeli Academy of Sciences, and by Israel-USA Bi-national Science Foundation (BSF) grant number 2014389. We thank Noga Alon, Sergiu Hart, Noam Nisan, Binyamin Oz, two anonymous referees, and multiple seminar participants, for helpful comments and suggestions.

\bibliographystyle{abbrvnat}
\bibliography{hydraulic-selection}

\appendix

\section{Proofs and Auxiliary Results}

\subsection{Proofs of Lemmas and Corollaries from Sections~\refintitle{definitions} and~\refintitle{derivation}, and Auxiliary Results}\label{proofs-derivation}

\subsubsection{Definitions for Formalizing the Observations from Section~\refintitle{hydraulic-intuition}}

We begin with an immediate consequence of \cref{nondecreasing}.

\begin{lemma}\label{equalize-aux}
Let $m\in\mathbb{N}$ and let $f_1,\ldots,f_m:\Rge\rightarrow\Rundefined$ be nondecreasing functions.
Let $\mu_1,\ldots,\mu_m\in\Rge$ s.t.\ $f_1(\mu_1)=f_2(\mu_2)=\cdots=f_m(\mu_m)\in\mathbb{R}$ and let $\mu'_1,\ldots,\mu'_m\in\Rge$ s.t.\ 
$\sum_{j=1}^{m}\mu'_j\ge\sum_{j=1}^{m}\mu_j$.
\begin{parts}
\item\label{equalize-aux-uneven}
If $f_1(\mu'_1),f_2(\mu'_2),\ldots,f_m(\mu'_m)\in\mathbb{R}$, then there exists $j\in[m]$ s.t.\ $f_j(\mu'_j)\ge f_j(\mu_j)$.
\item\label{equalize-aux-even}
If $f_1(\mu'_1)=f_2(\mu'_2)=\cdots=f_m(\mu'_m)\in\mathbb{R}$, then $f_1(\mu'_1)\ge f_1(\mu_1)$.
\end{parts}
\end{lemma}

\begin{proof}
For \cref{equalize-aux-uneven}, since $\sum_{j=1}^{m}\mu'_j\ge\sum_{j=1}^{m}\mu_j$, there exists $j\in[m]$ s.t.\ $\mu'_j\ge\mu_j$. As $f_j$ is nondecreasing and as $\mu_j,\mu'_j\in f_j^{-1}(\mathbb{R})$, we have $f_j(\mu'_j)\ge f_j(\mu_j)$, as required.
For \cref{equalize-aux-even}, by \cref{equalize-aux-uneven} there exists $j\in[m]$ s.t.\ $f_j(\mu'_j)\ge f_j(\mu_j)$; therefore, $f_1(\mu'_1)=f_j(\mu'_j)\ge f_j(\mu_j)=f_1(\mu_1)$.
\end{proof}

\begin{proof}[Proof of \cref{equalize-well-defined-nondecreasing}]
We start by showing that $\equalize{f_1,\ldots,f_m}(\mu)$ is well defined for every $\mu\in\Rge$.
We have to show that if there exist $\mu_1,\ldots,\mu_m\in\Rge$ s.t.\ $\sum_{j=1}^{m}\mu_m=\mu$ and $f_1(\mu_1)=f_2(\mu_2)=\cdots=f_m(\mu_m)\in\mathbb{R}$, then $f_1(\mu'_1)=f_1(\mu_1)$ for every $\mu'_1,\ldots,\mu'_m\in\Rge$ s.t.\ $\sum_{j=1}^{m}\mu'_m=\mu$ and $f_1(\mu'_1)=f_2(\mu'_2)=\cdots=f_m(\mu'_m)\in\mathbb{R}$ as well. This follows directly from \crefpart{equalize-aux}{even}, as both $\sum_{j=1}^{m}\mu_m\le\sum_{j=1}^{m}\mu'_m$ and
$\sum_{j=1}^{m}\mu'_m\le\sum_{j=1}^{m}\mu_m$.

The fact that $\equalize{f_1,\ldots,f_m}$ is nondecreasing follows directly from \crefpart{equalize-aux}{even} as well.
\end{proof}

\begin{proof}[Proof of \cref{intermediate-equalize}]
\cref{intermediate-equalize-idempotent} follows directly by definition, as when $m=1$, we always have $\mu_1=\mu$. We move on to prove \cref{intermediate-equalize-compose};
let $k\in[m]$ and $1\le j_1<j_2<\cdots<j_k<m$; define $j_0\eqdef0$ and $j_{k+1}\eqdef m$.

If $h\eqdef\equalize{f_1,\ldots,f_m}(\mu)\in\mathbb{R}$, then there exist $\mu_1,\ldots,\mu_m\in\Rge$ s.t.\ $\sum_{j=1}^{m}\mu_m=\mu$ and $f_1(\mu_1)=f_2(\mu_2)=\cdots=f_m(\mu_m)=h$. Let $i\in[k+1]$; as $f_{j_{i-1}+1}(\mu_{j_{i-1}+1})=f_{j_{i-1}+2}(\mu_{j_{i-1}+2})=\cdots=f_{j_i}(\mu_{j_i})$, we have that 
$\equalize{f_{{j_{i-1}}+1},\ldots,f_{j_i}}(\sum_{\ell=j_{i-1}+1}^{j_i}\mu_{\ell})=f_{j_{i-1}+1}(\mu_{j_{i-1}+1})=h$.
Hence, we have $h=\equalize{\equalize{f_1,\ldots,f_{j_1}},\equalize{f_{{j_1}+1},\ldots,f_{j_2}},\ldots,\equalize{f_{{j_k}+1},\ldots,f_m}}(\sum_{i=1}^{k+1}\sum_{\ell=j_{i-1}+1}^{j_i}\mu_{\ell})=
\equalize{\equalize{f_1,\ldots,f_{j_1}},\equalize{f_{{j_1}+1},\ldots,f_{j_2}},\ldots,\equalize{f_{{j_k}+1},\ldots,f_m}}(\mu)$, as required.

Conversely, if $h\eqdef\equalize{\equalize{f_1,\ldots,f_{j_1}},\equalize{f_{{j_1}+1},\ldots,f_{j_2}},\ldots,\equalize{f_{{j_k}+1},\ldots,f_m}}(\mu)\in\mathbb{R}$,
then there exist $\tilde{\mu}_1,\ldots,\tilde{\mu}_{k+1}$ s.t.\ $\sum_{i=1}^{k+1}\tilde{\mu}_i=\mu$ and $\equalize{f_{{j_{i-1}}+1},\ldots,f_{j_i}}(\tilde{\mu}_i)=h$ for every $i\in[k+1]$. Therefore, for every $i\in[k+1]$, there exist $\mu_{j_{i-1}+1},\ldots,\mu{j_i}$ s.t.\
$\sum_{\ell=j_{i-1}+1}^{j_i}\mu_{\ell}=\tilde{\mu}_i$ and $f_{j_{i-1}+1}(\mu_{j_{i-1}+1})=\cdots=f_{j_i}(\mu_{j_i})=h$. As $\sum_{j=1}^{m}\mu_m=\sum_{i=1}^{k+1}\tilde{\mu}_i=\mu$ and $h=f_1(\mu_1)=f_2(\mu_2)=\cdots=f_m(\mu_m)$, we have that $\equalize{f_1,\ldots,f_m}(\mu)=h$, as required.
\end{proof}

\begin{proof}[Proof of \cref{equalize-continuous}]
By \cref{intermediate-equalize}, when proving either \lcnamecref{equalize-continuous-continuous} it is enough to consider the case in which $m=2$. (The case $m=1$ follows from \crefpart{intermediate-equalize}{idempotent}, while the case $m>2$ follows from the case $m=2$ by iteratively applying \crefpart{intermediate-equalize}{compose}.)

We start by proving \cref{equalize-continuous-continuous}.
Let $\mu\in\Rge$ s.t.\ $h\eqdef\equalize{f_1,f_2}(\mu)\in\mathbb{R}$ and let $\varepsilon>0$; assume w.l.o.g.\ that $f_1$ is continuous.
By definition of $h$, there exists $\mu_1\in[0,\mu]$ s.t.\ $f_1(\mu_1)=f_2(\mu-\mu_1)=h$. By continuity of $f_1$, there exists $\delta>0$ s.t.\ $|f_1(\mu')-h|<\varepsilon$ for every $\mu'\in(\mu-\delta,\mu+\delta)\cap f_1^{-1}(\mathbb{R})$.
Let $\mu'\in(\mu-\delta,\mu+\delta)\cap \equalize{f_1,f_2}^{-1}(\mathbb{R})$; by definition, there exists $\mu'_1\in[0,\mu']$ s.t.\ $f_1(\mu'_1)=f_2(\mu'-\mu'_1)=h'\eqdef \equalize{f_1,f_2}(\mu')$. If $h'=h$, then we trivially have $|h'-h|=0<\varepsilon$, as required; assume, therefore, that $h'\ne h$.
We show that $\mu'_1\in(\mu_1-\delta,\mu_1+\delta)$ by considering two cases. If $h'>h$, then as $f_1,f_2$ are nondecreasing and as $f_1(\mu_1)=h<h'=f_1(\mu'_1)$ and $f_2(\mu-\mu_1)=h<h'=f_1(\mu'-\mu'_1)$,
we have $\mu_1<\mu'_1$ and $\mu-\mu_1<\mu'-\mu'_1$; combining these, we have that $\mu'_1\in(\mu_1,\mu_1+\mu'-\mu)\subseteq(\mu_1,\mu_1+\delta)\subseteq(\mu_1-\delta,\mu_1+\delta)$ in this case. If $h'<h$, then similarly, as $f_1,f_2$ are nondecreasing and as $f_1(\mu_1)=h>h'=f_1(\mu'_1)$ and $f_2(\mu-\mu_1)=h>h'=f_1(\mu'-\mu'_1)$,
we have $\mu_1>\mu'_1$ and $\mu-\mu_1>\mu'-\mu'_1$; combining these, we have that $\mu'_1\in(\mu_1+\mu'-\mu,\mu_1)\subseteq(\mu_1-\delta,\mu_1)\subseteq(\mu_1-\delta,\mu_1+\delta)$ in this case as well. By definition of $\delta$ and as $f_1(\mu')=h'\in\mathbb{R}$, we obtain $|h'-h|=|f_1(\mu')-h|<\varepsilon$, as required.

We proceed to the proof of \cref{equalize-continuous-suffix}. By \cref{equalize-continuous-continuous}, $\equalize{f_1,f_2}$ is continuous; it therefore remains to show that $\equalize{f_1,f_2}$
is defined on a suffix of $\Rge$. Recall that for every $\mu\in\mathbb{R}$, by definition $\equalize{f_1,f_2}(\mu)\in\mathbb{R}$ iff there exists $\mu_1\in[0,\mu]$ s.t.\  $f_1(\mu_1)=f_2(\mu-\mu_1)\in\mathbb{R}$.
Let $\mu\in\Rge$ s.t.\ $\equalize{f_1,f_2}(\mu)\in\mathbb{R}$; therefore, there exists $\mu_1\in[0,\mu]$ s.t.\ $f_1(\mu_1)=f_2(\mu-\mu_1)\in\mathbb{R}$.
Let $\mu'>\mu$;
note that as $\mu_1\le\mu$, we have $\mu_1+\mu'-\mu\le\mu'$.
Since $f_1(\mu_1),f_2(\mu-\mu_1)\in\mathbb{R}$ and as $\mu'>\mu$, we have, by $f_1$ and $f_2$ being defined on a suffix of $\Rge$, that $f_1(\mu_1+\mu'-\mu),f_2(\mu'-\mu_1)\in\mathbb{R}$ as well. Furthermore, as $f_1$
and $f_2$ are nondecreasing, we have $f_1(\mu_1)=f_2(\mu-\mu_1)\le f_2(\mu'-\mu_1)$ and $f_1(\mu_1+\mu'-\mu)\ge f_1(\mu_1)=f_2(\mu-\mu_1)=f_2\bigl(\mu'-(\mu_1+\mu'-\mu)\bigr)$.
By continuity of $f_1$ and $f_2$ and as $[\mu_1,\mu_1+\mu'-\mu]\subseteq f_1^{-1}(\mathbb{R})$ and $[\mu'-(\mu_1+\mu'-\mu),\mu'-\mu_1]=[\mu-\mu_1,\mu'-\mu_1]\subseteq f_2^{-1}(\mathbb{R})$,
we have by the intermediate value theorem that there exists $\mu'_1\in[\mu_1,\mu_1+\mu'-\mu]\subseteq[0,\mu']$ s.t.\ $f_1(\mu'_1)=f_2(\mu'-\mu'_1)\in\mathbb{R}$, as required.
\end{proof}

\begin{proof}[Proof of \cref{equalize-same-bottom}]
By \crefpart{equalize-continuous}{suffix}, we have that $\equalize{f_1,\ldots,f_m}$ is a real function iff 
$\equalize{f_1,\ldots,f_m}(0)\in\mathbb{R}$, which by definition holds iff
$f_1(0)=f_2(0)=\cdots=f_m(0)$.
\end{proof}

\sloppy
\begin{proof}[Proof of \cref{p-nonempty}]
By definition, $S\notin M_G(S)$ for every $S\in\types$ (by taking $\mu\eqdef\sum_{\type\in\nonemptysubs{S}}\mu^{\type}$).
Therefore, $M_G\bigl(\{1\}\bigr)=\emptyset$; furthermore, by \crefpart{intermediate-equalize}{idempotent}, $E_G\bigl(\{1\}\bigr)=\equalize{f_1}(\mu^{\{1\}})=f_1(\mu^{\{1\}})\in\mathbb{R}$.
Therefore, $\{1\}\in D_G$. In particular, we have that $D_G\ne\emptyset$, and so,
by finiteness of $D_G$, we have that $P_G\ne\emptyset$ and that $h_G\in\mathbb{R}$ is well defined.
\end{proof}
\fussy

\begin{proof}[Proof of \cref{associative-remove}]
Both \lcnamecrefs{associative-remove-empty} of the \lcnamecref{associative-remove} follow straight from definition.
\end{proof}

\subsubsection{Uniqueness and Strength}

\begin{proof}[Proof of \cref{highest-stopping}]
We start by proving \cref{highest-stopping-on}.
Let $\type\in\types$ s.t.\ there exists $j\in P^s$ s.t.\ $s_j(\type)>0$; it is enough to show that $\type\in\nonemptysubs{P^s}$. By definition of $s$, $\height{j}\le\height{k}$ for every $k\in\type$, and as $\height{j}=\Max_{i\in[n]}\height{i}\ge\height{k}$, we have $\height{k}=\height{j}$ and so $k\in P^s$ for every $k\in\type$. Therefore, $\type\in\nonemptysubs{P^s}$ as required.

We proceed to the proof of \cref{highest-stopping-rest}. We first show that $s'$ is a consumption profile in the game $G'\eqdef G-P^s$.
Let $\type'\in\nonemptysubs{[n]\setminus P^s}$.
By definition of $s'$, we have that 
$s'_j(\type')=\sum_{\type\in \mathcal{R}(\type',G')}s_j(\type)\ge 0$ for every $j\in[n]\setminus P^s$;
furthermore, for every $j\in\bigl([n]\setminus P^s\bigr)\setminus\type'$, we have by definition that $j\notin\type$ for every $\type\in \mathcal{R}(\type',G')$, and so $s'_j(\type')=\sum_{\type\in \mathcal{R}(\type',G')}s_j(\type)= 0$.
Finally, we have that
$\sum_{j\in[n]\setminus P^s}s'_j(\type')=
\sum_{j\in[n]\setminus P^s}\sum_{\type\in \mathcal{R}(\type',G')}s_j(\type)=
\sum_{\type\in \mathcal{R}(\type',G')}\sum_{j\in[n]\setminus P^s}s_j(\type)=
\sum_{\type\in \mathcal{R}(\type',G')}\sum_{j\in[n]}s_j(\type)=
\sum_{\type\in \mathcal{R}(\type',G')}\mu^{\type}$, where the penultimate equality is by \cref{highest-stopping-on}.

We move on to show that $\height{j}=\heightt{j}$ for every $j\in[n]\setminus P^s$.
By definition of $s'$, we have for every $j\in[n]\setminus P^s$ that $\loadt{j}=
\sum_{\type'\in\nonemptysubs{[n]\setminus P^s}}s'_j(\type')=
\sum_{\type'\in\nonemptysubs{[n]\setminus P^s}}\sum_{\type\in \mathcal{R}(\type',G')}s_j(\type)=
\sum_{\type\in\types\setminus\nonemptysubs{P^s}}s_j(\type)=
\sum_{\type\in\types}s_j(\type)=
\load{j}$ (where the penultimate equality is since $j\notin\type$ for every $\type\in\nonemptysubs{P^s}$), and hence $\heightt{j}=f_j(\loadt{j})=f_j(\load{j})=\height{j}$, as required.

We conclude by showing that $s'$ is indeed a Nash equilibrium in $G'$.
Let $\type'\in\nonemptysubs{[n]\setminus P^s}$, and let $k\in\supp\bigl(s'(\type')\bigr)$ and $j\in\type'$.
As $0<s'_k(\type')=\sum_{\type\in \mathcal{R}(\type',G')}s_k(\type)$, we have that there exists $\type\in \mathcal{R}(\type',G')$ s.t.\ $k\in\supp\bigl(s(\type)\bigr)$.
As $j\in\type'\subseteq\type$, since $s'$ is a Nash equilibrium in $G$, we have that $\height{k}\le\height{j}$; therefore, $\heightt{k}=\height{k}\le\height{j}=\heightt{j}$ and so $s'$ is a Nash equilibrium in $G'$, as required.

Before proceeding to prove \cref{highest-stopping-who,highest-stopping-where}, we prove a few auxiliary results.
We first show that
\begin{equation}\label{highest-stopping-where-ps}
\forall j\in P^s: \height{j}=E_G(P^s)=\equalize{f_k:k\in P^s}\Bigl(\smashoperator{\sum_{\type\in\nonemptysubs{P^s}}}\mu^{\type}\Bigr).
\end{equation}
By definition of $P^s$, $f_j(\load{j})=\height{j}=\height{k}=f_k(\load{k})$ for every $j,k\in P^s$. Therefore,
$\height{j}=\equalize{f_k:k\in P^s}\bigl(\sum_{k\in P^s}\load{k}\bigr)$ for every $j\in P^s$. It is therefore enough to show that
$\sum_{k\in P^s}\load{k}=\sum_{\type\in\nonemptysubs{P^s}}\mu^{\type}$.
Indeed, we have $\sum_{k\in P^s}\load{k}=
\sum_{k\in P^s}\sum_{\type\in\types}s_k(\type)=
\sum_{\type\in\types}\sum_{k\in P^s}s_k(\type)=
\sum_{\type\in\nonemptysubs{P^s}}\sum_{k\in P^s}s_k(\type)=
\sum_{\type\in\nonemptysubs{P^s}}\mu^{\type}$, where the penultimate equality is by \cref{highest-stopping-on}, and the last equality is because $s(\type)\in\mu^{\type}\cdot\Delta^{\type}\subseteq\mu^{\type}\cdot\Delta^{P^s}$ for every $\type\in\nonemptysubs{P^s}$.

Next, we show that for every $S\in\types$ s.t.\ $E_G(S)=\equalize{f_k:k\in S}\bigl(\sum_{\type\in\nonemptysubs{S}}\mu^{\type}\bigr)\in\mathbb{R}$,
there exists $k\in S$ s.t.\ $f_k(\load{k})\ge E_G(S)$. Indeed,
since $s(\type)\in\mu^{\type}\cdot\Delta^{\type}\subseteq\mu^{\type}\cdot\Delta^{S}$ for every $\type\in\nonemptysubs{S}$, we have that
$\sum_{\type\in\nonemptysubs{S}}\mu^{\type}=\sum_{\type\in\nonemptysubs{S}}\sum_{k\in S}s_k(\type)\le\sum_{\type\in\types}\sum_{k\in S}s_k(\type)=\sum_{k\in S}\sum_{\type\in\types}s_k(\type)=\sum_{k\in S}\load{k}$ and so, by \crefpart{equalize-aux}{uneven}, there exists $k\in S$ s.t.\ $f_k(\load{k})\ge\equalize{f_k:k\in S}\bigl(\sum_{\type\in\nonemptysubs{S}}\mu^{\type}\bigr)=E_G(S)$, as required.

We now show that $P^s\in\arg\Max_{S\in\types}E_G(S)$, where the value $\undefined$ is treated as $-\infty$ for comparisons by the $\Max$ operator.
Let $S\in\types$ s.t.\ $E_G(S)\in\mathbb{R}$.
As shown above, there exists $k\in S$ s.t.\ $f_k(\load{k})\ge E_G(S)$. Therefore, by \cref{highest-stopping-where-ps} and by definition of $P^s$ we obtain that
$E_G(P^s)=
\Max_{j\in[n]}\height{j}\ge\height{k}=f_k(\load{k})\ge E_G(S)$, and so indeed $P^s\in\arg\Max_{S\in\types}E_G(S)$.

Finally, we show that $M_G(P^s)=\emptyset$. We have to show that 
for every $S\in\nonemptysubs{P^s}$ there exists $\mu\le\sum_{\type\in\nonemptysubs{P^s}\setminus\nonemptysubs{P^s\setminus S}}\mu^{\type}$
s.t.\ $\equalize{f_k:k\in S}(\mu)=E_G(P^s)$.
Let, therefore, $S\in\nonemptysubs{P^s}$ and define $\mu\eqdef\sum_{j\in S}\load{j}$. By \cref{highest-stopping-where-ps} and by definition of $P^s$, it is enough to show that both
$\equalize{f_k:k\in S}(\mu)=\Max_{j\in[n]}\height{j}$ and $\mu\le\sum_{\type\in\nonemptysubs{P^s}\setminus\nonemptysubs{P^s\setminus S}}\mu^{\type}$.
Since $S\subseteq P^s$, we have $f_k(\load{k})=\height{k}=\Max_{\in[n]}\height{j}$ for every $k\in S$,
and so, by definition, $\equalize{f_k:k\in S}(\mu)=\equalize{f_k:k\in S}(\sum_{j\in S}\load{j})=\Max_{j\in[n]}\height{j}$. 
For every $j\in S$, we have
$\load{j}=\sum_{\type\in\types}s_j(\type)=\sum_{\type\in\nonemptysubs{P^s}}s_j(\type)=\sum_{\type\in\nonemptysubs{P^s}\setminus\nonemptysubs{P^s\setminus S}}s_j(\type)$,
where the penultimate equality is by \cref{highest-stopping-on} since $j\in S\subseteq P^s$, and the last inequality is since $j\notin\type$ for every $\type\in\nonemptysubs{P^s\setminus S}$.
Therefore,
$\sum_{j\in S}\load{j}=\sum_{j\in S}\sum_{\type\in\nonemptysubs{P^s}\setminus\nonemptysubs{P^s\setminus S}}s_j(\type)=
\sum_{\type\in\nonemptysubs{P^s}\setminus\nonemptysubs{P^s\setminus S}}\sum_{j\in S}s_j(\type)\le
\sum_{\type\in\nonemptysubs{P^s}\setminus\nonemptysubs{P^s\setminus S}}\sum_{j\in [n]}s_j(\type)=
\sum_{\type\in\nonemptysubs{P^s}\setminus\nonemptysubs{P^s\setminus S}}\mu^{\type}$, as required, and so $M_G(P^s)=\emptyset$.

We proceed to prove \cref{highest-stopping-who} by showing mutual containment between the two sides of the equality.

$\subseteq$:
It is enough to show that $P^s\in\arg\Max_{S\in D_G}E_G(S)$.
As $M_G(P^s)=\emptyset$ and as by \cref{highest-stopping-where-ps} $E_G(P^s)\in\mathbb{R}$, we have $P^s\in D_G$.
As $E_G(P^s)=\Max_{S\in\types}E_G(S)\ge\Max_{S\in D_G}E_G(S)$, we therefore have $P^s\in\arg\Max_{S\in D_G}E_G(S)$, as required.

$\supseteq$:
We must show that $S\subseteq P^s$ for every $S\in\arg\Max_{S''\in D_G}E_G(S'')$.
Define $S'\eqdef S\setminus P^s\in2^S$ and assume by way of contradiction that $S'\ne\emptyset$.
It is enough to show that $\equalize{f_k:k\in S'}(\mu)\ne E_G(S)$ for every $\mu\le\sum_{\type\in\nonemptysubs{S}\setminus\nonemptysubs{S\setminus S'}}\mu^{\type}$, since this implies $S'\in M_G(S)$ --- a contradiction, as $S\in D_G$.
Let, therefore, $\mu\le\sum_{\type\in\nonemptysubs{S}\setminus\nonemptysubs{S\setminus S'}}\mu^{\type}$;
as by \cref{p-nonempty}, $E_G(S)=h_G\in\mathbb{R}$, it is enough to show that if $\equalize{f_k:k\in S'}(\mu)\in\mathbb{R}$, then $\equalize{f_k:k\in S'}(\mu)<E_G(S)$.
Recall from the proof of the other direction (``$\subseteq$'') that $P^s\in\arg\Max_{S''\in D_G}E_G(S'')$;
therefore, by definition of $S$, by \cref{highest-stopping-where-ps} and by definition of $P^s$, we obtain that $E_G(S)=E_G(P^s)=\Max_{k\in[n]}\height{k}$. It is thus enough to show that 
$\equalize{f_k:k\in S'}(\mu)<\Max_{k\in[n]}\height{k}$.

By definition of $S'$ and $P^s$, we have that $\height{j}<\Max_{k\in[n]}\height{k}$ for every $j\in S'$ and $\height{j}=\Max_{k\in[n]}\height{k}$ for every $j\in S\setminus S'$;
ergo, $s_j(\type)=0$ for every $j\in S\setminus S'$ and $\type\in\nonemptysubs{S}\setminus\nonemptysubs{S\setminus S'}$.
Hence,
$\sum_{j\in S'}\load{j}=\sum_{j\in S'}\sum_{\type\in\types}s_j(\type)\ge
\sum_{j\in S'}\sum_{\type\in\nonemptysubs{S}\setminus\nonemptysubs{S\setminus S'}}s_j(\type)=
\sum_{\type\in\nonemptysubs{S}\setminus\nonemptysubs{S\setminus S'}}\sum_{j\in S'}s_j(\type)=
\sum_{\type\in\nonemptysubs{S}\setminus\nonemptysubs{S\setminus S'}}\sum_{j\in S}s_j(\type)=\sum_{\type\in\nonemptysubs{S}\setminus\nonemptysubs{S\setminus S'}}\mu^{\type}\ge\mu$. Therefore,
by \crefpart{equalize-aux}{uneven} there exists $j\in S'$ s.t.\ $f_j(\load{j})\ge\equalize{f_k:k\in S'}(\mu)$, and thus
$\equalize{f_k:k\in S'}(\mu)\le f_j(\load{j})=\height{j}<\Max_{k\in[n]}\height{k}$, as required.

We conclude by proving \cref{highest-stopping-where}.
Recall from the proof of the first direction ($\subseteq$) of \cref{highest-stopping-who} that $E_G(P^s)=\Max_{S\in D_G}E_G(S)$.
Therefore, by \cref{highest-stopping-where-ps}, $\height{j}=E_G(P^s)=\Max_{S\in D_G}E_G(S)=h_G$ for every $j\in P^s$, as required.
\end{proof}

\subsubsection{Constrained Distribution}\label{constrained-distribution}

Before proving \cref{construct}, we first formulate and prove a combinatorial result that we use in the proof of this \lcnamecref{construct}.

\begin{definition}[Distribution Constraint]\label{distribute-compatible}
\leavevmode
\begin{itemize}
\item
A \emph{distribution constraint} is a pair $\constraint$, where
$n\in\mathbb{N}$, $\mu^{\type}\in\Rge$ for every $\type\in\types$, and $t_j\le T_j\in\Rge$ for every $j\in[n]$.
\item
We say that a distribution constraint $C=\constraint$ is \emph{satisfiable} if
there exist $(\mu_j^{\type})_{j\in[n]}^{\type\in\types}$ s.t.\ $(\mu_j^{\type})_{j\in[n]}\in\mu^{\type}\cdot\Delta^{\type}$ for every $\type\in\types$ and $\sum_{\type\in\types}\mu_j^{\type}\in[t_j,T_j]$ for every $j\in[n]$.
\item
Given a distribution constraint $C=\constraint$, for every $S\in\types$ we define
$m_C(S)\eqdef\sum_{\type\in\nonemptysubs{S}}\mu^{\type}$, $M_C(S)\eqdef\sum_{\type\in\nonemptysubs{[n]}\setminus\nonemptysubs{[n]\setminus S}}\mu^{\type}$, $t_C(S)\eqdef\sum_{j\in S}t_j$ and $T_C(S)\eqdef\sum_{j\in S}T_j$. We say that $C$ is \emph{normal} if both $t_C(S)\le M_C(S)$ and $m_C(S)\le T_C(S)$ for every $S\in\types$.
\end{itemize}
\end{definition}

We note that it is trivial to show that every satisfiable distribution constraint is normal. In this \lcnamecref{constrained-distribution}, we \emph{constructively} show (without the use of, e.g., linear programming) that the other direction holds as well, and give a procedure for explicitly finding a solution to (i.e., a witness to the satisfiability of) any given normal distribution.

\begin{lemma}\label{distribute-aux}
Every normal distribution constraint is satisfiable.
\end{lemma}

Before proving \cref{distribute-aux}, we first develop some machinery.

\begin{lemma}\label{bad-Mt-Tm}
Let $C=\constraint$ be a normal distribution constraint.
\begin{parts}
\item\label{bad-Mt-Tm-Mt-union}
$M_C(S\cup S')=t_C(S\cup S')$, for every $S,S'\in\types$ s.t.\ $M_C(S)=t_C(S)$ and $M_C(S')=t_C(S')$.
\item\label{bad-Mt-Tm-Tm-intersect}
$m_C(S\cap S')=T_C(S\cap S')$, for every $S,S'\in\types$ s.t.\ $m_C(S)=T_C(S)$, $m_C(S')=T_C(S')$ and $S\cap S'\ne\emptyset$.
\end{parts}
\end{lemma}

\begin{proof}
For every $S,S'\in\types$ s.t.\ $M_C(S)=t_C(S)$ and $M_C(S')=t_C(S')$, we have
$M_C(S\cup S') \le
M_C(S)+M_C(S')-M_C(S\cap S') = t_C(S)+t_C(S')-M_C(S\cap S') \le t_C(S)+t_C(S')-t_C(S\cap S')=t_C(S\cup S')$, as required. (The other side of the inequality follows from normality of $C$.)

For every $S,S'\in\types$ s.t.\ $m_C(S)=T_C(S)$, $m_C(S')=T_C(S')$ and $S\cap S'\ne\emptyset$, we have
$m_C(S\cap S') \ge
m_C(S)+m_C(S')-m_C(S\cup S') = T_C(S)+T_C(S')-m_C(S\cup S') \ge T_C(S)+T_C(S')-T_C(S\cup S')=T_C(S\cap S')$, as required. (Once again, the other side of the inequality follows from normality of $C$.)
\end{proof}

\begin{lemma}[Moving Mass from $\type$ to $\{n\}$]\label{distribute-consume}
Let $C=\constraint$ be a normal distribution constraint.
For every $\type\in\types$ s.t.\ $\{n\}\subsetneq\type$, we define
\[
\may{\type}_C\eqdef
\min\Bigl\{
\smashoperator[r]{\min_{\substack{S\in\nonemptysubs{[n-1]}: \\ S\cap\type\ne\emptyset}}}\:\bigl(M_C(S)-t_C(S)\bigr),
\smashoperator[r]{\min_{\substack{S\in\nonemptysubs{[n]}: \\ \type\not\subseteq S\And n\in S}}}\:\bigl(T_C(S)-m_C(S)\bigr)
\Bigr\}.
\]
\begin{parts}
\item\label{distribute-consume-nonnegative}
$\may{\type}_C\ge0$ for every $\type\in\types$ s.t.\ $\{n\}\subsetneq\type$.
\item\label{distribute-consume-may-if-needed}
If $t_n>\mu^{\{n\}}$, then there exists $\type\in\types$ s.t.\ $\{n\}\subsetneq\type$, $\mu^{\type}>0$ and $\may{\type}_C>0$.
\end{parts}
Let $\type\in\types$ s.t.\ $\{n\}\subsetneq\type$ and let $\mu\in[0,\mu^{\type}]$. For every $\type'\in\types\setminus\bigl\{\type,\{n\}\bigr\}$,
let $\mu'^{\type'}\eqdef\mu^{\type'}$ and let $\mu'^{\type}\eqdef\mu^{\type}-\mu\ge0$ and $\mu'^{\{n\}}\eqdef\mu^{\{n\}}+\mu$. Define $C'\eqdef\bigl((\mu'^{\type})_{\type\in\types},\bigl([t_j,T_j]\bigr)_{j\in[n]}\bigr)$.
\begin{parts}
\setcounter{partsi}{2}
\item\label{distribute-consume-normal}
If $\mu\le\may{\type}_C$, then $C'$ is normal. Furthermore, in this case $\may{\type}_{C'}=\may{\type}_C-\mu$, and $\may{\type'}_{C'}\le\may{\type'}_C$ for every
$\type'\in\types$ s.t.\ $\{n\}\subsetneq\type'$.
\item\label{distribute-consume-satisfiable}
If $C'$ is satisfiable, then $C$ is satisfiable.
\end{parts}
\end{lemma}

\begin{remark}
The condition of \crefpart{distribute-consume}{normal} is actually also necessary; i.e., $C'$ is normal iff $\mu\le\may{\type}_C$.
\end{remark}

\begin{proof}[Proof of \cref{distribute-consume}]
\cref{distribute-consume-nonnegative} follows directly from the fact that $C$ is normal, and so $M_C(S)-t_C(S)\ge0$ and $T_C(S)-m_C(S)\ge0$ for every $S\in\types$.

To prove \cref{distribute-consume-may-if-needed}, let $S_1\eqdef\bigcup\bigl\{S\in\nonemptysubs{[n-1]}~\big|~ M_C(S)=t_C(S)\bigr\}\subseteq\nonemptysubs{[n-1]}$ and $S_2\eqdef[n]\cap\bigcap\bigl\{S\in\nonemptysubs{[n]}~\big|~n\in S\And T_C(S)=m_C(S)\bigr\}\supseteq\{n\}$ (the intersection with $[n]$ has an effect only if $[n]$ is the sole element in the intersection defining $S_2$). We first show that there exists $\type\subseteq S_2\setminus S_1$ s.t.\ $\{n\}\subsetneq\type$ and $\mu^{\type}>0$.

For ease of notation, we extend the definition of $m_C(S)$, $M_C(S)$, $t_C(S)$, and $T_C(S)$ also to the case $S=\emptyset$, via the same definition; we note that these all equal zero when $S=\emptyset$,
as they are all defined by empty sums in this case.
We note that if $S_1\ne\emptyset$, then $M_C(S_1)=t_C(S_1)$ by \crefpart{bad-Mt-Tm}{Mt-union}, and if $S_1=\emptyset$,
then $M_C(S_1)=0=t_C(S_1)$ by definition.

We first consider the case where $S_2\ne[n]$. In this case, by \crefpart{bad-Mt-Tm}{Tm-intersect}, $T_C(S_2)=m_C(S_2)$.
Let $S\eqdef S_1\cap S_2\subseteq[n-1]$. We note that $t_C(S_1)-t_C(S) = t_C(S_1\setminus S) \le M_C(S_1\setminus S) = M_C(S_1) - \sum_{\type\in\types\setminus\nonemptysubs{[n]\setminus S}:\type\cap (S_1\setminus S)=\emptyset}\mu^{\type} \le M_C(S_1)-\sum_{\type\in\nonemptysubs{S_2}\setminus\nonemptysubs{S_2\setminus S}}\mu^{\type}
= t_C(S_1) - \sum_{\type\in\nonemptysubs{S_2}\setminus\nonemptysubs{S_2\setminus S}}\mu^{\type}$; therefore,
$\sum_{\type\in\nonemptysubs{S_2}\setminus\nonemptysubs{S_2\setminus S}}\mu^{\type} \le t_C(S)\le T_C(S)$.
Hence, and as $T_n\ge t_n>\mu^{\{n\}}$, we have that $m_C(S_2) - \sum_{\type\in\nonemptysubs{S_2}\setminus\nonemptysubs{S_2\setminus S}}\mu^{\type} - \mu^{\{n\}}
- \sum_{\type\in\nonemptysubs{S_2\setminus S}:\{n\}\subsetneq\type}\mu^{\type}=
m_C\bigl(S_2\setminus (S\cup\{n\})\bigr) \le T_C\bigl(S_2\setminus (S\cup\{n\})\bigr) = T_C(S_2) - T_C(S) - T_C\bigl(\{n\}\bigr) = m_C(S_2) - T_C(S) - T_n < m_C(S_2) - \sum_{\type\in\nonemptysubs{S_2}\setminus\nonemptysubs{S_2\setminus S}}\mu^{\type} - \mu^{\{n\}}$.
Therefore, $\sum_{\type\in\nonemptysubs{S_2\setminus S}:\{n\}\subsetneq\type}\mu^{\type}>0$, and so there exists
$\type\subseteq S_2\setminus S = S_2\setminus S_1$ s.t.\ $\{n\}\subsetneq\type$
and $\mu^{\type}>0$, as required.

We now consider the case in which $S_2=[n]$. Note that $M_C\bigl(S_1\cup\{n\}\bigr)\ge t_C\bigl(S_1\cup\{n\}\bigr)=t_C(S_1)+t_n=M_C(S_1)+t_n>M_C(S_1)+\mu^{\{n\}} = M_C\bigl(S_1\cup\{n\}\bigr)-\sum_{\type\in\nonemptysubs{[n]\setminus S_1}:\{n\}\subsetneq\type}\mu^{\type}$; therefore, $\sum_{\type\in\nonemptysubs{[n]\setminus S_1}:\{n\}\subsetneq\type}\mu^{\type}>0$,
and so there exists $\type\subseteq [n]\setminus S_1 = S_2\setminus S_1$ s.t.\ $\{n\}\subsetneq\type$ and $\mu^{\type}>0$, as required.

Either way, there exists $\type\subseteq S_2\setminus S_1$ s.t.\ $\{n\}\subsetneq\type$ and $\mu^{\type}>0$.
Therefore, for every $S\in\nonemptysubs{[n-1]}$ s.t.\ $S\cap\type\ne\emptyset$, we have $S\not\subseteq S_1$ and so $M_C(S)\ne t_C(S)$ and by normality of $C$,
$M_C(S)>t_C(S)$; for every $S\in\types$ s.t.\ $n\in S$ and $\type\not\subseteq S$, we have $S_2\not\subseteq S$ and so $T_C(S)\ne m_C(S)$ and by normality of $C$,
$T_C(S)>m_C(S)$. By both of these, $\may{\type}_C>0$
and the proof of \cref{distribute-consume-may-if-needed} is complete.

We move on to \cref{distribute-consume-normal}; let $S\in\types$.
If $\type\subseteq S$ (and so also $n\in S$) or both $\type\not\subseteq S$ and $n\notin S$, then by normality of $C$,
\[
T_{C'}(S) = T_C(S) \ge m_C(S) = \smashoperator[r]{\sum_{\type\in\nonemptysubs{S}}}\mu^{\type} = \smashoperator[r]{\sum_{\type\in\nonemptysubs{S}}}\mu'^{\type} = m_{C'}(S);
\]
otherwise, $\type\not\subseteq S$ and $n\in S$, and by definition of $\mu$ and of $\may{\type}_C$,
\[
T_{C'}(S) = T_C(S) \ge m_C(S)+\may{\type}_C \ge m_C(S)+\mu = \smashoperator[r]{\sum_{\type\in\nonemptysubs{S}}}\mu'^{\type} = m_{C'}(S).
\]
If $S\cap \type=\emptyset$ (and so also $n\notin S$) or both $S\cap \type\ne\emptyset$ and $n\in S$, then by normality of $C$,
\[
t_{C'}(S) = t_C(S) \le M_C(S) = \smashoperator[r]{\sum_{\type\in\nonemptysubs{[n]}\setminus\nonemptysubs{[n]\setminus S}}}\mu^{\type} =
\smashoperator[r]{\sum_{\type\in\nonemptysubs{[n]}\setminus\nonemptysubs{[n]\setminus S}}}\mu'^{\type} = M_{C'}(S);
\]
otherwise, $S\cap \type\ne\emptyset$ and $n\notin S$, and by definition of $\mu$ and of $\may{\type}_C$,
\[
t_{C'}(S) = t_C(S) \le M_C(S) - \may{\type}_C \le M_C(S) - \mu = \smashoperator[r]{\sum_{\type\in\nonemptysubs{[n]}\setminus\nonemptysubs{[n]\setminus S}}}\mu'^{\type} = M_{C'}(S).
\]
Therefore, $C'$ is normal.

For every $S\in\nonemptysubs{[n-1]}$ s.t.\ $S\cap\type\ne\emptyset$, we have that $M_{C'}(S)=M_C(S)-\mu$;
for every $S\in\nonemptysubs{[n]}$ s.t.\ $\type\not\subseteq S$ and $n\in S$, we have that $m_{C'}(S)=m_C(S)+\mu$. Therefore,
\begin{align*}
\may{\type}_{C'}=&\:
\min\Bigl\{
\smashoperator[r]{\min_{\substack{S\in\nonemptysubs{[n-1]}: \\ S\cap\type\ne\emptyset}}}\:\bigl(M_{C'}(S)-t_{C'}(S)\bigr),
\smashoperator[r]{\min_{\substack{S\in\nonemptysubs{[n]}: \\ \type\not\subseteq S\And n\in S}}}\:\bigl(T_{C'}(S)-m_{C'}(S)\bigr)
\Bigr\}=\\
=&\:
\min\Bigl\{
\smashoperator[r]{\min_{\substack{S\in\nonemptysubs{[n-1]}: \\ S\cap\type\ne\emptyset}}}\:\bigl(M_C(S)-\mu-t_C(S)\bigr),
\smashoperator[r]{\min_{\substack{S\in\nonemptysubs{[n]}: \\ \type\not\subseteq S\And n\in S}}}\:\bigl(T_C(S)-m_C(S)-\mu\bigr)
\Bigr\}=\\
=&\:\may{\type}_C-\mu.
\end{align*}

For every $S\in\nonemptysubs{[n-1]}$, we have $M_{C'}(S)\in\bigl\{M_C(S),M_C(S)-\mu\bigr\}$ (as shown above, depending on whether or not both $S\cap\type\ne\emptyset$ and $n\notin S$);
for every $S\in\nonemptysubs{[n]}$ s.t.\ $n\in S$, we have that $m_{C'}(S)\in\bigl\{m_C(S),m_C(S)+\mu\bigr\}$ (as shown above, depending on whether or not both $\type\not\subseteq S$ and $n\in S$). Therefore,
for every $\type'\in\types$ s.t.\ $\{n\}\subsetneq\type'$,
\begin{align*}
\may{\type'}_{C'}=&\:
\min\Bigl\{
\smashoperator[r]{\min_{\substack{S\in\nonemptysubs{[n-1]}: \\ S\cap\type'\ne\emptyset}}}\:\bigl(M_{C'}(S)-t_{C'}(S)\bigr),
\smashoperator[r]{\min_{\substack{S\in\nonemptysubs{[n]}: \\ \type'\not\subseteq S\And n\in S}}}\:\bigl(T_{C'}(S)-m_{C'}(S)\bigr)
\Bigr\}\le\\
\le&\:
\min\Bigl\{
\smashoperator[r]{\min_{\substack{S\in\nonemptysubs{[n-1]}: \\ S\cap\type'\ne\emptyset}}}\:\bigl(M_C(S)-t_C(S)\bigr),
\smashoperator[r]{\min_{\substack{S\in\nonemptysubs{[n]}: \\ \type'\not\subseteq S\And n\in S}}}\:\bigl(T_C(S)-m_C(S)\bigr)
\Bigr\}=\\
=&\:\may{\type'}_C.
\end{align*}
Therefore, the proof of \cref{distribute-consume-normal} is complete.

We conclude by proving \cref{distribute-consume-satisfiable}.
As $C'$ is satisfiable, by definition there exist $(\mu_j'^{\type'})_{j\in[n]}^{\type'\in\types}$ s.t.\ $(\mu_j'^{\type'})_{j\in[n]}\in\mu'^{\type'}\cdot\Delta^{\type'}$ for every $\type'\in\types$ and $\sum_{\type'\in\types}\mu_j'^{\type'}\in[t_j,T_j]$ for every $j\in[n]$.
For every $(j,\type')\in[n]\times\types$, if $j\ne n$ or $\type'\notin\bigl\{\type,\{n\}\bigr\}$,
let $\mu_j^{\type'}\eqdef\mu_j'^{\type'}$; let $\mu_n^{\type}\eqdef\mu_n'^{\type}+\mu$ and $\mu_n^{\{n\}}\eqdef\mu_n'^{\{n\}}-\mu$. (As $\mu_n'^{\{n\}}=\mu'^{\{n\}}=\mu^{\{n\}}+\mu$, we have that $\mu_n^{\{n\}}\in\Rge$.)

For every $\type'\in\types\setminus\bigl\{\type,\{n\}\bigr\}$, by definition $(\mu_j^{\type'})_{j\in[n]}=(\mu_j'^{\type'})_{j\in[n]}\in\mu'^{\type'}\cdot\Delta^{\type'}=\mu^{\type'}\cdot\Delta^{\type'}$.
Furthermore, as $(\mu_j'^{\type})_{j\in[n]}\in\mu'^{\type}\cdot\Delta^{\type}$ and by definition of $(\mu_j^{\type})_{j\in[n]}$ and as $n\in\type$,
we have that $(\mu_j^{\type})_{j\in[n]}\in(\mu'^{\type}+\mu)\cdot\Delta^{\type}=\mu^{\type}\cdot\Delta^{\type}$.
Similarly, as $(\mu_j'^{\{n\}})_{j\in[n]}\in\mu'^{\{n\}}\cdot\Delta^{\{n\}}$ and by definition of $(\mu_j^{\{n\}})_{j\in[n]}$,
we have that $(\mu_j^{\{n\}})_{j\in[n]}\in(\mu'^{\{n\}}-\mu)\cdot\Delta^{\{n\}}=\mu^{\{n\}}\cdot\Delta^{\{n\}}$.
Finally, $\sum_{\type'\in\types}\mu_j^{\type'}=\sum_{\type'\in\types}\mu_j'^{\type'}\in[t_j,T_j]$ for every $j\in[n]$, and the proof is complete.
\end{proof}

\begin{lemma}[Distributing All Mass but $\mu^{\{n\}}$ among ${[n-1]}$]\label{distribute-remove}
Let $C=\constraint$ be a normal distribution constraint s.t.\ $\mu^{\{n\}}\ge t_n$.
We say that \emph{condition $D_C$ holds} if  $T_C(S)\ge m_C\bigl(S\cup\{n\}\bigr)-\mu^{\{n\}}$ for every $S\in\nonemptysubs{[n-1]}$.
\begin{parts}
\item\label{distribute-remove-nonnegative}
If $T^n=\mu^{\{n\}}$, then condition $D_C$ holds.
\item\label{distribute-remove-may-if-needed}
If $\may{\type}_C=0$ for every $\type\in\types$ s.t. $\{n\}\subsetneq\type$ and $\mu^{\type}>0$, then condition $D_C$ holds.
\end{parts}
For every $\type\in\nonemptysubs{[n-1]}$,
let $\mu'^{\type}\eqdef\mu^{\type}+\mu^{\type\cup\{n\}}$. Define $C'\eqdef\bigl((\mu'^{\type})_{\type\in\nonemptysubs{[n-1]}},\bigl([t_j,T_j]\bigr)_{j\in[n-1]}\bigr)$.
\begin{parts}
\setcounter{partsi}{2}
\item\label{distribute-remove-normal}
If condition $D_C$ holds, then $C'$ is normal.
\item\label{distribute-remove-satisfiable}
If $C'$ is satisfiable, the $C$ is satisfiable.
\end{parts}
\end{lemma}

\begin{remark}
Once again, the condition of \crefpart{distribute-remove}{normal} is actually also necessary; i.e., $C'$ is normal iff condition $D_C$ holds.
\end{remark}

\begin{proof}[Proof of \cref{distribute-remove}]
\cref{distribute-remove-nonnegative} holds as for every $S\in\nonemptysubs{[n-1]}$, $T_C(S)=T_C\bigl(S\cup\{n\}\bigr)-T^n\ge m_C\bigl(S\cup\{n\}\bigr)-T^n=m_C\bigl(S\cup\{n\}\bigr)-\mu^{\{n\}}$.

To prove \cref{distribute-remove-may-if-needed}, define $S_1$ and $S_2$ as in the proof of \crefpart{distribute-consume}{may-if-needed};
as in that proof, it suffices to show that if condition $D_C$ does not hold, then there exists $\type\subseteq S_2\setminus S_1$ s.t.\ $\{n\}\subsetneq\type$ and $\mu^{\type}>0$.
As in that proof, we extend the definition of $m_C(S)$, $M_C(S)$, $t_C(S)$ and $T_C(S)$ also to the case $S=\emptyset$.
By \cref{distribute-remove-nonnegative}, $T_n>\mu^{\{n\}}$. If $S_2\ne[n]$, then the proof follows as in the proof of \crefpart{distribute-consume}{may-if-needed} (as that proof only uses the fact that $\mu^{\{n\}}<T_n$ when $S_2\ne[n]$, and does not rely on the inequality $\mu^{\{n\}}<t_n$ for this case). It therefore remains to consider the case in which $S_2=[n]$.

Recall that if $S_1\ne\emptyset$, then $M_C(S_1)=t_C(S_1)$ by \crefpart{bad-Mt-Tm}{Mt-union}, and if $S_1=\emptyset$,
then $M_C(S_1)=0=t_C(S_1)$ by definition.
As condition $D_C$ does not hold, there exists $S\in\nonemptysubs{[n-1]}$ s.t.\ $T_C(S) < m_C\bigl(S\cup\{n\}\bigr)-\mu^{\{n\}}$.

Let $S'\eqdef S_1\cap S\subseteq[n-1]$. We note that $t_C(S_1)-t_C(S') = t_C(S_1\setminus S') \le M_C(S_1\setminus S') = M_C(S_1) - \sum_{\type\in\types\setminus\nonemptysubs{[n]\setminus S'}:\type\cap (S_1\setminus S')=\emptyset}\mu^{\type} \le M_C(S_1)-\sum_{\type\in\nonemptysubs{S\cup\{n\}}\setminus\nonemptysubs{S\cup\{n\}\setminus S'}}\mu^{\type}
= t_C(S_1) - \sum_{\type\in\nonemptysubs{S\cup\{n\}}\setminus\nonemptysubs{S\cup\{n\}\setminus S'}}\mu^{\type}$; therefore,
$\sum_{\type\in\nonemptysubs{S\cup\{n\}}\setminus\nonemptysubs{S\cup\{n\}\setminus S'}}\mu^{\type} \le t_C(S') \le T_C(S')$.
Hence, by definition of $S$ we have that
$m_C\bigl(S\cup\{n\}\bigr) - \sum_{\type\in\nonemptysubs{S\cup\{n\}}\setminus\nonemptysubs{(S\cup\{n\})\setminus S'}}\mu^{\type} - \mu^{\{n\}}
- \sum_{\type\in\nonemptysubs{(S\cup\{n\})\setminus S'}:\{n\}\subsetneq\type}\mu^{\type}=
m_C(S\setminus S') \le T_C(S\setminus S') = T_C(S) - T_C(S') < m_C\bigl(S\cup\{n\}\bigr) - T_C(S') - \mu^{\{n\}} \le m_C\bigl(S\cup\{n\}\bigr) - \sum_{\type\in\nonemptysubs{S\cup\{n\}}\setminus\nonemptysubs{(S\cup\{n\})\setminus S'}}\mu^{\type} - \mu^{\{n\}}$.
Therefore, $\sum_{\type\in\nonemptysubs{(S\cup\{n\})\setminus S'}:\{n\}\subsetneq\type}\mu^{\type}>0$, and so there exists
$\type\subseteq (S\cup\{n\})\setminus S' \subseteq S_2\setminus S_1$ s.t.\ $\{n\}\subsetneq\type$
and $\mu^{\type}>0$, as required, and the proof of \cref{distribute-remove-may-if-needed} is complete.

We move on to \cref{distribute-remove-normal}. For every $S\in\nonemptysubs{[n-1]}$, as condition $D_C$ holds,
\[
T_{C'}(S) = T_C(S) \ge m_C\bigl(S\cup\{n\}\bigr)-\mu^{\{n\}} = \smashoperator[r]{\sum_{\type\in\nonemptysubs{S}}}(\mu^{\type}+\mu^{\type\cup\{n\}}) = m_{C'}(S);
\]
furthermore,
\[
t_{C'}(S) = t_C(S) \le M_C(S) = \smashoperator[r]{\sum_{\type\in\nonemptysubs{[n]}\setminus\nonemptysubs{[n]\setminus S}}}\mu^{\type} =
\smashoperator[r]{\sum_{\type\in\nonemptysubs{[n-1]}\setminus\nonemptysubs{[n-1]\setminus S}}}(\mu^{\type}+\mu^{\type\cup\{n\}}) = M_{C'}(S).
\]
Therefore, $C'$ is normal, as required.

We conclude by proving \cref{distribute-remove-satisfiable}.
As $C'$ is satisfiable, by definition there exist $(\mu_j'^{\type})_{j\in[n-1]}^{\type'\in\nonemptysubs{[n-1]}}$ s.t.\ $(\mu_j'^{\type})_{j\in[n-1]}\in\mu'^{\type}\cdot\Delta^{\type}$ for every $\type\in\types$ and $\sum_{\type\in\nonemptysubs{[n-1]}}\mu_j'^{\type}\in[t_j,T_j]$ for every $j\in[n-1]$.
For every $\type\in\nonemptysubs{[n-1]}$, if $\mu'^{\type}=0$, then we define $\mu_j^{\type}\eqdef0$ and $\mu_j^{\type\cup\{n\}}\eqdef0$ for every $j\in[n-1]$;
otherwise, we define $\mu_j^{\type}\eqdef\frac{\mu^{\type}}{\mu'^{\type}}\cdot\mu_j'^{\type}$ and $\mu_j^{\type\cup\{n\}}\eqdef\frac{\mu^{\type\cup\{n\}}}{\mu'^{\type}}\cdot\mu_j'^{\type}$ for every $j\in[n-1]$. We further define $\mu_n^{\type}\eqdef0$ for every $\type\in\types\setminus\{n\}$, 
$\mu_n^{\{n\}}\eqdef\mu^{\{n\}}$ and $\mu_j^{\{n\}}\eqdef 0$ for every $j\in[n-1]$.

For every $\type\in\nonemptysubs{[n-1]}$, if $\mu'^{\type}=0$ then by definition
$(\mu_j^{\type})_{j\in[n]}\equiv0\in0\cdot\Delta^{\type}=\mu^{\type}\cdot\Delta^{\type}$ and similarly
$(\mu_j^{\type\cup\{n\}})_{j\in[n]}\equiv0\in0\cdot\Delta^{\type}=\mu^{\type\cup\{n\}}\cdot\Delta^{\type\cup\{n\}}$; otherwise.
as $(\mu_j'^{\type})_{j\in[n-1]}\in\mu'^{\type}\cdot\Delta^{\type}$ and by definition of $(\mu_j^{\type})_{j\in[n]}$
and $(\mu_j^{\type\cup\{n\}})_{j\in[n]}$, we have that $(\mu_j^{\type})_{j\in[n]}\in\frac{\mu^{\type}}{\mu'^{\type}}\mu'^{\type}\cdot\Delta^{\type}=\mu^{\type}\cdot\Delta^{\type}$ and similarly
$(\mu_j^{\type\cup\{n\}})_{j\in[n]}\in\frac{\mu^{\type\cup\{n\}}}{\mu'^{\type}}\mu'^{\type}\cdot\Delta^{\type}=\mu^{\type\cup\{n\}}\cdot\Delta^{\type}\subseteq\mu^{\type\cup\{n\}}\cdot\Delta^{\type\cup\{n\}}$.
Furthermore, by definition $(\mu_j^{\{n\}})_{j\in[n]}\in\mu^{\{n\}}\cdot\Delta^{\{n\}}$.
Finally, $\sum_{\type\in\types}\mu_j^{\type}=\sum_{\type\in\nonemptysubs{[n-1]}}\mu_j'^{\type}\in[t_j,T_j]$ for every $j\in[n-1]$,
and $\sum_{\type\in\types}\mu_n^{\type}=\mu_n^{\{n\}}=\mu^{\{n\}}\in[t_n,T_n]$ (where $\mu^{\{n\}}=m_C\bigl(\{n\}\bigr)\le T_C\bigl(\{n\}\bigr)=T_n$ by normality of $C$)
and the proof is complete.
\end{proof}

\begin{proof}[Proof of \cref{distribute-aux}]
Let $C=\constraint$ be a normal distribution constraint. We prove the claim by induction on $n\in\mathbb{N}$.

(Outer induction) Base: For $n=1$, we have by definition that $m_C\bigl(\{1\}\bigr)=\mu^{\{1\}}=M_C\bigl(\{1\}\bigr)$, and so $t_1=t_C\bigl(\{1\}\bigr)\le M_C\bigl(\{1\}\bigr)=\mu^{\{1\}}=m_C\bigl(\{1\}\bigr)\le T_C\bigl(\{1\}\bigr)=T_1$. Therefore, setting $\mu^{\{1\}}_1\eqdef\mu^{\{1\}}$ completes the proof of the (outer) induction base.

(Outer induction) Step: Let $n>1$ and assume that the \lcnamecref{distribute-aux} holds for $n-1$.
We prove the induction step by full induction on $\Bigl|\bigl\{\type\in\types~\big|~\{n\}\subsetneq\types\And\mu^{\type}>0\And \may{\type}_C>0\bigr\}\Bigr|$.

(Inner induction) Base: If $\Bigl|\bigl\{\type\in\types~\big|~\{n\}\subsetneq\types\And\mu^{\type}>0\And \may{\type}_C>0\bigr\}\Bigr|=0$, then by
\crefpart{distribute-consume}{may-if-needed}, $t_n\le\mu^{\{n\}}$, and by \crefpart{distribute-remove}{may-if-needed},
condition $D_C$ holds. Therefore, by \crefpart{distribute-remove}{normal}, $C'$ as defined in \cref{distribute-remove} is normal, and by the (outer) induction hypothesis for $n-1$, $C'$ is satisfiable. By \crefpart{distribute-remove}{satisfiable}, $C$ is satisfiable as well.

(Inner induction) Step: Assume that $\Bigl|\bigl\{\type\in\types~\big|~\{n\}\subsetneq\type\And\mu^{\type}>0\And \may{\type}_C>0\bigr\}\Bigr|>0$ and that the claim holds whenever this set
is of smaller cardinality.
Therefore, there exists $\type\in\types$ s.t.\ $\{n\}\subsetneq\type$, $\mu^{\type}>0$ and $\may{\type}_C>0$;
let $\mu\eqdef\min\{\may{\type}_C,\mu^{\type}\}>0$, and define $C'$ w.r.t.\ $\type$ and $\mu$ as in \cref{distribute-consume}. By \crefpart{distribute-consume}{normal}, $C'$ is normal.
If $\mu=\mu^{\type}$, then $\mu'^{\type}=0$; otherwise, $\mu=\may{\type}_C$ and by \crefpart{distribute-consume}{normal}, $\may{\type}_{C'}=\may{\type}_C-\mu=0$. Either way, and by
definition of $C'$ and as by \crefpart{distribute-consume}{normal} $\may{\type'}_{C'}=0$ whenever $\may{\type'}_C=0$, we have that
$\Bigl|\bigl\{\type'\in\types~\big|~\{n\}\subsetneq\type'\And\mu'^{\type'}>0\And \may{\type'}_{C'}>0\bigr\}\Bigr|\le\Bigl|\bigl\{\type'\in\types~\big|~\{n\}\subsetneq\type'\And\mu^{\type'}>0\And \may{\type'}_C>0\bigr\}\setminus\{\type\}\Bigr|=\Bigl|\bigl\{\type'\in\types~\big|~\{n\}\subsetneq\type'\And\mu^{\type'}>0\And \may{\type'}_C>0\bigr\}\Bigr|-1$ and so, by the (inner) induction hypothesis, $C'$ is satisfiable. By \crefpart{distribute-consume}{satisfiable}, $C$ is satisfiable as well and the proof is complete.
\end{proof}

\subsubsection{Existence}

\begin{lemma}[$h_G=\Max_{S\in\types}E_G(S)$]\label{max-on-all}
Let $G=\game$ be a resource selection game.
\begin{parts}
\item\label{max-on-all-a-max-r}
$\Max_{S\in\types}E_G(S)\in\mathbb{R}$ is well defined.
\item\label{max-on-all-eq}
If $f_1,\ldots,f_n$ are continuous, then
$h_G=\Max_{S\in\types}E_G(S)$.
\end{parts}
In both \lcnamecrefs{max-on-all-a-max-r}, the value $\undefined$ is treated as $-\infty$ for comparisons by the $\Max$ operator.
\end{lemma}

\begin{proof}
To show \cref{max-on-all-a-max-r}, note that by \crefpart{intermediate-equalize}{idempotent}, $E_G\bigl(\{1\}\bigr)=\equalize{f_1}(\mu^{\{1\}})=f_1\bigl(\mu^{\{1\}}\bigr)\in\mathbb{R}$;
therefore, $E_G(\{1\})\in\mathbb{R}$, and so $\Max_{S\in\types}E_G(S)\in\mathbb{R}$, as required.

Define $A\eqdef\arg\Max_{S\in\types}E_G(S)$.
Before proving \cref{max-on-all-eq}, we first show that for every $S\in A$ and $M'\in M_G(S)$, if $(f_j)_{j\in S\setminus M'}$ are continuous, then also $S\setminus M' \in A$.
Let, therefore, $S\in A$ and let $M'\in M_G(S)$ s.t.\ $(f_j)_{j\in S\setminus M'}$ are continuous. By definition of $M'$, we have both that $M'\subsetneq S$ (see \cref{p-nonempty})
and that
$\equalize{f_k:k\in M'}(\mu)\ne E_G(S)$
for every $\mu\le\sum_{\type\in\nonemptysubs{S}\setminus\nonemptysubs{S\setminus M'}}\mu^{\type}$.
By \cref{max-on-all-a-max-r}, $E_G(S)\in\mathbb{R}$, and so by definition there exist $(\mu_j)_{j\in S}\in\Rge^S$ s.t.\ $\sum_{j\in S}\mu_j=\sum_{\type\in\nonemptysubs{S}}\mu^{\type}$ and $f_j(\mu_j)=E_G(S)$ for every $j\in S$. Therefore,
$\equalize{f_k:k\in S\setminus M'}\bigl(\sum_{j\in S\setminus M'}\mu_j\bigr)=E_G(S)$, and also
$\equalize{f_k:k\in M'}(\mu)=E_G(S)$,
where $\mu\eqdef\sum_{j\in M'}\mu_j$.

As $M'\in M_G(S)$, we thus have that $\mu>\sum_{\type\in\nonemptysubs{S}\setminus\nonemptysubs{S\setminus M'}}\mu^{\type}$.
Therefore, $\sum_{\type\in\nonemptysubs{S\setminus M'}}\mu^{\type}=\sum_{\type\in\nonemptysubs{S}}\mu^{\type}-\sum_{\type\in\nonemptysubs{S}\setminus\nonemptysubs{S\setminus M'}}\mu^{\type}>
\sum_{\type\in\nonemptysubs{S}}\mu^{\type}-\mu=
\sum_{j\in S}\mu_j-\mu=\sum_{j\in S}\mu_j-\sum_{j\in M'}\mu_j=\sum_{j\in S\setminus M'}\mu_j$.
Recall that $\equalize{f_k:k\in S\setminus M'}\bigl(\sum_{j\in S\setminus M'}\mu_j\bigr)=E_G(S)\in\mathbb{R}$; therefore, by continuity of $(f_k)_{k\in S\setminus M'}$ and by \crefpart{equalize-continuous}{suffix}, we obtain that also
$\equalize{f_k:k\in S\setminus M'}\bigl(\sum_{\type\in\nonemptysubs{S\setminus M'}}\mu^{\type}\bigr)\in\mathbb{R}$. Hence, by \cref{equalize-well-defined-nondecreasing} we conclude that
$E_G(S\setminus M')=\equalize{f_k:k\in S\setminus M'}\bigl(\sum_{\type\in\nonemptysubs{S\setminus M'}}\mu^{\type}\bigr)\ge\equalize{f_k:k\in S\setminus M'}\bigl(\sum_{j\in S\setminus M'}\mu_j\bigr)=E_G(S)$,
and so indeed $S\setminus M' \in A$, as required.

We conclude by proving \cref{max-on-all-eq}. By definition $h_G\le\Max_{S\in\types}E_G(S)$; we therefore have to show that $h_G\ge\Max_{S\in\types}E_G(S)$. Let $S\in A$. We sequentially define a series $(S_i)_{i=0}^{k}$, for $k\in\mathbb{N}$ to be determined, as follows:
\begin{itemize}
\item
$S_0\eqdef S$.
\item
If $M_G(S_i)=\emptyset$, then we set $k\eqdef i$ and conclude.
Otherwise, choose $M_i\in M_G(S_i)$ arbitrarily, and set $S_{i+1}\eqdef S_i\setminus M_i$.
\end{itemize}
We now show by induction that $S_i\in A$ and $|S_i|\le |S|-i$ for every $i$ for which $S_i$ is defined.
\begin{itemize}
\item
Base: By definition, $S_0=S\in A$ and $|S_0|\le|S_0|-0=|S|-0$, as required.
\item
Step: Let $i>0$ for which $S_i$ is defined. By the induction hypothesis, $S_{i-1}\in A$; therefore, as shown above and by continuity of $(f_j)_{j\in S_{i-1}\setminus M_{i-1}}$, we have that $S_i=S_{i-1}\setminus M_{i-1}\in A$.
Furthermore, as by definition $M_{i-1}\ne\emptyset$, we have by the induction hypothesis that $|S_i|=|S_{i-1}|-|M_{i-1}|\le |S_{i-1}|-1 \le |S|-(i-1)-1 = |S|-i$, as required.
\end{itemize}
We conclude that the process constructing $(S_i)_i$ indeed stops (i.e., $k$ is well defined), and with $k<|S|$.
By definition, $M_G(S_k)=\emptyset$, and as $S_k\in A$, by \cref{max-on-all-a-max-r} we have $E_G(S_k)=\Max_{S'\in\types}E_G(S')\in\mathbb{R}$; therefore, $S_k\in D_G$.
Therefore, $h_G\ge E_G(S_k)=\Max_{S'\in\types}E_G(S')$, as required.
\end{proof}

We note that it can be shown that, in the context of \crefpart{max-on-all}{eq}, for every $S\in A$ s.t.\ $M_G(S)\ne\emptyset$, in fact $\bigcup M_G(S)\in M_G(S)$ and $M_G\bigl(S\setminus \bigcup M_G(S)\bigr)=\emptyset$. While this may be used to avoid the inductive construction concluding the proof of this \lcnamecref{max-on-all}, the need to prove these facts would result in a considerably longer total length for the proof.

\begin{lemma}\label{distribute}
Let $G=\game$ be a resource selection game s.t.\ $f_1,\ldots,f_n$ are continuous.
For every $S\in\arg\Max_{S'\in D_G}E_G(S')$, there exists a strategy profile $s$ in the $|S|$-resource selection game $G'\eqdef\bigl((f_j)_{j\in S};(\mu^{\type})_{\type\in\nonemptysubs{S}}\bigr)$,
s.t.\ $\height{j}=h_G$ for every $j\in S$.
\end{lemma}

\begin{proof}
For every $j\in S$, let $t_j\eqdef\min f_j^{-1}(h_G)$ and $T_j\eqdef\Max f_j^{-1}(h_G)$ if $\sup f_j^{-1}(h_G)\ne\infty$ and $T_j\eqdef\sum_{\type\in\types}\mu^{\type}$ otherwise ($t_j$ and $T_j$ are well defined by continuity of $f_j$ and since $E_G(S)=h_G$); regardless of how we define $T_j$, we have that both $f_j(T_j)=h_G$ and $T_j\ge t_j$ (when
$\sup f_j^{-1}(h_G)=\infty$, this is since $E_G(S)=h_G$ and since
$f_j$ is nondecreasing).

We now show that $C\eqdef\bigl((\mu^{\type})_{\type\in\nonemptysubs{S}},\bigl([t_j,T_j]\bigr)_{j\in S}\bigr)$ is a normal distribution constraint. (See \cref{constrained-distribution}; we slightly abuse notation
by treating $S$ in the context of $C$ as $\bigl[|S|\bigr]$, using an arbitrary isomorphism.)
Let $S'\in\nonemptysubs{S}$. As $M(S)=\emptyset$, there exists
$\mu\le\sum_{\type\in\nonemptysubs{S}\setminus\nonemptysubs{S\setminus S'}}\mu^{\type}$ s.t.\ $\equalize{f_k:k\in S'}(\mu)=E_G(S)=h_G$; therefore,
$t_C(S')=\sum_{j\in S'} t_j\le\mu\le\sum_{\type\in\nonemptysubs{S}\setminus\nonemptysubs{S\setminus S'}}\mu^{\type}=M_C(S')$.
Assume by way of contradiction that $\sum_{\type\in\nonemptysubs{S'}}\mu^{\type}>\sum_{j\in S'}T_j$. As $f_j(T_j)=h_G$ for every $j\in S'$, we have
$\equalize{f_k:k\in S'}\bigl(\sum_{j\in S'}T_j\bigr)=h_G\in\mathbb{R}$. By continuity of $(f_j)_{j\in S'}$ and by \crefpart{equalize-continuous}{suffix}, we thus have that
$E_G(S')\in\mathbb{R}$; therefore, there exist $(\mu_j)_{j\in S'}$ s.t.\ $\sum_{j\in S'}\mu_j=\sum_{\type\in\nonemptysubs{S'}}\mu^{\type}$ and
$f_j(\mu_j)=E_G(S')$ for every $j\in S'$. As $\sum_{j\in S'}\mu_j=\sum_{\type\in\nonemptysubs{S'}}\mu^{\type}>\sum_{j\in S'}T_j$, there exists $j\in S'$ s.t.\ $\mu_j>T_j$;
As $T_j<\mu_j\le\sum_{\type\in\nonemptysubs{S'}}\mu^{\type}\le\sum_{\type\in\types}\mu^{\type}$, we have that $T_j=\Max f_j^{-1}(h_G)$, and so $E_G(S')=f_j(\mu_j)>h_G=\Max_{S''\in\types}E_G(S'')$ (where the last equality is by \crefpart{max-on-all}{eq}) --- a contradiction. Therefore,
$m_C(S')=\sum_{\type\in\nonemptysubs{S'}}\mu^{\type}\le\sum_{j\in S'}T_j=T_C(S')$.

As $C$ is normal, by \cref{distribute-aux} it is satisfiable, and so there exist $\bigl(s_j(\type)\bigr)_{j\in S}^{\type\in\nonemptysubs{S}}$ s.t.\ $s(\type)\in\mu^{\type}\cdot\Delta^{\type}$ for every $\type\in\nonemptysubs{S}$ and $\sum_{\type\in\nonemptysubs{S}}s_j(\type)\in[t_j,T_j]$ for every $j\in S$. By the former, $s$ is a strategy profile in $G'$, and by the latter, for every $j\in S$ we have that
$\load{j}\in[t_j,T_j]$, and so by definition of $t_j$ and $T_j$ and since $f_j$ is nondecreasing, $\height{j}=f_j(\load{j})=h_G$ and the proof is complete.
\end{proof}

\begin{lemma}[$E_G(P_G)=h_G$]\label{p-in-argmax}
Let $G=\game$ be a resource selection game.
If $f_1,\ldots,f_n$ are continuous, then
$P_G\in\arg\Max_{S\in D_G}E_G(S)$.
\end{lemma}

\begin{proof}
Define $A\eqdef\arg\Max_{S\in D_G}E_G(S)$.
By \cref{p-nonempty}, $A\ne\emptyset$.
Therefore, by definition of $P_G$, it is enough to show that $S'\cup S''\in A$ for every $S',S''\in A$.
Let, therefore, $S',S''\in A$; it is enough to show that $S'\cup S''\in D_G$ and that $E_G(S'\cup S'')=h_G$.

By \cref{distribute}, there exists a strategy profile $s'$ in the game $\bigl((f_j)_{j\in S'};(\mu^{\type})_{\type\in\nonemptysubs{S'}}\bigr)$
s.t.\ $\heightt{j}=h_G$ for every $j\in S'$; similarly, there exists a strategy profile $s''$ in the game $\bigl((f_j)_{j\in S''};(\mu^{\type})_{\type\in\nonemptysubs{S''}}\bigr)$ s.t.\ $\heighttt{j}=h_G$ for every $j\in S''$.
For every $\type\in\nonemptysubs{S'}$, we define $s(\type)\eqdef s'(\type)$ and set $\tilde{\mu}^{\type}\eqdef\mu^{\type}$; for every $\type\in\nonemptysubs{S''}\setminus\nonemptysubs{S'}$, we define $s_j(\type)\eqdef s''_j(\type)$ for every $j\in S''\setminus S'$ and $s_j(\type)\eqdef0$ for every $j\in S'\cap S''$
and set $\tilde{\mu}^{\type}\eqdef\mu^{\type}-\sum_{j\in S'\cap S''}s_j(\type)$; finally, for every $\type\in\nonemptysubs{S'\cup S''}\setminus(\nonemptysubs{S'}\cup\nonemptysubs{S''})$,
we define $s(\type)\eqdef0$ and set $\tilde{\mu}^{\type}\eqdef0$. By definition, $s$ is a consumption profile in the game $\bigl((f_j)_{j\in S'\cup S''}; (\tilde{\mu}^{\type})_{\type\in\nonemptysubs{S'\cup S''}}\bigr)$. Note that for every $\type\in\nonemptysubs{S'\cup S''}$, we have that $\tilde{\mu}^{\type}\le\mu^{\type}$.

For every $j\in S'$, by definition of $s$ we have $\load{j}=\loadt{j}$ and so $f_j(\load{j})=f_j(\loadt{j})=\heightt{j}=h_G$. For every $j\in S''\setminus S'$, by definition we have $\load{j}=\loadtt{j}$ and so $f_j(\load{j})=f_j(\loadtt{j})=\heighttt{j}=h_G$. Therefore, $h_G=\equalize{f_k:k\in S'\cup S''}\bigl(\sum_{j\in S'\cup S''}\load{j}\bigr)=\equalize{f_k:k\in S'\cup S''}\bigl(\sum_{\type\in\nonemptysubs{S'\cup S''}}\tilde{\mu}^{\type}\bigr)$.
As $\sum_{\type\in\nonemptysubs{S'\cup S''}}\tilde{\mu}^{\type}\le\sum_{\type\in\nonemptysubs{S'\cup S''}}\mu^{\type}$, we have by continuity of $(f_k)_{k\in S'\cup S''}$,
by \cref{p-nonempty} and by \crefpart{equalize-continuous}{suffix} that $E_G(S'\cup S'')\in\mathbb{R}$.
Therefore, by \cref{equalize-well-defined-nondecreasing}, we obtain $E_G(S'\cup S'')\ge\equalize{f_k:k\in S'\cup S''}\bigl(\sum_{\type\in\nonemptysubs{S'\cup S''}}\tilde{\mu}^{\type}\bigr)=h_G$.
Thus, by \crefpart{max-on-all}{eq}, $E_G(S'\cup S'')=h_G$. It therefore remains to show that $S'\cup S''\in D_G$.

We have to show that for every $S\in\nonemptysubs{S'\cup S''}$, there exists $\mu\le\sum_{\type\in\nonemptysubs{S'\cup S''}\setminus\nonemptysubs{(S'\cup S'')\setminus S}}\mu^{\type}$ s.t.\ $\equalize{f_k:k\in S}(\mu)=E_G(S'\cup S'')$; let therefore $S\in\nonemptysubs{S'\cup S''}$.
Define $\mu\eqdef\sum_{j\in S}\load{j}$. As $f_j(\load{j})=h_G$ for every $j\in S'\cup S''$, we have $\equalize{f_k:k\in S}(\mu)=h_G=E_G(S'\cup S'')$.
By definition of $s$, we have that $\mu=\sum_{j\in S}\load{j}\le\sum_{\type\in\nonemptysubs{S'\cup S''}\setminus\nonemptysubs{(S'\cup S'')\setminus S}}\tilde{\mu}^{\type}\le
\sum_{\type\in\nonemptysubs{S'\cup S''}\setminus\nonemptysubs{(S'\cup S'')\setminus S}}\mu^{\type}$, and the proof is complete.
\end{proof}

\begin{proof}[Proof of \cref{construct}]
\cref{construct-distribute} follows directly from \cref{distribute,p-in-argmax}.
We therefore prove \cref{construct-stopping-order}. Let $G'\eqdef G-P_G$ and assume by way of contradiction that $h_{G'}\ge h_G$; recall that by definition $P_{G'}\subseteq[n]\setminus P_G$ and so $P_{G'}$ and $P_G$ are disjoint.
As by \cref{p-nonempty}, $P_{G'}\ne\emptyset$, we aim to obtain a contradiction by showing that $P_{G'}\subseteq P_G$.

By \cref{p-in-argmax}, $P_G\in\arg\Max_{S\in D_G} E_G(S)$; therefore, we have by definition that
$h_G=E_G(P_G)=\equalize{f_k:k\in P_G}\bigl(\sum_{\type\in\nonemptysubs{P_G}}\mu^{\type}\bigr)$.
By \cref{p-nonempty}, $h_G\in\mathbb{R}$ and so there exist $(\mu_j)_{j\in P_G}\in\Rge^{P_G}$ s.t.\ $\sum_{j\in P_G}\mu_j=\sum_{\type\in\nonemptysubs{P_G}}\mu^{\type}$ and $f_j(\mu_j)=h_G$ for every $j\in P_G$.

Similarly, by \cref{p-in-argmax}, $P_{G'}\in\arg\Max_{S\in D_{G'}} E_{G'}(S)$; therefore, and by definition of $G'$, we have that
$h_{G'}=E_{G'}(P_{G'})=\equalize{f_k:k\in P_{G'}}\bigl(\sum_{\type\in\nonemptysubs{P_G\cup P_{G'}}\setminus\nonemptysubs{P_G}}\mu^{\type}\bigr)$.
By \cref{p-nonempty}, $h_{G'}\in\mathbb{R}$ and so there exist $(\mu'_j)_{j\in P_{G'}}\in\Rge^{P_{G'}}$ s.t.\ $\sum_{j\in P_{G'}}\mu'_j=\sum_{\type\in\nonemptysubs{P_G\cup P_{G'}}\setminus\nonemptysubs{P_G}}\mu^{\type}$ and $f_j(\mu'_j)=h_{G'}\ge h_G$ for every $j\in P_{G'}$. Let $j\in P_{G'}$. As $f_j$ is nondecreasing, by \crefpart{intermediate-equalize}{idempotent} and by definition of $h_G$, we have $f_j(0)\le f_j(\mu^{\{j\}})=E_G(\{j\})\le h_G$. By continuity
of $f_j$ and by the intermediate value theorem, there thus exists $\mu_j\in[0,\mu'_j]$ s.t.\ $f_j(\mu_j)=h_G$.

As $f_j(\mu_j)=h_G$ for every $j\in P_G\cup P_{G'}$, by definition
$\equalize{f_k:k\in P_G\cup P_{G'}}\bigl(\sum_{j\in P_G\cup P_{G'}}\mu_j\bigr)=h_G$.
As $\sum_{j\in P_G\cup P_{G'}}\mu_j=\sum_{j\in P_G}\mu_j+\sum_{j\in P_{G'}}\mu_j\le
\sum_{j\in P_G}\mu_j+\sum_{j\in P_{G'}}\mu'_j=
\sum_{\type\in\nonemptysubs{P_G}}\mu^{\type}+\sum_{\type\in\nonemptysubs{P_G\cup P_{G'}}\setminus\nonemptysubs{P_G}}\mu^{\type}=
\sum_{\type\in\nonemptysubs{P_G\cup P_{G'}}}\mu^{\type}$, we have by \crefpart{equalize-continuous}{suffix} that $E_G(P\cup P_{G'})\in\mathbb{R}$ and therefore, by \cref{equalize-well-defined-nondecreasing}, $E_G(P_G\cup P_{G'})\ge h_G$.
Therefore, by definition of $P_G$, in order to show that $P_{G'}\subseteq P_G$ and complete the proof, it is enough to show that $M_G(P_G\cup P_{G'})=\emptyset$.

Let $S\in\nonemptysubs{P_G\cup P_{G'}}$. By \cref{p-in-argmax}, $M_G(P_G)=\emptyset$ and so, if $S\cap P_G\ne\emptyset$, then there exists $\mu''\le\sum_{\type\in\nonemptysubs{P_G}\setminus\nonemptysubs{P_G\setminus(S\cap P_G)}}\mu^{\type}$ s.t.\ $\equalize{f_k:k\in {S\cap P_G}}(\mu'')=E_G(P_G)=h_G$.
If $S\cap P_G=\emptyset$, then set $\mu''\eqdef0$.

Similarly, by \cref{p-in-argmax}, $M_{G'}(P_{G'})=\emptyset$ and so, if $S\cap P_{G'}\ne\emptyset$, then there exists $\mu'\le\sum_{\type\in\nonemptysubs{P_G\cup P_{G'}}\setminus\nonemptysubs{P_G\cup P_{G'}\setminus(S\cap P_{G'})}}\mu^{\type}$ s.t.\ $\equalize{f_k:k\in {S\cap P_{G'}}}(\mu')=E_{G'}(P_{G'})=h_{G'}$.
As $f_j(\mu_j)=h_G$ for every $j\in P_{G'}$, we also have in this case that $\equalize{f_k:k\in S\cap P_{G'}}\bigl(\sum_{j\in S\cap P_{G'}}\mu_j\bigr)=h_G$.
By both of these and by \cref{equalize-well-defined-nondecreasing}, $\equalize{f_k:k\in S\cap P_{G'}}\bigl(\min\{\mu',\sum_{j\in S\cap P_{G'}}\mu_j\}\bigr)=\min\{h_{G'},h_G\}=h_G$
in this case.
If $S\cap P_{G'}=\emptyset$, then set $\mu'\eqdef0$.

Define $\mu\eqdef\mu''+\min\{\mu',\sum_{j\in S\cap P_{G'}}\mu_j\}$.
By definition of $\mu$ (and by \crefpart{intermediate-equalize}{compose} if neither $S\cap P_G=\emptyset$
nor $S\cap P_{G'}=\emptyset$), we have that $\equalize{f_k:k\in S}(\mu)=h_G$; it is therefore enough to show that $\mu\le\sum_{\type\in\nonemptysubs{P_G\cup P_{G'}}\setminus\nonemptysubs{P_G\cup P_{G'}\setminus S}}\mu^{\type}$ in order to complete the proof. Indeed, since $P_G$ and $P_{G'}$ are disjoint, we obtain that
$\mu=
\mu''+\min\{\mu',\sum_{j\in S\cap P_{G'}}\mu_j\}\le
\mu''+\mu'\le
\sum_{\type\in\nonemptysubs{P_G}\setminus\nonemptysubs{P_G\setminus(S\cap P_G)}}\mu^{\type}+
\sum_{\type\in\nonemptysubs{P_G\cup P_{G'}}\setminus\nonemptysubs{P_G\cup P_{G'}\setminus(S\cap P_{G'})}}\mu^{\type}\le
\sum_{\type\in\nonemptysubs{P_G\cup P_{G'}}\setminus\nonemptysubs{P_G\cup P_{G'}\setminus S}}\mu^{\type}$, as required.
\end{proof}

\subsection{Proof of the Theorems and Corollary from Section~\refintitle{results}}\label{proofs-results}

We defer the proof of \cref{nash-exists} until after that of \cref{super-strong}.

\begin{proof}[Proof of \cref{indifference-resource-costs}]
We prove the \lcnamecref{indifference-resource-costs} by full induction on $n$.

Let $n\in\mathbb{N}$ and assume that the \lcnamecref{indifference-resource-costs} holds for all smaller natural values of $n$. Let $G=\game$ be an $n$-resource selection game
and let $s,s'$ be Nash equilibria in~$G$.
By \cref{highest-stopping}(\labelcref{highest-stopping-who},\labelcref{highest-stopping-where}),
$\height{j}=h_G=\heightt{j}$, for every $j\in P_G$.
If $P_G=[n]$, then the proof is complete. Otherwise,
let $s'',s''':\nonemptysubs{[n]\setminus P_G}\rightarrow\Rge^{[n]\setminus P_G}$ be the functions defined by
$s''_j(\type')\eqdef\sum_{\type\in \mathcal{R}(\type',G-P_G)}s_j(\type)$ and $s'''_j(\type')\eqdef\sum_{\type\in \mathcal{R}(\type',G-P_G)}s'_j(\type)$ for every $j\in[n]\setminus P_G$.
By \cref{highest-stopping}(\labelcref{highest-stopping-who},\labelcref{highest-stopping-rest}), $s'',s'''$ are both Nash equilibria in
$G-P_G$,
and so, by the induction hypothesis (since $P_G\ne\emptyset$ by \cref{p-nonempty}), we obtain that $\heighttt{j}=\heightttt{j}$ for every $j\in[n]\setminus P_G$. Therefore, by \cref{highest-stopping}(\labelcref{highest-stopping-who},\labelcref{highest-stopping-rest}), we have
$\height{j}=\heighttt{j}=\heightttt{j}=\heightt{j}$ for every $j\in[n]\setminus P_G$ as well, and so $\height{j}=\heightt{j}$ for every $j\in[n]$, as required.
\end{proof}

\begin{proof}[Proof of \cref{indifference}]
We start by proving \cref{indifference-players}. Let $s,s'$ be Nash equilibria in $G$, and
let $\type\in\types$. By definition of Nash equilibrium and by \cref{indifference-resource-costs},
we have for every $k\in\supp\bigl(s(\type)\bigr)$ and $k'\in\supp\bigl(s'(\type)\bigr)$ that $\height{k}=\min_{j\in\type}\height{j}=\min_{j\in\type}\heightt{j}=\heightt{k'}$, as required.

We move on to prove \cref{indifference-resources}; Let $j\in[n]$. Let $S=\bigl\{k\in[n]~\big|~\height{k}=\height{j}\}$; by \cref{indifference-resource-costs},
$S=\bigl\{k\in[n]~\big|~\heightt{k}=\heightt{j}\}$ as well. Let $\mathcal{R}\eqdef\bigl\{\type\in\types~\big|~\supp\bigl(s(\type)\bigr)\subseteq S\bigr\}$; by \cref{indifference-resource-costs,indifference-players}, $\mathcal{R}=\bigl\{\type\in\types~\big|~\supp\bigl(s'(\type)\bigr)\subseteq S\bigr\}$ as well. Assume w.l.o.g.\ that $\height{j}$ is not a plateau height of any of $S\setminus\{j\}$; we therefore have to show that $\load{k}=\loadt{k}$ for every $k\in S$. For every $k\in S\setminus\{j\}$, as $\height{j}$ is not a plateau height of $f_k$, there exists a unique value $\mu_k\in\Rge$ s.t.\ $f_k(\mu_k)=\height{j}$. Therefore, and as by definition of $S$ and by \cref{indifference-resource-costs} we have that $f_k(\load{k})=\height{k}=\height{j}=\heightt{j}=\heightt{k}=f_k(\loadt{k})$ for every $k\in S\setminus\{j\}$, we have that
$\load{k}=\mu_k=\loadt{k}$ for every $k\in S\setminus\{j\}$. By \cref{indifference-players}, we have that $\sum_{k\in S}\load{k}=\sum_{\type\in\mathcal{R}}\mu^{\type}=\sum_{k\in S}\loadt{k}$, and so $\load{j}=\sum_{\type\in\mathcal{R}}\mu^{\type}-\sum_{k\in S\setminus\{j\}}\load{k}=\sum_{\type\in\mathcal{R}}\mu^{\type}-\sum_{k\in S\setminus\{j\}}\loadt{k}=\loadt{j}$. Therefore, $\load{k}=\loadt{k}$ for every $k\in S$ and the proof is complete.
\end{proof}

\begin{proof}[Proof of \cref{super-strong}]
We begin by proving \cref{super-strong-strong} by full induction on $n$.
Let $n\in\mathbb{N}$ and assume that the claim holds for all smaller natural values of $n$. Let $G=\game$ be an $n$-resource selection game
and let $s$ be Nash equilibrium in $G$. For every $\type\in\types$ with $\mu^{\type}\ne 0$, let $h^{\type}\eqdef\height{j}$ for
every $j\in\supp\bigl(s(\type)\bigr)$.
Let $s'$ be a consumption profile s.t.\
$\heightt{k}< h^{\type}$ for every $\type\in\types$ and $k\in\supp\bigl(s'(\type)\bigr)$ s.t.\ $s'_k(\type)>s_k(\type)$. We must show that $s'=s$.

We begin by showing that $s'(\type)=s(\type)$ for every $\type\in\nonemptysubs{P^s}$, for $P^s$ as defined in \cref{highest-stopping}.
Assume by way of contradiction that $s'(\type)\ne s(\type)$ for some $\type\in\nonemptysubs{P^s}$. Let $S=\{j\in P^s\mid \heightt{j}<h_G\}\subseteq P^s$.
As $s'(\type)\ne s(\type)$, there exists $k\in\type$ s.t.\ $s'_k(\type)>s_k(\type)$ and so $k\in\supp\bigl(s'(\type)\bigr)$. Therefore, by definition of $s'$ and by \crefpart{highest-stopping}{where}, we have that $\heightt{k}<\height{k}=h_G$ and so $k\in S$; in particular, $S\ne\emptyset$.

For every $j\in S$, by definition of $S$ and by \crefpart{highest-stopping}{where}, we have that $f_j(\loadt{j})=\heightt{j}<h_G=\height{j}=f_j(\load{j})$; therefore,
as $f_j$ is nondecreasing we have that $\loadt{j}<\load{j}$ for every such $j$. Therefore, as $S\ne\emptyset$,
$\sum_{j\in S}\loadt{j}<\sum_{j\in S}\load{j}$.
By definition of consumption profile and by \crefpart{highest-stopping}{on}, $\sum_{j\in P^s}\loadt{j}\ge\sum_{\type'\in\nonemptysubs{P^s}}\mu^{\type'}=\sum_{j\in P^s}\load{j}$.
By both of these, $\sum_{j\in P^s\setminus S}\loadt{j}>\sum_{j\in P^s\setminus S}\load{j}$, and so there exists $j\in P^s\setminus S$ s.t. $\loadt{j}>\load{j}$;
hence, there exists $\type'\in\types$ s.t.\ $s'_j(\type')>s_j(\type')$
(and so $j\in\supp\bigl(s'(\type')\bigr)$),
however, by definition of $S$ and as $f_j$ is nondecreasing, we have that
\begin{equation}\label{super-strong-strong-ge}
\heightt{j}=f_j(\loadt{j})\ge f_j(\load{j})=\height{j}=h_G \ge h^{\type'}
\end{equation} (as $s'_j(\type')>0$, $h^{\type'}$ is well defined), even though
$s'_j(\type')>s_j(\type')$ --- a contradiction. Therefore, $s'(\type)=s(\type)$ for every $\type\in\nonemptysubs{P^s}$.
By definition of $s'$, by definition of $P^s$
and by \crefpart{highest-stopping}{on}, we thus obtain that $s'_j\equiv s_j$ for every $j\in P^s$.

If $P^s=[n]$, then the proof is complete. Otherwise, define $s'':\nonemptysubs{[n]\setminus P^s}\rightarrow\Rge^{[n]\setminus P^s}$ by $s''_j(\type')\eqdef\sum_{\type\in \mathcal{R}(\type',G-P^s)}s_j(\type)$ for every $j\in[n]\setminus P^s$.
By \crefpart{highest-stopping}{rest}, $s''$ is a Nash equilibrium in
$G-P^s$, and $\heighttt{j}=\height{j}$ for every $j\in[n]\setminus P^s$.
For every $\type'\in\nonemptysubs{[n]\setminus P^s}$ with $\mu^{\type'}\ne 0$, let $h'^{\type'}\eqdef\heighttt{j}$ for every $j\in\supp\bigl(s''(\type)\bigr)$;
by definition, $h'^{\type'}=h^{\type}$ for every $\type'\in\nonemptysubs{[n]\setminus P^s}$ and $\type\in \mathcal{R}(\type',G-P^s)$ s.t.\ $\mu^{\type}\ne0$.

Similarly, define $s''':\nonemptysubs{[n]\setminus P^s}\rightarrow\Rge^{[n]\setminus P^s}$, by $s'''_j(\type')\eqdef\sum_{\type\in \mathcal{R}(\type',G-P^s)}s'_j(\type)$ for every $j\in[n]\setminus P^s$.
As $s'_j\equiv s_j$ for every $j\in P^s$, we have that, similarly to the proof of \crefpart{highest-stopping}{rest},
$s'''$ is a strategy profile in $G-P^s$ and $\heightttt{j}=\heightt{j}$ for every $j\in[n]\setminus P^s$.
By definition, we have that
$\heightttt{k}=\heightt{k}<h^{\type}=h'^{\type}$ for every $\type'\in\nonemptysubs{[n]\setminus P^s}$ and $k\in\supp\bigl(s'''(\type')\bigr)$ s.t.\ $s'''_k(\type')>s''_k(\type')$, where $\type\in \mathcal{R}(\type',G-P^s)$
s.t.\ $k\in\supp\bigl(s'(\type)\bigr)$ and $s'_k(\type)>s_k(\type)$ (there exists such $\type$ by definition of $\type'$). By the induction hypothesis (since $P^s\ne\emptyset$ by definition), $s'''=s''$,
and so $s'=s$ and the proof of \cref{super-strong-strong} is complete.

The proof of \cref{super-strong-super-strong} is very similar;
the main difference is that in \cref{super-strong-strong-ge} we would have, by $\height{j}$ not being a plateau height
of $f_j$,
that $\heightt{j}=f_j(\loadt{j})> f_j(\load{j})=\height{j}=h_G \ge h^{\type'}$. The remaining trivial differences between \cref{super-strong-strong,super-strong-super-strong} are left to the reader.
\end{proof}

\begin{proof}[Proof of \cref{nash-exists}]
We prove by full induction on $n$ that in every $n$-resource selection game $G=\game$ s.t.\ $f_1,\ldots,f_n$ are continuous, there exists a Nash equilibrium $s$ s.t.\ $\Max_{j\in [n]}\height{j}\le h_G$;
from this claim, the existence of Nash equilibrium \emph{a fortiori} follows (while the existence of Nash equilibrium also follows from a theorem of \cite{Schmeidler73}, we constructively reprove it here via hydraulic analysis rather than a nonconstructive fixed-point theorem). The \lcnamecref{nash-exists} then follows by \cref{super-strong}.
Let $n\in\mathbb{N}$ and assume that the claim holds for all smaller natural values of $n$. Let $G=\game$ be an $n$-resource selection game.

By \crefpart{construct}{distribute}, there exists a strategy profile $s''$ in the $|P_G|$-resource selection game $G''\eqdef\bigl((f_j)_{j\in P_G};(\mu^{\type})_{\type\in\nonemptysubs{P_G}}\bigr)$
s.t.\ $\heighttt{j}=h_G$ for every $j\in P_G$. By definition of Nash equilibrium, $s''$ is a Nash equilibrium in $G''$.
If $P_G=[n]$, then $s\eqdef s''$ is a Nash equilibrium as required, and the proof of the induction step is complete.
Assume, therefore, that $P_G\subsetneq[n]$; hence, and since $P_G\ne\emptyset$ by \cref{p-nonempty}, by the induction hypothesis there exists a Nash equilibrium $s'$ in the $\bigl|[n]\setminus P_G\bigr|$-resource selection game $G'\eqdef G-P_G$,
s.t.\ $\Max_{j\in [n]\setminus P_G}\height{j}\le h_{G'}$.

We construct a strategy profile $s$ in $G$ as follows:
$s(\type)\eqdef s''(\type)$ for every $\type\in\nonemptysubs{P_G}$, and for every $\type'\in\nonemptysubs{[n]\setminus P_G}$, we pick $\bigl(s(\type)\bigr)_{\type\in \mathcal{R}(\type',G')}$
arbitrarily among the tuples satisfying $s(\type)\in\mu^{\type}\cdot\Delta^{\type'}$ for every $\type\in \mathcal{R}(\type',G')$ and $\sum_{\type\in \mathcal{R}(\type',G')}s(\type')=s'(\type')$. This is a well-defined strategy profile in $G$ since $\type'=\type\setminus P_G\subseteq\type$ for every $\type\in \mathcal{R}(\type',G')$ and $\type'\in\nonemptysubs{[n]\setminus P_G}$, and by definition of the player mass in $G'$ and $G''$.
By definition of $s$, we have that $\height{j}=\heighttt{j}$ for every $j\in P_G$ and $\height{j}=\heightt{j}$ for every $j\in[n]\setminus P_G$.
Therefore, by definition of $s''$ we have that $\height{j}=\heighttt{j}=h_G$ for every $j\in P_G$, and by definition of $s'$ and by \crefpart{construct}{stopping-order}, we have that
$\height{j}=\heightt{j}\le h_{G'}<h_G$ for every $j\in[n]\setminus P_G$. Therefore, we have that $\height{j}\le h_G$ for every $j\in[n]$.

We complete the proof by showing that $s$ is a Nash equilibrium in $G$.
For every $\type\in\nonemptysubs{P_G}$, $k\in\supp\bigl(s(R)\bigr)\subseteq\type$ and $j\in\type$, we have by definition of $s,s''$ that $\height{k}=\heighttt{k}=h_G=\heighttt{j}=\height{j}$. Let $\type\in\types\setminus\nonemptysubs{P_G}$, $k\in\supp\bigl(s(R)\bigr)$ and $j\in\type$.
By definition of $s$, we have that $k\in\supp\bigl(s'(R\setminus P_G)\bigr)\subseteq\nonemptysubs{[n]\setminus P_G}$.
If $j\in[n]\setminus P_G$, then $j\in\type\setminus P_G$ and by definition of $s,s'$ we have that $\height{k}=\heightt{k}\le\heightt{j}=\height{j}$; otherwise,
i.e., if $j\in P_G$, then by \crefpart{construct}{stopping-order} and
by definition of $s,s',s''$ we have that $\height{k}=\heightt{k}\le h'_G<h_G=\heighttt{j}=\height{j}$. Either way, $\height{k}\le\height{j}$ and the proof is complete.
\end{proof}

\subsection{Proof of Theorem~\refintitle{hall-frac} from Section~\refintitle{hall-and-beyond}}\label{proofs-hall}

\begin{proof}[Proof of \cref{hall-frac}]
As in the main text, we prove only one direction, leaving the proof of the other (trivial) direction to the reader. Assume that no fractional perfect marriage exists.

By \cref{nash-exists}, there exists a (strong) Nash equilibrium $s$ in $G$. As no fractional perfect marriage exists, by \cref{hall-char} we have that not all loads in $s$ are $1$. As by definition of~$G$ the average of all loads in $s$ is $1$, we have that the highest load in $s$ is greater than $1$. By \crefpart{highest-stopping}{where}, we therefore have $h_G>1$.
Therefore, by definition there exists a set of pistons $S\in\types$ s.t.\ $E_G(S)>1$. As $f_j=\id$ for all $j\in S$, we have, as in \cref{symmetric-equalize} and by definition of~$\mu^{\type}$, that $1<E_G(S)=\equalize{f_j:j\in S}\Bigl(\sum_{\type\in\nonemptysubs{S}}\mu^{\type}\Bigr)=\frac{\sum_{\type\in\nonemptysubs{S}}\mu^{\type}}{|S|}=\frac{|\{i\in[n]\mid\type^i\in\nonemptysubs{S}\}|}{|S|}$.
As by definition, $\type^{\{i\in[n]\mid\type^i\in\nonemptysubs{S}\}}\subseteq S$, we have that $\Bigl|\bigl\{i\in[n]~\big|~\type^i\in\nonemptysubs{S}\bigl\}\Bigl|>|S|\ge\Bigl|\type^{\{i\in[n]\mid\type^i\in\nonemptysubs{S}\}}\Bigl|$, i.e.,
that $|I|>|\type^I|$ for $I\eqdef\bigl\{i\in[n]~\big|~\type^i\in\nonemptysubs{S}\bigr\}$, as required.
\end{proof}

\subsection{Proof of Proposition~\refintitle{cost-lipschitz} from Section~\refintitle{discussion}, and Auxiliary Results}\label{proofs-discussion}

\begin{definition}
Let $G=\game$ be a resource selection game s.t.\ $f_1,\ldots,f_n$ are continuous. We denote by $h_j(G)\in\mathbb{R}$ the value $\height{j}$ for every Nash equilibrium $s$ in $G$.
\end{definition}

\begin{remark}
$h_j(G)$ is well defined by \cref{nash-exists,indifference-resource-costs}.
\end{remark}

\begin{lemma}\label{remove-rises}
Let $G=\game$ be a resource selection game s.t.\ $f_1,\ldots,f_n$ are continuous, and let $S\subseteq[n]$.
For every $j\in[n]\setminus S$,\ \ $h_j(G-S)\ge h_j(G)$.
\end{lemma}

\begin{remark}\label{remove-highest-same}
By \cref{highest-stopping}(\labelcref{highest-stopping-who},\labelcref{highest-stopping-rest}), taking $S\eqdef P_G$ in \cref{remove-rises} yields an equality.
\end{remark}

\begin{proof}[Proof of \cref{remove-rises}]
We prove the \lcnamecref{remove-rises} by full induction on $n$.

(Outer induction) Step: Let $n\in\mathbb{N}$ and assume that the \lcnamecref{remove-rises} holds for all smaller values of $n$. We prove the (outer) induction step by full induction on $n-|S|$.

(Inner induction) Base: If $S=[n]$, then the claim vacuously holds.

(Inner induction) Step: Let $S\subsetneq[n]$ and assume that the (outer induction) step holds for all $S$ of larger cardinality.
We consider two cases.

If $S\supseteq P_G$, then by by \crefpart{associative-remove}{associative}, by the (outer) induction hypothesis (since $P_G\ne\emptyset$ by \cref{p-nonempty}), and by \cref{remove-highest-same}, we have $h_j(G-S) = h_j\bigl((G-P_G)-(S\setminus P_G)\bigr) \ge h_j(G-P_G) = h_j(G)$, as required.

Otherwise, i.e., if $P'\eqdef P_G\setminus S\ne\emptyset$, we claim that $h_{G-S}\ge h_G$.
By \cref{p-in-argmax}, we have that $P_G\in\arg\Max_{S'\in D_G}E_G(S')$; in particular, $P_G\in D_G$. Since $E_G(P_G)=h_G\in\mathbb{R}$ by \cref{p-in-argmax,p-nonempty}, we therefore have that $M(P_G)=\emptyset$. Therefore, $P'\notin M(P_G)$, and so there exists $\mu\le\sum_{\smash{\type\in\nonemptysubs{P_G}\setminus\nonemptysubs{P_G\setminus P'}}}\mu^{\type}$ s.t.\ $\equalize{f_k:k\in P'}(\mu)=E_G(P_G)$.
We note that $\nonemptysubs{P_G}\setminus\nonemptysubs{P_G\setminus P'}\subseteq
\nonemptysubs{S\cup P'}\setminus\nonemptysubs{(S\cup P')\setminus P'}=
\bigcup_{\smash{\type'\in\nonemptysubs{P'}}}\mathcal{R}(\type',G-S)$, where the union is of disjoint sets;
therefore, $\mu\le
\sum_{\smash{\type\in\nonemptysubs{P_G}\setminus\nonemptysubs{P_G\setminus P'}}}\mu^{\type}\le
\sum_{\smash{\type'\in\nonemptysubs{P'}}}\sum_{\type\in\mathcal{R}(\type',G-S)}\mu^{\type}$.
By \namecrefs{equalize-well-defined-nondecreasing}~\labelcref{equalize-well-defined-nondecreasing} and~\labelcref{equalize-continuous}(\labelcref{equalize-continuous-suffix}), we therefore have that $h_G=E_G(P_G)=
\equalize{f_k:k\in P'}(\mu)\le
\equalize{f_k:k\in P'}\bigl(\sum_{\smash{\type'\in\nonemptysubs{P'}}}\sum_{\type\in\mathcal{R}(\type',G-S)}\mu^{\type}\bigr)=
E_{G-S}(P')\le h_{G-S}$, where the last inequality is by \crefpart{max-on-all}{eq}.

For every $j\in P_{G-S}$, By \cref{highest-stopping}(\labelcref{highest-stopping-who},\labelcref{highest-stopping-where}), we have $h_j(G-S)=h_{G-S}\ge h_G \ge h_j(G)$. For every $j\in [n]\setminus (S\cup P_{G-S})$, by \cref{remove-highest-same}, by \crefpart{associative-remove}{associative} and by the (inner) induction hypothesis (since $P_{G-S}\ne\emptyset$ by \cref{p-nonempty}), we have $h_j(G-S)=h_j\bigl((G-S)-P_{G-S}\bigr)=h_j\bigl(G-(S\cup P_{G-S})\bigr)\ge h_j(G)$ and the proof is complete.
\end{proof}

\begin{proof}[Proof of \cref{cost-lipschitz}]
We prove that $h_1,\ldots,h_n$ are nondecreasing by full induction on $n$. Let $n\in\mathbb{N}$ and assume that the claim holds for all smaller values of $n$. Let $G=\game$ be an $n$-resource selection game s.t.\ $f_1,\ldots,f_n$ are continuous, let $\type\in\types$, let $\mu'^{\type}>\mu^{\type}$ and let $G'$ be the game obtained from $G$ by increasing the mass of player type $\type$ from $\mu^{\type}$ to $\mu'^{\type}$.

By \crefpart{max-on-all}{eq} and by \cref{equalize-well-defined-nondecreasing}, $h_{G'}=\Max_{\smash{S\in\types}} E_{G'}(S)\ge\Max_{\smash{S\in\types}} E_G=h_G$.
Therefore, by \cref{highest-stopping}(\labelcref{highest-stopping-who},\labelcref{highest-stopping-where}), $h_j(G')=h_{G'}\ge h_G=h_j(G)$ for every $j\in P_{G'}$;
it therefore remains to show that $h_j(G')\ge h_j(G)$ for every $j\in [n]\setminus P_{G'}$ as well.
Before we show this, we claim that $h_j(G'-P_{G'})\ge h_j(G-P_{G'})$ for every $j\in[n]\setminus P_{G'}$; to show this, we consider two cases.
If $\type\subseteq P_{G'}$, then by definition $G'-P_{G'}=G-P_{G'}$, and so $h_j(G'-P_{G'})=h_j(G-P_{G'})$ for every $j\in[n]\setminus P_{G'}$. Otherwise, i.e., if $\type\setminus P_{G'}\ne\emptyset$, then by definition $G'-P_{G'}$ is the game obtained from $G-P_{G'}$ by increasing the mass of player type $\type\setminus P_{G'}$ by $\mu'^{\type}-\mu^{\type}>0$. Therefore, by the induction hypothesis (since $P_{G'}\ne\emptyset$ by \cref{p-nonempty}), we therefore have that $h_j(G'-P_{G'})\ge h_j(G-P_{G'})$ for every $j\in [n]\setminus P_{G'}$ in this case as well.
Finally, by \cref{remove-highest-same,remove-rises}, we therefore have for every $j\in[n]\setminus P_{G'}$ that $h_j(G')=h_j(G'-P_{G'})\ge h_j(G-P_{G'})\ge h_j(G)$, and the proof by induction is complete.

We move on to prove continuity of $h_1,\ldots,h_n$; for simplicity, we show continuity only for the case in which $f_1,\ldots,f_n$ are strictly increasing.
W.l.o.g.\ we will show that $h_1$ is continuous; let $\varepsilon>0$.
For every $j\in[n]$, let $\mu_j(G)=f_j^{-1}\bigl(h_j(G)\bigr)$ --- this is the value $\load{j}$ for every Nash equilibrium $s$ in $G$. By continuity of $f_1$, there exists $\delta>0$ s.t.\ $\bigl|f_1(\mu)-f_1\bigl(\mu_1(G)\bigr)\bigr|<\varepsilon$
for every $\mu\in\bigl(\mu_1(G)-\delta,\mu_1(G)+\delta\bigr)$. Let $\mu'^{\type}\in(\mu^{\type}-\delta,\mu^{\type}+\delta)$ and denote by $G'$ the game obtained from $G$ by changing the mass of player type $\type$ from $\mu^{\type}$ to $\mu'^{\type}$. For every $j\in[n]$, let $\mu_j(G')=f_j^{-1}\bigl(h_j(G')\bigr)$.
We first consider the case in which $\mu'^{\type}\ge\mu^{\type}$. In this case, as shown above, for every $j\in[n]$ we have that $h_j(G')\ge h_j(G)$; as $f_j$ is increasing, therefore $\mu_j(G')\ge \mu_j(G)$ for every such $j$. As $\sum_{j\in[n]} \mu_j(G') = \sum_{\smash{\type'\in\types\setminus\{\type\}}} \mu^{\type'} + \mu'^{\type} =
\sum_{j\in[n]} \mu_j(G) + \mu'^{\type}-\mu^{\type} < \sum_{j\in[n]} \mu_j(G) + \delta$, we therefore have that $\mu_j(G) \le \mu_j(G')< \mu_j(G)+\delta$ for every $j$. In particular, $\mu_1(G) \le \mu_1(G') < \mu_1(G) + \delta$, and so $0\le h_1(G')-h_1(G)=f_1(\mu_1(G'))-f_1(\mu_1(G))<\varepsilon$, as required. The case in which $\mu'^{\type}<\mu^{\type}$ is analogous.

The proof of \cref{cost-lipschitz-lipschitz} is virtually identical to that of the continuity of $h_1,\ldots,h_n$, noticing that we can choose $\delta=\frac{\varepsilon}{K}$, where $K$ is the Lipschitz constant of $f_1$.
\end{proof}

\end{document}